\documentclass[a4paper]{article}

\usepackage{a4wide}
\usepackage[UKenglish]{babel}
\usepackage{graphicx,subfigure}
	\graphicspath{{./}{figure/}}
\usepackage{epstopdf}
\usepackage[colorlinks=true,linkcolor=blue,citecolor=red]{hyperref}
\usepackage{amssymb,amsthm,mathtools,mathrsfs,bbold,cool}
\usepackage[inline]{enumitem}
\usepackage{soul}

\newtheorem{theorem}{Theorem}[section]

\theoremstyle{remark}\newtheorem{remark}[theorem]{Remark}

\DeclarePairedDelimiter\abs{\lvert}{\rvert}
\DeclarePairedDelimiter\ave{\langle}{\rangle}
\DeclarePairedDelimiter\norm{\lVert}{\rVert_\infty}
\newcommand{\B}{\operatorname{B}}
\newcommand{\cB}{\mathscr{B}}
\renewcommand{\E}{\mathbb{E}}
\newcommand{\bE}{\mathbf{E}}
\newcommand{\cF}{\mathscr{F}}
\newcommand{\vecf}{\mathbf{f}}
\newcommand{\cL}{\mathscr{L}}
\newcommand{\N}{\mathbb{N}}
\renewcommand{\P}{\operatorname{Prob}}
\newcommand{\R}{\mathbb{R}}
\newcommand{\supp}{\operatorname{supp}}
\newcommand{\cU}{\mathscr{U}}
\newcommand{\Var}{\operatorname{Var}}

\allowdisplaybreaks

\begin{document}
\title{Uncertainty damping in kinetic traffic models \\ by driver-assist controls}

\author{Andrea Tosin\thanks{Department of Mathematical Sciences ``G. L. Lagrange'', Politecnico di Torino, Italy},
		Mattia Zanella\thanks{Department of Mathematics ``F. Casorati'', University of Pavia, Italy}}
\date{}

\maketitle

\begin{abstract}
In this paper, we propose a kinetic model of traffic flow with uncertain binary interactions, which explains the scattering of the fundamental diagram in terms of the macroscopic variability of aggregate quantities, such as the mean speed and the flux of the vehicles, produced by the microscopic uncertainty. Moreover, we design control strategies at the level of the microscopic interactions among the vehicles, by which we prove that it is possible to dampen the propagation of such an uncertainty across the scales. Our analytical and numerical results suggest that the aggregate traffic flow may be made more ordered, hence predictable, by implementing such control protocols in driver-assist vehicles. Remarkably, they also provide a precise relationship between a measure of the macroscopic damping of the uncertainty and the penetration rate of the driver-assist technology in the traffic stream.

\medskip


\end{abstract}

\section{Introduction}
In recent times, the challenge of vehicle automation has inspired new paradigms for the management and the governance of traffic and has imposed a deep reflection about the behavioural aspects involved in vehicle dynamics, see e.g.~\cite{Aria2016,Jamson2013}. On one hand, the emerging technologies constitute a potentially high innovation in the realm of the smart cities, for example with clear benefits in terms of emission reduction and of mitigation of road risk factors. On the other hand, the time horizon of their actual full implementation is still uncertain and questionable. For this reason, recent advances in the modelling of driver-assist/autonomus vehicles have focused, in particular, on the assessment of the effectiveness of \textit{Advanced Driver-Assistance Systems} (ADAS) in a mixed scenario, with few automated vehicles embedded in a stream of human-manned vehicles~\cite{tosin2019MMS}. In this context, it is particularly interesting to address the design of suitable \textit{control protocols}, which are sufficiently robust to cope with some behavioural aspects of human drivers, and then to study their aggregate impact on the traffic stream, see e.g.~\cite{stern2018TRC}. 

In this work, we study a kinetic traffic model with \textit{structural uncertainty} in the vehicle-to-vehicle interactions. Such an uncertainty, represented by a parameter whose value is not known deterministically, accounts, for example, for the heterogeneous response of different types of vehicles to the speed variations. Kinetic models for traffic flow have a quite long history, going back to the pioneering works \cite{paveri1975TR,prigogine1971BOOK} up to more recent advances, such as e.g.~\cite{Herty2010,Klar1997}. The kinetic approach has the clear advantage of linking, in a self-consistent way, the different scales of the problem: from the microscopic one of the interactions among the vehicles to the mesoscopic one of the aggregate distribution of observable quantities, such as the fundamental diagrams, and further also to the macroscopic one of the density waves flowing along a road~\cite{tosin2019MMS}. In~\cite{puppo2016CMS,puppo2017CMS} a kinetic approach has been developed to explain the scattering of the fundamental diagram of traffic. The latter is an experimentally measured relation between the traffic congestion and the vehicle flux, which typically features consistent macroscopic fluctuations at high vehicle density, see e.g.~\cite{kerner2004BOOK}. The mentioned works are based on multi-population/mixture models, which take into account the heterogeneous composition of the traffic stream by means of several kinetic equations, one for each class of vehicles. It is worth mentioning that also different approaches have been considered in the literature to model heterogeneous traffic flows. We recall, for instance, the car-following microscopic models~\cite{mason1997PRE,nagatani2000PHYSA} and the systems of conservation laws with a flux function depending simultaneously on all the class densities~\cite{benzoni2003EJAM,colombo2018CHAPTER}. In the present work, we show that a kinetic explanation of the scattering of the fundamental diagram in the congested traffic regime can be effectively obtained also by introducing, in the microscopic interactions of a \textit{single-population} model, an uncertain parameter accounting for the different reaction strengths of the various classes of vehicles. On one hand, this method leads to an increased dimensionality of the kinetic problem, owing to the necessity to handle the uncertain parameter. On the other hand, however, it avoids the necessity to deal with systems of kinetic equations, which would require to define several microscopic interaction rules and mesoscopic collision operators for the various classes of vehicles, thereby making the whole approach more amenable to additional modelling and analytical developments.

Along this line, in this work we exploit the potential offered by the uncertain kinetic setting to further address a control problem, aimed at reducing the scattering of the fundamental diagram. Specifically, we design two possible control protocols, which, consistently with the aforementioned ADAS technologies, are applied at the level of the microscopic interactions among the vehicles and, in particular, act in such a way to align the speed of the vehicles to a prescribed congestion-dependent value. Since, in a realistic scenario, all vehicles are not equipped with the ADAS technology, such controls are active only on a possibly small percentage of vehicles, corresponding to the so-called \textit{penetration rate}. Our main result is that such \textit{microscopic} control strategies are able to dampen the structural uncertainty present in the vehicle-to-vehicle interactions, thereby reducing the \textit{macroscopic} scattering of the fundamental diagram in a way precisely linkable to the penetration rate of the ADAS technology. An immediate consequence of this is a more ordered, hence predictable, macroscopic flow of the vehicles, with clear advantages for the traffic governance at large scale. It is worth anticipating that this theoretical result holds for \textit{any} possible statistical distribution of the uncertain parameter. Sticking to a purely model-based approach, in this paper we assume that such a distribution is given. However, we mention that, in the realm of Uncertainty Quantification (UQ), several tools from Bayesian inference have been developed to estimate the distribution of uncertain parameters in physical models, see e.g.~\cite{marzouk2009CCP}. As far as the traffic dynamics considered here are concerned, in the recent work~\cite{herty2019PREPRINT} the statistical distribution of the uncertain parameter entering our interaction rules has been estimated from raw traffic data so as to fit the observable speed distributions.

From the methodological point of view, the proposed control setting takes advantage of recently introduced methods for the \textit{binary control} of Boltzmann-type kinetic equations, see~\cite{Albi2015,Albi2014} and also~\cite{tosin2018IFAC,tosin2019MMS} for the specific application to vehicular traffic. First, we derive an uncertain control, namely one obtained from a microscopic cost functional depending pointwise on the uncertain parameter of the interactions. As a second case, we derive a deterministic control, namely one obtained from a microscopic cost functional averaged with respect to the uncertain parameter. These two microscopic controls lead to two different Boltzmann-type models. The one associated with the uncertain control is a Maxwellian model, since the corresponding Boltzmann equation features a unitary interaction kernel. Conversely, the one associated with the deterministic control is a kinetic model \textit{with cut-off}, because the corresponding Boltzmann equation needs a non-constant interaction kernel which discards some microscopic interactions possibly violating the physical bounds of the microscopic variable. We discuss the analytical properties of both models and, in particular, we show that, in the asymptotic regime of the \textit{quasi-invariant interactions}~\cite{cordier2005JSP,toscani2006CMS}, they can be described by the same Fokker-Planck equation, which leads to explicitly computable and realistic steady states highlighting the reduction of the scattering of the fundamental diagram discussed above. We stress that the adopted kinetic approach makes possible a \textit{multiscale} investigation of the propagation of both the structural uncertainty and the effect of the control across the different scales of the problem. We also propose several numerical tests, which, by appealing to specific numerical methods for Uncertainty Quantification for kinetic and mean field equations, cf. e.g.~\cite{carrillo2019VJM,Dimarco2017,Hu2017,Jin2017,zanella2020MCS,Zhu2017}, visualise the theoretical results and allow us to consider various probability distributions of the uncertain parameter, spanning also interesting cases in which explicit computations are not possible.

In more detail, the paper is structured as follows: in Section~\ref{sect:micro}, we discuss a basic microscopic model of the interactions among the vehicles without control and we stress, in particular, the role of the uncertain parameter. Then we pass to a Boltzmann-type kinetic description, whence we obtain explicitly the asymptotic trend of the mean speed, which we use to define the fundamental diagram together with its scattering induced by the uncertain parameter. Finally, in the quasi-invariant interaction limit, we recover a Fokker-Planck-type description, whence we deduce an explicit form of the equilibrium speed distribution along with its uncertainty. In Section~\ref{sect:uncertainty_damping}, we introduce the two microscopic control strategies to be applied to the interactions among the vehicles. Specifically, in Section~\ref{sect:pointwise}, we discuss the uncertain control, which leads to a Boltzmann-type kinetic description for Maxwellian-like particles and allows us to prove a precise result concerning the reduction of the scattering of the fundamental diagram in terms of the penetration rate of the control strategy. In Section~\ref{sect:average}, we discuss instead the deterministic control, which is more realistically implementable in practice but requires to deal with a more difficult Boltzmann-type kinetic description with cut-off. However, we prove that, in the quasi-invariant interaction regime, such a description yields the same limit equations for both the mean speed and the statistical speed distribution as the previous description, whence we recover the validity of all the results proved before. In Section~\ref{sect:numerics}, we present several numerical tests supporting the theoretical findings of the previous sections. Finally, in Section~\ref{sect:conclusions} we summarise the results of the paper and we briefly outline possible research developments.

\section{Scattering of the fundamental diagram from uncertain binary interactions}
\label{sect:micro}
\subsection{Description of the microscopic interactions with uncertainty}
In accordance with the general approach of the kinetic theory, a kinetic description of traffic flow is based on the identification of microscopic interaction rules for pairs of vehicles. We assume, in particular, that these interactions modify the speed of the vehicles in consequence of accelerations and decelerations. Therefore, we characterise the microscopic state of a generic vehicle by means of a variable $v\in [0,\,1]$ representing its (dimensionless) speed. Denoting by $v_\ast\in [0,\,1]$ the speed of the leading vehicle, we describe the speed variation in a binary interaction as
\begin{align}
	\begin{aligned}[c]
		v' &= v+\gamma I(v,\,v_\ast;\,z)+D(v)\eta \\
		v_\ast' &= v_\ast.
	\end{aligned}
	\label{eq:binary.gen}
\end{align}
In~\eqref{eq:binary.gen}, $\gamma>0$ is a proportionality parameter and $I$ is the \textit{interaction function}, which describes the acceleration/deceleration dynamics mentioned above. This function depends on the pre-interaction speeds $v$, $v_\ast$ of the interacting vehicles and on an \textit{uncertain parameter}, viz. random variable, $z\in\R$ with known probability distribution $\Psi=\Psi(z):\R\to\R_+$, i.e.
$$ \P(z\leq\bar{z})=\int_{-\infty}^{\bar{z}}\Psi(z)\,dz. $$
Such an uncertain parameter represents a \textit{structural uncertainty} in the binary rule modelled by $I$, due to the fact that the physics of the interactions among the vehicles is inevitably partly heuristic.

Getting inspiration from~\cite{tosin2019MMS}, we consider the following interaction function:
\begin{equation}
	I(v,\,v_\ast;\,z)=P(\rho;\,z)(1-v)+(1-P(\rho;\,z))(P(\rho;\,z)v_\ast-v),
	\label{eq:I}
\end{equation}
where $P(\rho;\,z)\in [0,\,1]$ is the \textit{probability of acceleration}. The function~\eqref{eq:I} expresses the fact that with probability $P$ the $v$-vehicle accelerates towards the maximum speed, while with probability $1-P$ it adapts to the fraction $P$ of the speed of the $v_\ast$-vehicle. The probability of acceleration depends on the (dimensionless) \textit{density} of the vehicles $\rho\in [0,\,1]$, in such a way that the higher $\rho$ the lower $P$, because a dense traffic hinders accelerations. Moreover, $P$ depends also on the uncertain parameter $z$. A possible form is:
\begin{equation}
	P(\rho;\,z)=(1-\rho)^z, \qquad z>0.
	\label{eq:P}
\end{equation}

\begin{figure}[!t]
\centering
\includegraphics[width=0.6\textwidth]{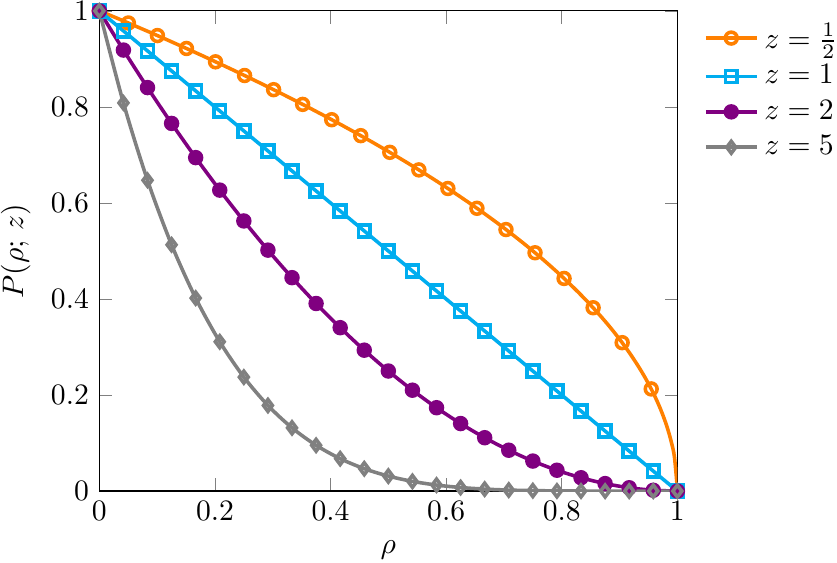}
\caption{The probability of accelerating $P(\rho;\,z)$ given in~\eqref{eq:P} plotted for various $z$.}
\label{fig:P}
\end{figure}

From Figure~\ref{fig:P}, we observe that the mapping $\rho\mapsto P(\rho;\,z)$ expresses the same qualitative trend ($P$ decreasing with $\rho$) for all $z>0$. Nevertheless, it may induce quantitatively different dynamics, because such a trend is either concave, linear or convex depending on the specific value of $z$. It is therefore reasonable to regard the exponent $z\in\R_+$ as an uncertain parameter, considering that there are apparently no \textit{a priori} motivations to opt for a specific value of $z$ in the heuristic model~\eqref{eq:P}. Moreover, since $\rho\in [0,\,1]$, we notice that the higher $z$ the lower the probability of acceleration at all densities $\rho>0$. This suggests that the values of $z$ may be associated with different classes of vehicles characterised by a different motility in the traffic stream. We will come back to this interpretation in Section~\ref{sect:Boltzmann.funddiag}.

Returning to~\eqref{eq:binary.gen}, the term $D(v)\eta$ describes stochastic fluctuations caused by the intrinsic randomness of the driver behaviour. Specifically, $\eta\in\R$ is a centred random variable, i.e. one with zero mean:
\begin{equation}
	\ave{\eta}=0, \qquad \ave{\eta^2}=\sigma^2,
	\label{eq:eta}
\end{equation}
where $\ave{\cdot}$ denotes the average with respect to the probability distribution of $\eta$ and $\sigma>0$ is the standard deviation of $\eta$. Moreover, $D:[0,\,1]\to\R_+$ is a speed-dependent non-negative diffusion coefficient modulating the amplitude of the stochastic fluctuation. We stress that, unlike $z$, the random variable $\eta$ does not describe a structural uncertainty of the model but rather a stochastic completion of the deterministic dynamics expressed by the interaction function $I$.

Finally, we notice that~\eqref{eq:binary.gen} prescribes that the $v_\ast$-vehicle does not change speed during an interaction. This is because vehicle interactions are normally anisotropic, the leading vehicle being unperturbed by the rear vehicle.

\subsubsection{Physical admissibility of the interaction rules}
\label{sect:phys.admiss}
In order to be physically admissible, the interaction rules
\begin{align}
	\begin{aligned}[c]
		v' &= v+\gamma[P(\rho;\,z)(1-v)+(1-P(\rho;\,z))(P(\rho;\,z)v_\ast-v)]+D(v)\eta \\
		v_\ast' &= v_\ast
	\end{aligned}
	\label{eq:binary}
\end{align}
have to guarantee $v',\,v_\ast'\in [0,\,1]$ for every $v,\,v_\ast\in [0,\,1]$. This is actually obvious for $v_\ast'$, while it is more delicate for $v'$.

Let us assume $\gamma\in [0,\,1]$. We begin by observing that, since $0\leq P(\rho;\,z)\leq 1$ and $v_\ast\geq 0$, a sufficient condition for $v'\geq 0$ is
$$ (1-\gamma)v+D(v)\eta\geq 0, $$
which is certainly satisfied if there exists a constant $c>0$ such that
$$ \eta\geq c(\gamma-1), \qquad cD(v)\leq v. $$
Likewise, since $v_\ast\leq 1$, a sufficient condition for $v'\leq 1$ is
$$ (\gamma-1)(1-v)+D(v)\eta\leq 0, $$
which is satisfied if
$$ \eta\leq c(1-\gamma), \qquad cD(v)\leq 1-v. $$

On the whole, the physical admissibility of~\eqref{eq:binary} is guaranteed by the following sufficient conditions:
\begin{equation}
	\abs{\eta}\leq c(1-\gamma), \qquad cD(v)\leq\min\{v,\,1-v\},
	\label{eq:eta.D}
\end{equation}
where $c>0$ is an arbitrary constant. These conditions imply that the stochastic fluctuation $\eta$ is bounded and that $D(0)=D(1)=0$.

\subsection{Boltzmann-type aggregate description and traffic diagrams}
\label{sect:Boltzmann.funddiag}
A statistical description of the aggregate dynamics resulting from the superposition of many binary interactions~\eqref{eq:binary} may be obtained by introducing the distribution function $f=f(t,\,v;\,z)$, where $t\geq 0$ is the time, see~\cite{cercignani1994BOOK}. Specifically, $f$ is such that $f(t,\,v;\,z)dv$ gives the probability that, at time $t$, a vehicle travels with a speed comprised between $v$ and $v+dv$, given the uncertain parameter $z$. Since, under~\eqref{eq:eta.D}, all interactions~\eqref{eq:binary} are physically admissible, the distribution function $f$ evolves according to the following \textit{Boltzmann-type equation} for \textit{Maxwellian-like particles}, here written in weak form, see~\cite{pareschi2013BOOK}:
\begin{equation}
	\frac{d}{dt}\int_0^1\varphi(v)f(t,\,v;\,z)\,dv=
		\frac{1}{2\tau}\int_0^1\int_0^1\ave{\varphi(v')-\varphi(v)}f(t,\,v;\,z)f(t,\,v_\ast;\,z)\,dv\,dv_\ast
	\label{eq:Boltz}
\end{equation}
for every \textit{observable quantity} $\varphi:[0,\,1]\to\R$, namely any quantity which may be expressed as a function of the microscopic state $v$ of the vehicles. This equation states that the time variation of the mean of $\varphi$ (left-hand side) is due to the mean of the average variation of $\varphi$ in a representative binary interaction (right-hand side). The parameter $\tau>0$ at the right-hand side is a relaxation time or, in other words, $\frac{1}{\tau}$ is the interaction frequency.

We have stressed that $f$ is parametrised by the uncertain parameter $z$, because its evolution depends, among other things, on the uncertain interaction function $I$ contained in $v'$. Consequently,~\eqref{eq:Boltz} is a \textit{stochastic} kinetic equation, whose solution may be regarded, for fixed $t$, $v$, as a random variable given as a function of the uncertain parameter $z$.

Letting $\varphi(v)=1$ in~\eqref{eq:Boltz}, we discover
$$ \frac{d}{dt}\int_0^1f(t,\,v;\,z)\,dv=0, $$
therefore our kinetic model conserves in time the total mass of the system. We notice that the interpretation of the mapping $v\mapsto f(t,\,v;\,z)$ as a probability density function on $[0,\,1]$ for all $t$, $z$ is compatible with this property of equation~\eqref{eq:Boltz}. With $\varphi(v)=v$ and recalling~\eqref{eq:eta}, we obtain instead from~\eqref{eq:Boltz} the evolution of the mean speed
$$ V(t;\,z):=\int_0^1vf(t,\,v;\,z)\,dv, $$
i.e.:
\begin{align}
	\begin{aligned}[b]
		\frac{dV}{dt} &= \frac{1}{2\tau}\int_0^1\int_0^1 I(v,\,v_\ast;\,z)f(t,\,v;\,z)f(t,\,v_\ast;\,z)\,dv\,dv_\ast \\
		&=\frac{\gamma}{2\tau}\left[P(\rho;\,z)(1-V)-\bigl(1-P(\rho;\,z)\bigr)^2V\right].
	\end{aligned}
	\label{eq:V}
\end{align}	
This equation can be solved explicitly to find the mapping $t\mapsto V(t;\,z)$ parametrised by $\rho$. Here, we are in particular interested in the asymptotic value reached by $V$ for $t\to +\infty$, say $V_\infty$, which describes the mean speed emerging when interactions are in \textit{equilibrium}:
\begin{equation}
	V_\infty(\rho;\,z):=\frac{P(\rho;\,z)}{P(\rho;\,z)+{\bigl(1-P(\rho;\,z)\bigr)}^2}.
	\label{eq:Vinf}
\end{equation}
Notice that, due to the uncertain parameter $z$, this is in turn a stochastic quantity. We may compute the expectation of $V_\infty$ with respect to $z$ as
\begin{equation}
	\bar{V}_\infty(\rho):=\E_z(V_\infty(\rho;\,z))=\int_{\R_+}V_\infty(\rho;\,z)\Psi(z)\,dz
	\label{eq:Vinf_bar}
\end{equation}
and its variance as
\begin{equation}
	\varsigma_\infty^2(\rho):=\Var_z(V_\infty(\rho;\,z))=\int_{\R_+}V_\infty^2(\rho;\,z)\Psi(z)\,dz-\bar{V}_\infty^2(\rho).
	\label{eq:varsigmainf}
\end{equation}
Next, we observe that we may use the mapping $\rho\mapsto\rho\bar{V}_\infty(\rho)$ to define the \textit{fundamental diagram} of traffic, namely the equilibrium relationship between the traffic density and the macroscopic flux of the vehicles. Together with its $z$-standard deviation $\rho\mapsto\rho\varsigma_\infty(\rho)$, it produces the following set
\begin{equation}
	\left\{(\rho,\,q)\in [0,\,1]\times\R_+\,:\,q\in\left[\rho\bar{V}_\infty(\rho)-\rho\varsigma_\infty(\rho),\,\rho\bar{V}_\infty(\rho)+\rho\varsigma_\infty(\rho)\right]\right\}
	\label{eq:area}
\end{equation}
in the density-flux plane, where most of the random flux values $\rho V_\infty(\rho;\,z)$ lie. The set~\eqref{eq:area} explains the \textit{scattering} of the fundamental diagram, typically found in empirical measurements of the flow of vehicles, as the result of the superposition of different microscopic dynamics produced by different values of $z$ and weighted by the corresponding probability measure $\Psi(z)dz$.

For a general probability distribution $\Psi$, the exact computation of~\eqref{eq:Vinf_bar},~\eqref{eq:varsigmainf} may be non-trivial and one often needs to rely on numerical quadrature formulas. However, for special classes of probability distributions $\Psi$, such as those considered below, analytical results can be obtained.

Let us consider the case in which $z\in\{z_1,\,z_2,\,\dots,\,z_n\}\subset\R_+$ is a discrete random variable with law
$$ \P(z=z_k)=\alpha_k\in [0,\,1], \qquad \sum_{k=1}^{n}\alpha_k=1, $$
so that
$$ \Psi(z)=\sum_{k=1}^{n}\alpha_k\delta(z-z_k), $$
where $\delta(z-z_k)$ is the Dirac delta distribution centred in $z_k$. Then
$$ \bar{V}_\infty(\rho)=\sum_{k=1}^{n}\alpha_kV_\infty(\rho;\,z_k), \qquad
	\varsigma_\infty^2(\rho)=\sum_{k=1}^{n}\alpha_kV_\infty^2(\rho;\,z_k)-\left(\sum_{k=1}^{n}\alpha_kV_\infty(\rho;\,z_k)\right)^2. $$
We observe that $V_\infty(\rho;\,z_k)$, $k=1,\,\dots,\,n$, is the result of the microscopic dynamics~\eqref{eq:binary} with $z=z_k$. If we interpret the $z_k$'s as characteristic values of certain classes of vehicles (for instance, cars, lorries, motorcycles,~\dots) which may travel along the road, the formulas above show that $\bar{V}_\infty(\rho)$, $\varsigma_\infty^2(\rho)$ originate from the average superposition of the macroscopic dynamics produced by each of such classes. Within this interpretation, the $\alpha_k$'s can be understood as the proportions of the various classes of vehicles present in the traffic stream. This provides a point of view on multi-class traffic models based on structural uncertainties in the microscopic composition of the traffic ``mixture''. We mention that other multi-class traffic models are already present in the literature, see e.g.~\cite{benzoni2003EJAM,puppo2016CMS}, which also explain the scattering of the fundamental diagram by appealing to similar physical motivations but different mathematical formalisations.

\begin{figure}[!t]
\centering
\includegraphics[width=0.9\textwidth]{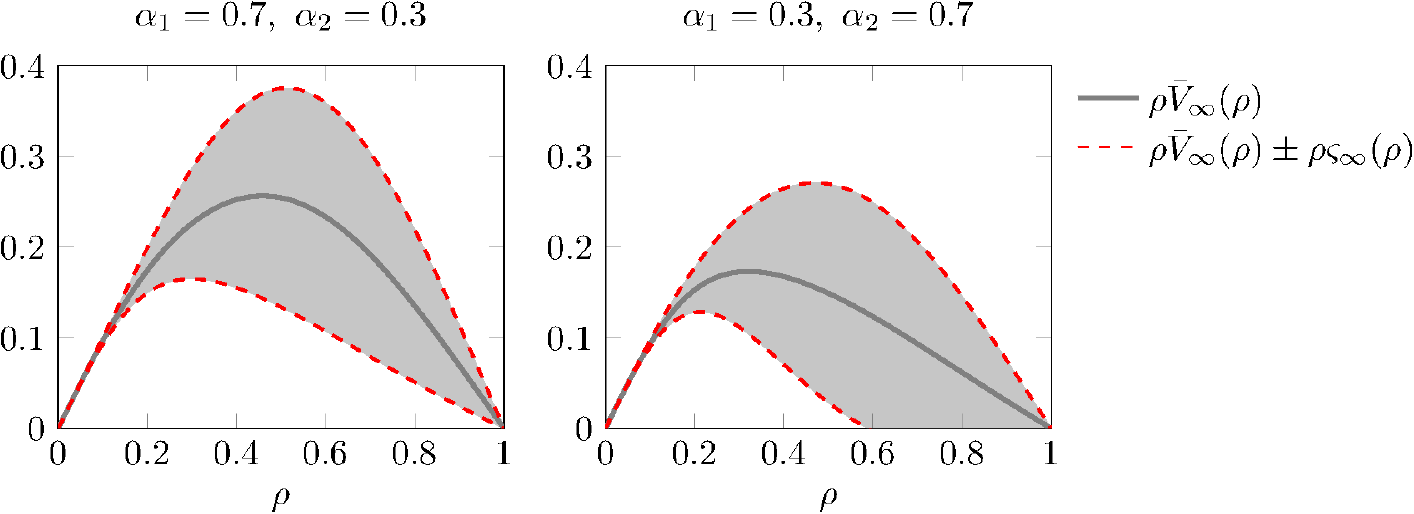}
\caption{The case $z\in\{1,\,3\}$ with $\P(z=1)=\alpha_1$, $\P(z=3)=\alpha_2$: fundamental diagram $\rho\mapsto\rho\bar{V}_\infty(\rho)$ (solid line) and uncertainty lines $\rho\mapsto\rho\bar{V}_\infty(\rho)\pm\rho\varsigma_\infty(\rho)$ (dash-dotted and dashed lines, respectively). The filled area is~\eqref{eq:area}. The function $P(\rho;\,z)$ is taken like in~\eqref{eq:P}.}
\label{fig:funddiag.Z_discr}
\end{figure}

Let us assume, for instance, that $z$ takes only the two values $z_1=1$, with probability $\alpha_1$, and $z_2=3$, with probability $\alpha_2$. The largest value of $z$, i.e. $z_2$, may represent e.g., lorries, which accelerate less especially at high traffic density, cf.~\eqref{eq:P}. Conversely, the smallest value of $z$, i.e. $z_1$, may represent e.g., cars. We consider two different compositions of the traffic stream: first, one with $\alpha_1=70\%$ of cars and $\alpha_2=30\%$ of lorries; then, the opposite one with $\alpha_1=30\%$ of cars and $\alpha_2=70\%$ of lorries. In Figure~\ref{fig:funddiag.Z_discr}, we see that both the fundamental diagram and the region~\eqref{eq:area} change realistically according to the composition of the traffic stream. In both cases, we notice that the scattering of the fundamental diagram is quite limited at low density, i.e. in the so-called \textit{free flow regime}, when the flux grows almost linearly with $\rho$. It becomes more marked at higher density, i.e. in the so-called \textit{congested flow regime}, when the flux decreases non-linearly with $\rho$. This is very nicely in agreement with the typical experimental observations, see e.g.~\cite[Chapter 2]{kerner2004BOOK} and also~\cite{helbing2001RMP,seibold2013NHM}.

As a second example, let us consider $z\sim\cU([a,\,b])$, i.e. the case in which $z$ is a continuous random variable uniformly distributed in the interval $[a,\,b]$ with $0\leq a<b$. Hence
$$ \Psi(z)=\frac{1}{b-a}\chi(a\leq z\leq b), $$
where $\chi$ denotes the characteristic function (specifically, $\chi(a\leq z\leq b)=1$ if $a\leq z\leq b$ while $\chi(a\leq z\leq b)=0$ otherwise). Using~\eqref{eq:P} and~\eqref{eq:Vinf}, from~\eqref{eq:Vinf_bar} we compute:
$$ \bar{V}_\infty(\rho)=\frac{2}{\sqrt{3}(b-a)\log{(1-\rho)}}\left[\arctan\!{\left(\frac{2x-1}{\sqrt{3}}\right)}\right\vert_{x={(1-\rho)}^a}^{x={(1-\rho)}^b}. $$
Also the variance $\varsigma_\infty^2(\rho)$ can be given an explicit representation, indeed from~\eqref{eq:varsigmainf} we obtain:
$$ \varsigma_\infty^2(\rho)=\frac{1}{3(b-a)\log{(1-\rho)}}\left[\frac{\sqrt{x}-2}{x-\sqrt{x}+1}
	+\frac{2}{\sqrt{3}}\arctan\!{\left(\frac{2\sqrt{x}-1}{\sqrt{3}}\right)}\right\vert_{x={(1-\rho)}^{2a}}^{x={(1-\rho)}^{2b}}
		-\bar{V}_\infty^2(\rho). $$
		
\begin{figure}[!t]
\centering
\includegraphics[width=0.6\textwidth]{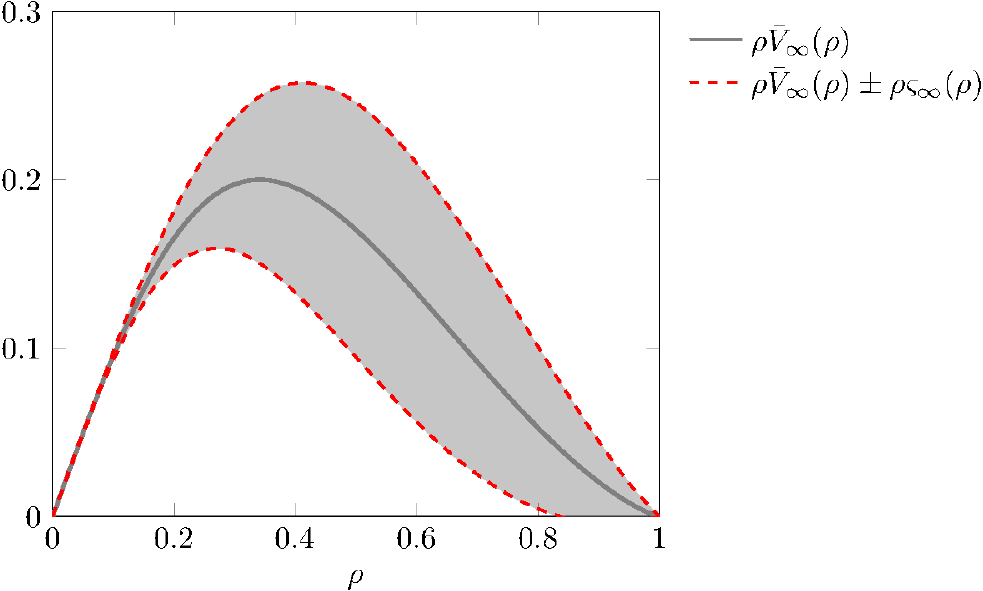}
\caption{The case $z\sim\cU([1,\,3])$: fundamental diagram $\rho\mapsto\rho\bar{V}_\infty(\rho)$ (solid line) and uncertainty lines $\rho\mapsto\rho\bar{V}_\infty(\rho)\pm\rho\varsigma_\infty(\rho)$ (dash-dotted and dashed lines, respectively). The filled area is~\eqref{eq:area}. The function $P(\rho;\,z)$ is taken like in~\eqref{eq:P}.}
\label{fig:funddiag.Z_unif}
\end{figure}

Figure~\ref{fig:funddiag.Z_unif} shows the area defined by~\eqref{eq:area} for $z\sim\cU([1,\,3])$, i.e. $a=1$ and $b=3$ in the formulas above. Also in this case, we notice that the computed scattering of the fundamental diagram reproduces correctly the qualitative empirical features mentioned in the previous example.

\begin{remark}
In the two examples discussed above, the statistical distribution $\Psi$ has been postulated. As already mentioned in the Introduction, the problem of estimating $\Psi$ is beyond the scope of this paper. Nevertheless, recent efforts in this direction based on the analysis of real traffic data, see e.g.~\cite{herty2018SIAP,herty2019PREPRINT}, show that the types of distribution and the range of $z$ considered here are consistent with the results obtained from experimental measurements.
\end{remark}

\subsection{Fokker-Planck description and equilibria}
\label{sect:FP}
Besides the fundamental diagram of traffic, from the kinetic description of the microscopic dynamics~\eqref{eq:binary} we may also obtain information on the statistical distribution of the speed of the vehicles emerging when interactions are in equilibrium, i.e. for large times ($t\to +\infty$ in the limit). Such a distribution corresponds to the \textit{Maxwellian distribution} of the classical kinetic theory of gases.

By definition, the equilibrium distribution, say $f_\infty=f_\infty(v;\,z)$, makes the right-hand side of~\eqref{eq:Boltz} vanish:
$$ \int_0^1\int_0^1\ave{\varphi(v')-\varphi(v)}f_\infty(v;\,z)f_\infty(v_\ast;\,z)\,dv\,dv_\ast=0 $$
for all observable quantities $\varphi$. Nevertheless, it is in general difficult to determine $f_\infty$ by tackling directly this equation, at least from the analytical point of view. If, however, the interactions~\eqref{eq:binary} are \textit{quasi-invariant}, i.e. if they produce each time a small change of speed, then it is possible to approximate asymptotically the Boltzmann-type integro-differential equation~\eqref{eq:Boltz} with a \textit{Fokker-Planck} partial differential equation, which may allow for an explicit computation of the steady states.

Quasi-invariant interactions are reminiscent of the \textit{grazing collisions} introduced in the classical kinetic theory~\cite{villani1998PhD,villani1998ARMA}. In our context, they amount to assuming $\gamma,\,\sigma^2\ll 1$ in~\eqref{eq:binary}, so that both the deterministic part of the interactions and the stochastic fluctuations are small. To balance such a smallness and observe asymptotic trends, the interaction frequency has to be increased accordingly. We obtain such an effect by means of the following scaling:
\begin{equation}
	\gamma=\epsilon, \qquad \sigma^2=\lambda\epsilon, \qquad \tau=\frac{\epsilon}{2},
	\label{eq:scaling}
\end{equation}
where $0<\epsilon\ll 1$ is a scaling parameter and $\lambda>0$ a proportionality constant. From~\eqref{eq:Boltz} we obtain therefore
\begin{equation}
	\frac{d}{dt}\int_0^1\varphi(v)f_\epsilon(t,\,v;\,z)\,dv=
		\frac{1}{\epsilon}\int_0^1\int_0^1\ave{\varphi(v')-\varphi(v)}f_\epsilon(t,\,v;\,z)f_\epsilon(t,\,v_\ast;\,z)\,dv\,dv_\ast,
	\label{eq:Boltz.eps}
\end{equation}
where we have stressed that the solution $f$ depends now also on the scaling parameter $\epsilon$.

From now on, we take $\varphi$ smooth, specifically $\varphi\in C^\infty_c(0,\,1)$. We assume that, for every $\epsilon>0$ and every $z\in\R_+$,~\eqref{eq:Boltz.eps} admits a unique solution $f_\epsilon(\cdot,\,\cdot;\,z)\in C(\R_+;\,L^1(0,\,1))$ which is a probability density function in the variable $v$, cf.~\cite[Appendix A]{freguglia2017CMS}. Moreover, we assume that there exists a probability density function $f(\cdot,\,\cdot;\,z)\in C(\R_+;\,L^1(0,\,1))$ such that, up to subsequences, $f_\epsilon(\cdot,\,\cdot;\,z)\to f(\cdot,\,\cdot;\,z)$ in $C(\R_+;\,L^1(0,\,1))$ when $\epsilon\to 0^+$.

Expanding $\varphi(v')$ in Taylor's series up to the third order about $v$ (because $v'\approx v$ in the quasi-invariant regime), we get
\begin{align*}
	\frac{d}{dt}\int_0^1\varphi(v)f_\epsilon(t,\,v;\,z)\,dv &=\int_0^1\int_0^1\varphi'(v)I(v,\,v_\ast;\,z)f_\epsilon(t,\,v;\,z)f_\epsilon(t,\,v_\ast;\,z)\,dv\,dv_\ast \\
	&\phantom{=} +\frac{\lambda}{2}\int_0^1\varphi''(v)D^2(v)f_\epsilon(t,\,v;\,z)\,dv+R_\epsilon(f_\epsilon,\,f_\epsilon)[\varphi](t;\,z),
\end{align*}
where $I$ is given by~\eqref{eq:I} and $R_\epsilon(f_\epsilon,\,f_\epsilon)[\varphi]$ is a remainder which depends on $\varphi'''$ and on $(v'-v)^3$, cf.~\cite{tosin2019MMS}. In the limit $\epsilon\to 0^+$, it can be proved that $R_\epsilon(f_\epsilon,\,f_\epsilon)[\varphi]\to 0$ for all $\varphi\in C^\infty_c([0,\,1])$, see again~\cite{tosin2019MMS} for details. Consequently, in view of the assumptions stated above, we find that the asymptotic equation satisfied by $f$ is
\begin{align*}
	\frac{d}{dt}\int_0^1\varphi(v)f(t,\,v;\,z)\,dv &=\int_0^1\int_0^1\varphi'(v)I(v,\,v_\ast;\,z)f(t,\,v;\,z)f(t,\,v_\ast;\,z)\,dv\,dv_\ast \\
	&\phantom{=} +\frac{\lambda}{2}\int_0^1\varphi''(v)D^2(v)f(t,\,v;\,z)\,dv,
\end{align*}
i.e., integrating by parts the right-hand side and invoking the compactness of the support and the arbitrariness of $\varphi$:
\begin{equation}
	\partial_tf=\frac{\lambda}{2}\partial_v^2\left(D^2(v)f\right)-\partial_v\bigl[\bigl(P(1+(1-P)V(t;\,z))-v\bigr)f\bigr],
	\label{eq:FP}
\end{equation}
where we have substituted the expression~\eqref{eq:I} of $I$.

Equation~\eqref{eq:FP} is the announced Fokker-Planck asymptotic model, which describes the trend of the system in the quasi-invariant regime. Notice that~\eqref{eq:FP} is a Fokker-Planck equation with non-constant coefficients, due to the fact that the mean speed $V(t;\,z)$ is not conserved in time by the interactions~\eqref{eq:binary}. The stationary solution $f_\infty$ to~\eqref{eq:FP}, which provides an asymptotic approximation of the Maxwellian of~\eqref{eq:Boltz}, solves
$$ \frac{\lambda}{2}\partial_v\left(D^2(v)f_\infty\right)-\left(V_\infty-v\right)f_\infty=0, $$
where $V_\infty$ is the asymptotic mean speed~\eqref{eq:Vinf}. In particular, if we choose\footnote{The function~\eqref{eq:D} does not actually comply with the requirement in~\eqref{eq:eta.D}, due to the vertical tangents at $v=0,\,1$. Nevertheless, such a $D$ may be obtained as the uniform limit for $\epsilon\to 0^+$ of e.g., the sequence of functions $D_\epsilon(v):=\sqrt{[(1+\epsilon)v(1-v)-\frac{\epsilon}{4}]_+} $, where $[\cdot]_+$ denotes the positive part. It can be checked that $D_\epsilon$ complies with~\eqref{eq:eta.D} with $c=\sqrt{\frac{\epsilon}{1+\epsilon}}$, see~\cite{toscani2006CMS} for details. Therefore, the use of~\eqref{eq:D} in~\eqref{eq:FP} is justified after performing the quasi-invariant limit.}
\begin{equation}
	D(v):=\sqrt{v(1-v)}
	\label{eq:D}
\end{equation}
we obtain that the unique stationary solution with unitary mass is
\begin{equation}
	f_\infty(v;\,z)=\frac{v^{\frac{2V_\infty(\rho;\,z)}{\lambda}-1}{(1-v)}^{\frac{2(1-V_\infty(\rho;\,z))}{\lambda}-1}}
		{\operatorname{Beta}\!\left(\frac{2V_\infty(\rho;\,z)}{\lambda},\,\frac{2(1-V_\infty(\rho;\,z))}{\lambda}\right)},
	\label{eq:ginf}
\end{equation}
where $\operatorname{Beta}(\cdot,\,\cdot)$ is the beta function. We notice that~\eqref{eq:ginf} is a beta probability density function with mean $V_\infty(\rho;\,z)$, as expected, and variance $\frac{\lambda V_\infty(\rho;\,z)(1-V_\infty(\rho;\,z))}{2+\lambda}$. Interestingly, beta probability densities have been found to fit particularly well the experimental data on the speed distribution of the vehicles~\cite{maurya2016TRP,ni2018AMM}.

\begin{figure}[!t]
\centering
\subfigure[$z-1\sim\B\!\left(50,\,\frac{1}{50}\right),\ \rho=0.2$]{\includegraphics[scale=0.4]{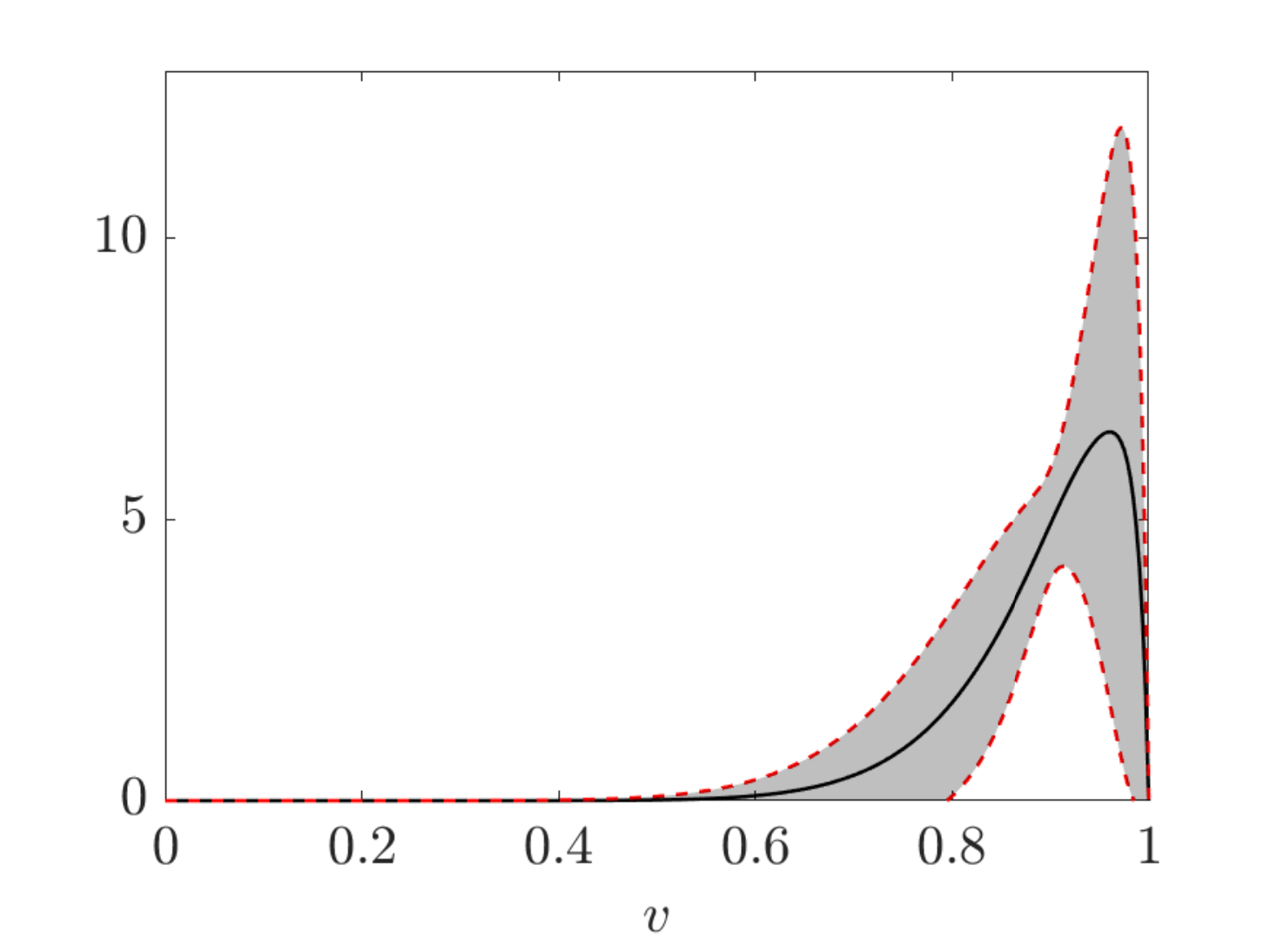}}
\subfigure[$z-1\sim\B\!\left(50,\,\frac{1}{50}\right),\ \rho=0.4$]{\includegraphics[scale=0.4]{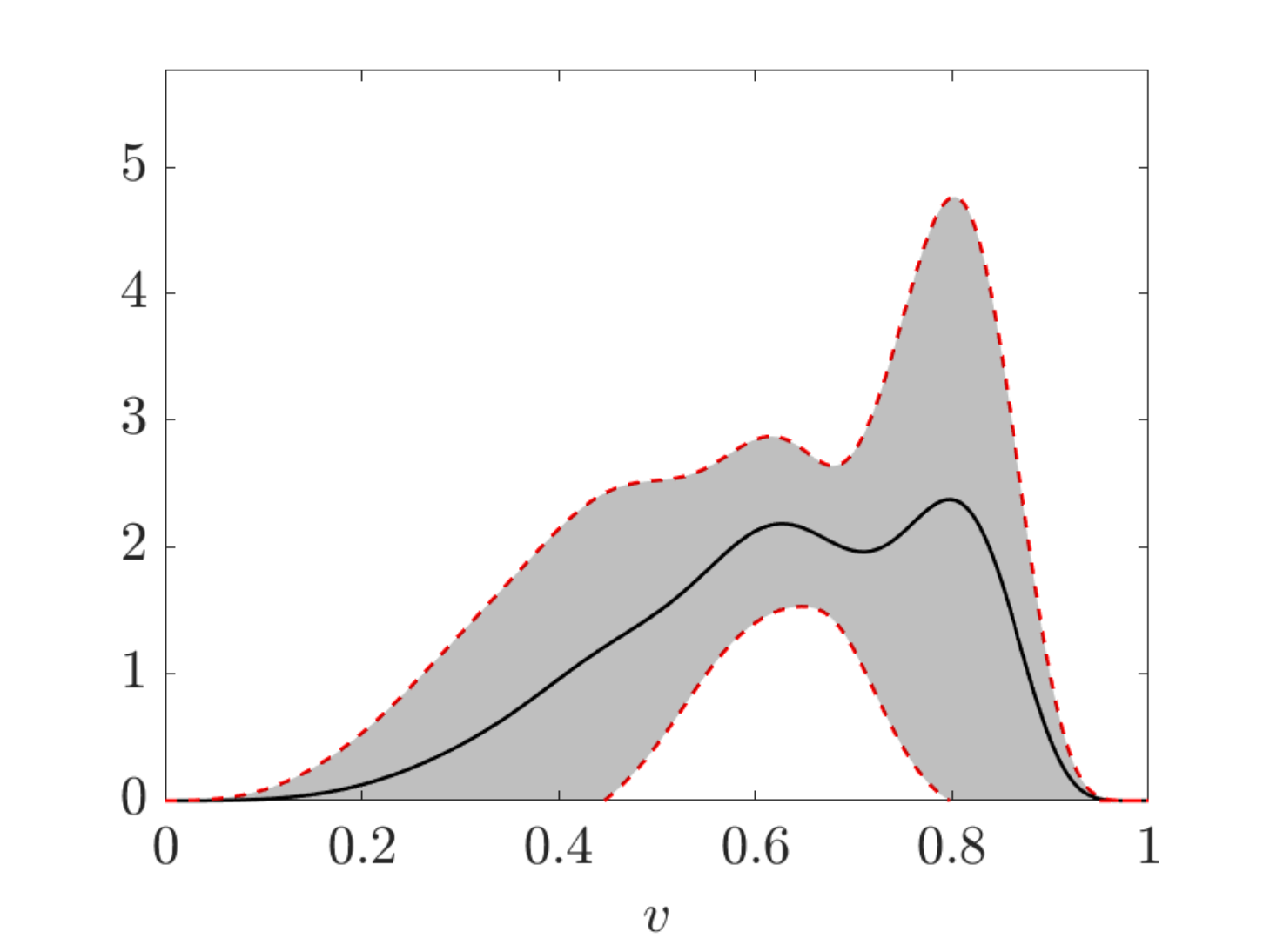}}
\caption{Representation of $\bar{f}_\infty(v)$ (solid lines) and of its $z$-standard deviation $\pm\sqrt{\Var_z(f_\infty(v;\,z))}$ (dashed lines) for $z$ such that $z-1$ has a binomial distribution. In (a), $\rho=0.2$. In (b), $\rho=0.4$.}
\label{fig:ginf.binom}
\end{figure}
\begin{figure}[!t]
\centering
\subfigure[$z\sim\cU({[1,\,3]})$]{\includegraphics[scale=0.4]{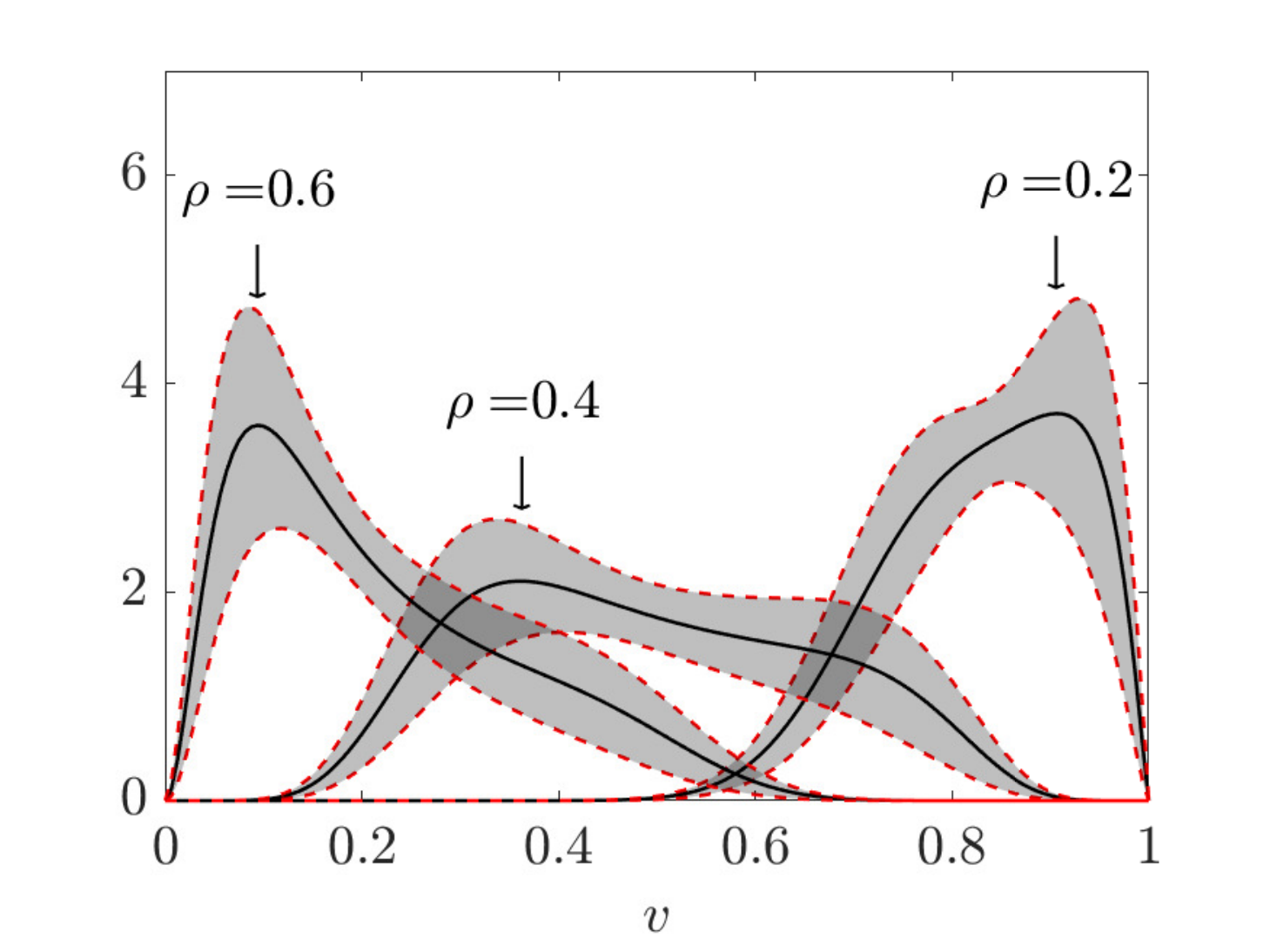}}
\subfigure[$z-2\sim\operatorname{Gamma}(3,\,3)$]{\includegraphics[scale=0.4]{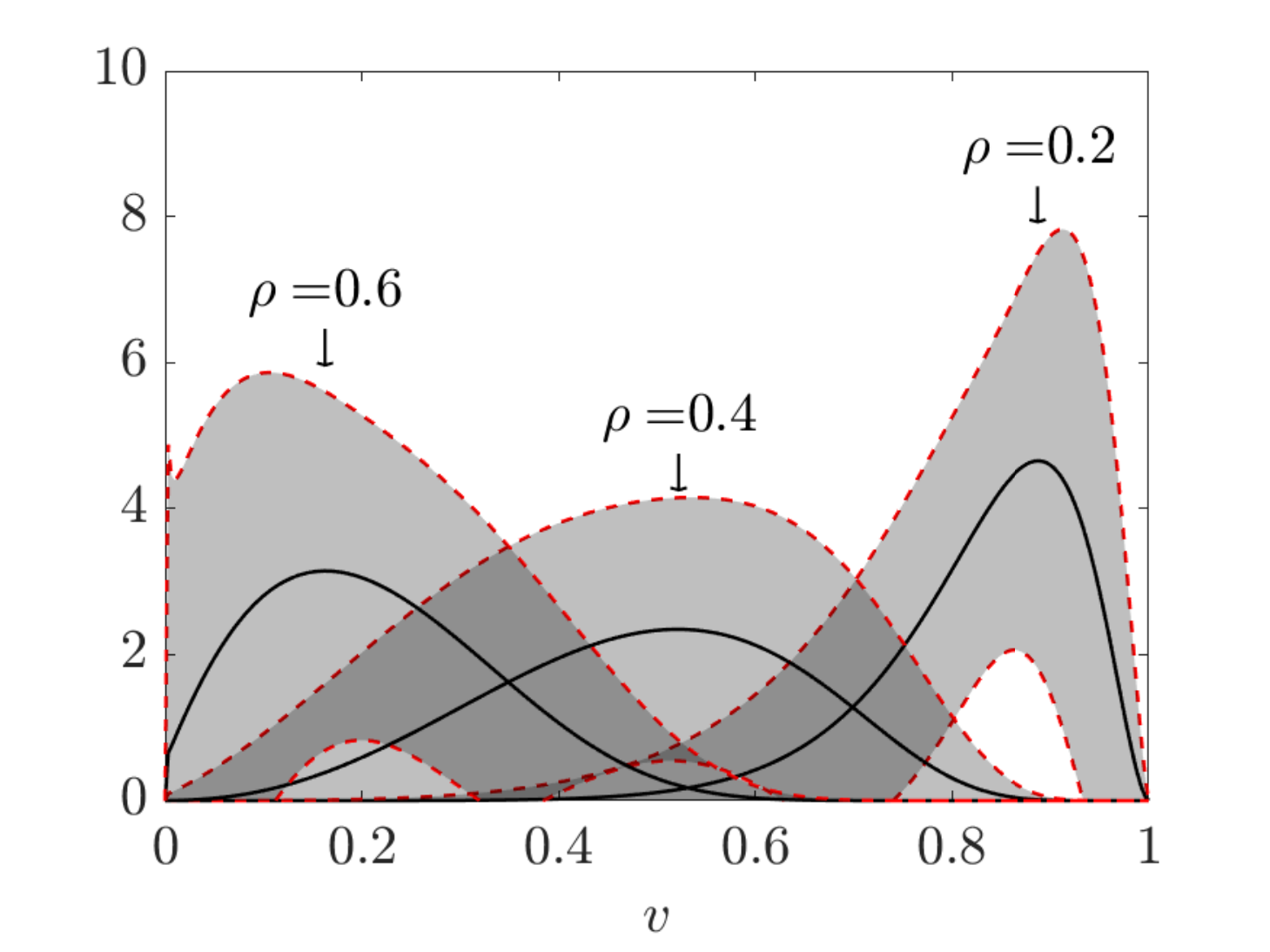}}
\caption{Representation of $\bar{f}_\infty(v)$ (solid lines) and of its $z$-standard deviation $\pm\sqrt{\Var_z(f_\infty(v;\,z))}$ (dashed lines) for different distributions of the uncertain parameter $z$ and various traffic densities. In (a), $z$ is uniformly distributed in $[1,\,3]$. In (b), $z$ is such that $z-2$ has a gamma distribution, thus it is, in particular, unbounded.}
\label{fig:ginf.unif-gamma}
\end{figure}

Actually,~\eqref{eq:ginf} is a family of equilibrium distributions parametrised by $z$. In order to quantify the intrinsic uncertainty in this description of the equilibrium, we may refer to the mean equilibrium distribution:
$$ \bar{f}_\infty(v):=\E_z(f_\infty(v;\,z))=\int_{\R_+}f_\infty(v;\,z)\Psi(z)\,dz $$
and to its variance
$$ \Var_z(f_\infty(v;\,z)):=\int_{\R_+}f_\infty^2(v;\,z)\Psi(z)\,dz-\bar{f}_\infty^2(v). $$
Notice, from~\eqref{eq:Vinf_bar}, that
$$ \bar{V}_\infty(v)=\int_0^1v\bar{f}_\infty(v)\,dv. $$

Figures~\ref{fig:ginf.binom},~\ref{fig:ginf.unif-gamma} show $\bar{f}_\infty$, together with the uncertainty bands determined by its standard deviation $\pm\sqrt{\Var_z(f_\infty(\cdot;\,z))}$, for various discrete and continuous probability distributions of $z$. The curves have been computed numerically by means of suitable quadrature formulas in $z$, starting from the analytical expression~\eqref{eq:ginf}. We observe that, depending on the distribution of $z$, in the mean distribution $\bar{f}_\infty$ several local peaks may appear, which suggest a slightly multi-modal average trend of the equilibrium speed distribution. Such a trend is actually documented in experimental data~\cite{herty2019PREPRINT,maurya2016TRP,ni2018AMM}.

\section{Uncertainty damping by driver-assist control}
\label{sect:uncertainty_damping}
In this section, we tackle the problem of damping the scattering of the fundamental diagram, i.e., owing to~\eqref{eq:area}, of reducing the variance $\varsigma_\infty^2(\rho)$. Although the scattering is a \textit{macroscopic} manifestation of the traffic, its origin relies in the structural uncertainty affecting the \textit{microscopic} interactions~\eqref{eq:binary}. Therefore, we will proceed to control \textit{a priori} (some of) the microscopic interactions~\eqref{eq:binary}, assessing then the aggregate impact of such a control, rather than trying to constrain \textit{a posteriori} the macroscopic flow of the vehicles as a whole. Indeed, in our view, actually feasible ways of controlling a multi-agent system, such as vehicular traffic, make use of \textit{bottom-up strategies}, which target the simple microscopic interactions, rather than of strategies pointing directly to the aggregate behaviour. In the context of vehicular traffic, such bottom-up strategies refer to the control of the local responses given by the driver-vehicle system to the external stimuli coming from other vehicles by means of \textit{Advanced Driver-Assistance Systems} (ADAS).

Let us consider, therefore, the following \textit{controlled} binary interaction:
\begin{align}
	\begin{aligned}[c]
		v' &= v+\gamma\left(I(v,\,v_\ast;\,z)+\Theta u\right)+D(v)\eta \\
		v_\ast' &= v_\ast,
	\end{aligned}
	\label{eq:binary-u.gen}
\end{align}
where $I$ is given by~\eqref{eq:I}, $u$ is a control to be determined from the optimisation of a suitable functional and $\Theta\in\{0,\,1\}$ is a Bernoulli random variable taking into account whether the randomly selected $v$-vehicle is or is not equipped with the ADAS technology. In particular, the law of $\Theta$ is
$$ \P(\Theta=1)=p, \qquad \P(\Theta=0)=1-p, $$
where $p\in [0,\,1]$ is the so-called \textit{penetration rate}, i.e. the percentage of ADAS vehicles in the traffic stream.

We propose, in the following, two different functionals, whose minimisation yields an \textit{instantaneous optimal control} $u^\ast$ to be used in~\eqref{eq:binary-u.gen}. These functionals correspond, in principle, to different control strategies, which however, as we will see, lead to the same macroscopic effect. Remarkably, one of such strategies has an eminently theoretical interest, whereas the other is closer to an actual implementability.

\subsection{Pointwise uncertainty control}
\label{sect:pointwise}
Let $v_d(\rho)\in [0,\,1]$ be a \textit{desired} (or \textit{recommended}) \textit{speed} determined as a function of the vehicle density $\rho$ and suggested to the ADAS vehicles by e.g., some sensors along the road, which are able to detect the level of traffic congestion. We assume $0\leq v_d(\rho)\leq 1$ for all $\rho\in [0,\,1]$. Further physically meaningful characteristics of the desired speed are that the mapping $\rho\mapsto v_d(\rho)$ be non-increasing and that $v_d(0)=1$, $v_d(1)=0$. A prototypical example may be $v_d(\rho)=1-\rho$.

Automated traffic flow dynamics have been modelled recently as cruise controls implemented either in the vehicle-to-vehicle interactions, see e.g.~\cite{helbing2001RMP,ntousakis2015TRP,piccoli2019PREPRINT,tosin2019MMS}, or directly in the aggregate traffic dynamics~\cite{delis2015CMA,delis2018TRR}. In general, the effect of such controls is the stabilisation of the traffic flow depending on the penetration rate of the automated vehicles, as also testified by the field experiment reported in~\cite{stern2018TRC}. Taking inspiration from this idea, an optimal control $u^\ast$ may be thought of as aimed at keeping the post-interaction speed $v'$ of an ADAS vehicle as close as possible to $v_d(\rho)$. This is obtained by minimising e.g. the following functional:
$$ J_1(v',\,u;\,z):=\frac{1}{2}\ave{{(v_d(\rho)-v')}^2+\nu u^2} $$
subject to the constraint~\eqref{eq:binary-u.gen}. Notice that $J_1$ depends pointwise on the uncertain parameter $z$ through $v'$. Moreover, $\nu u^2$, with $\nu>0$, is a quadratic penalisation term, which expresses the cost of the control.

The minimisation of $J_1$ may be carried out by forming the Lagrangian:
$$ \cL_1(v',\,u,\,\mu;\,z):= J_1(v',\,u;\,z)+\mu\ave{v'-v-\gamma\left(I(v,\,v_\ast;\,z)+\Theta u\right)-D(v)\eta}, $$
where $\mu\in\R$ is the Lagrange multiplier, and then by considering the optimality conditions:
\begin{equation*}
	\begin{cases}
		\partial_{v'}\cL_1=\mu-\ave{v_d(\rho)-v'}=0 \\
		\partial_u\cL_1=\nu u-\gamma\Theta\mu=0 \\
		\partial_\mu\cL_1=\ave{v'-v-\gamma\left(I(v,\,v_\ast;\,z)+\Theta u\right)-D(v)\eta}=0,
	\end{cases}
\end{equation*}
whence we deduce the optimal control
\begin{equation}
	u^\ast(v,\,v_\ast;\,z)=\frac{\gamma\Theta}{\nu+\gamma^2\Theta^2}\left(v_d(\rho)-v-\gamma I(v,\,v_\ast;\,z)\right).
	\label{eq:ustar.1}
\end{equation}
Notice that, for each $z\in\R_+$, this is a \textit{feedback} control, because it is expressed as a function of the pre-interaction speeds $v$, $v_\ast$. Plugged into~\eqref{eq:binary-u.gen}, it produces the controlled binary interactions
\begin{align}
	\begin{aligned}[c]
		v' &= v+\frac{\nu\gamma}{\nu+\gamma^2\Theta^2}I(v,\,v_\ast;\,z)
			+\frac{\gamma^2\Theta^2}{\nu+\gamma^2\Theta^2}\left(v_d(\rho)-v\right)+D(v)\eta \\
		v'_\ast &= v_\ast.
	\end{aligned}
	\label{eq:binary-u.1}
\end{align}
	
\subsubsection{Physical admissibility of the interactions~\texorpdfstring{\eqref{eq:binary-u.1}}{}}
\label{sect:phys.admiss-u.1}
The physical admissibility of the interactions~\eqref{eq:binary-u.1} requires $v'\in [0,\,1]$ for all $v,\,v_\ast\in [0,\,1]$. Considering the expression~\eqref{eq:I} for $I$, together with the observations made in Section~\ref{sect:phys.admiss} plus $v_d(\rho)\geq 0$ and $\Theta\leq 1$, we see that a sufficient condition for $v'\geq 0$ is
$$ \frac{\nu(1-\gamma)}{\nu+\gamma^2}v+D(v)\eta\geq 0, $$
which is certainly satisfied if there exists a constant $c>0$ such that
$$ \eta\geq c\frac{\nu(\gamma-1)}{\nu+\gamma^2}, \qquad cD(v)\leq v. $$
Likewise, since $v_d(\rho)\leq 1$, a sufficient condition for $v'\leq 1$ is
$$ \frac{\nu(\gamma-1)}{\nu+\gamma^2}(1-v)+D(v)\eta\leq 0, $$
which is certainly satisfied if
$$ \eta\leq c\frac{\nu(1-\gamma)}{\nu+\gamma^2}, \qquad cD(v)\leq 1-v. $$

On the whole,~\eqref{eq:binary-u.1} is physically admissible if there exists $c>0$ such that
\begin{equation}
	\abs{\eta}\leq c\frac{\nu(1-\gamma)}{\nu+\gamma^2}, \qquad cD(v)\leq\min\{v,\,1-v\}.
	\label{eq:eta.D-contr.1}
\end{equation}
Again, this implies that the random variable $\eta$ has to be bounded and that the diffusion coefficient $D$ vanishes at $v=0,\,1$. Notice that if $\nu\to+\infty$, i.e. if the control is so penalised that the optimum is not to control the interactions ($u^\ast=0$, cf.~\eqref{eq:ustar.1}), then we recover exactly~\eqref{eq:eta.D}.

\subsubsection{Boltzmann-type kinetic description}
Since, under~\eqref{eq:eta.D-contr.1}, all the interactions~\eqref{eq:binary-u.1} are physically admissible, the evolution of the distribution function $f=f(t,\,v;\,z)$ is described by a Boltzmann-type equation for Maxwellian-like particles analogous to~\eqref{eq:Boltz}. The only difference is that now, in the ``collisional'' term appearing on the right-hand side of~\eqref{eq:Boltz}, we need to take into account a further average with respect to the random variable $\Theta$. Thus we write:
\begin{equation}
	\frac{d}{dt}\int_0^1\varphi(v)f(t,\,v;\,z)\,dv=\frac{1}{2\tau}\E_\Theta\left[\int_0^1\int_0^1
		\ave{\varphi(v')-\varphi(v)}f(t,\,v;\,z)f(t,\,v_\ast;\,z)\,dv\,dv_\ast\right],
	\label{eq:Boltz.Theta}
\end{equation}
where $v'$ is given by~\eqref{eq:binary-u.1}.

In particular, we are interested in the trend of the mean speed $V=V(t;\,z)$, which we obtain with $\varphi(v)=v$ considering furthermore that $\Theta^2$ in~\eqref{eq:binary-u.1} is a Bernoulli random variable with parameter $p$, because so is $\Theta$ by definition. Hence:
$$ \frac{dV}{dt}=\frac{\gamma}{2\tau}\cdot\frac{\nu+(1-p)\gamma^2}{\nu+\gamma^2}\left[P(\rho;\,z)(1-V)-{(1-P(\rho;\,z))}^2V\right]
	+\frac{\gamma^2}{2\tau}\cdot\frac{p}{\nu+\gamma^2}\left(v_d(\rho)-V\right). $$
	
In order to better characterise the role of the various parameters in the large time trend of the mean speed, it is useful to refer to the quasi-invariant interaction regime introduced in~\eqref{eq:scaling}. In this case, to observe in the limit also a balanced contribution of the controlled part of the interactions, we scale also the control penalisation $\nu$ as
\begin{equation}
	\nu=\kappa\epsilon,
	\label{eq:scaling.nu}
\end{equation}
where $\kappa>0$ is a proportionality parameter which expresses the strength of the control in the quasi-invariant scale. Then, we obtain that the limit equation satisfied by $V$ for $\epsilon\to 0^+$ is
\begin{equation}
	\frac{dV}{dt}=P(\rho;\,z)(1-V)-{(1-P(\rho;\,z))}^2V+p^\ast(v_d(\rho)-V),
	\label{eq:V.1}
\end{equation}
where
$$ p^\ast:=\frac{p}{\kappa} $$
is what we call the \textit{effective penetration rate}. In particular, the equilibrium mean speed is
\begin{equation}
	V_\infty(\rho;\,z):=\frac{P(\rho;\,z)+p^\ast v_d(\rho)}{P(\rho;\,z)+{(1-P(\rho;\,z))}^2+p^\ast}.
	\label{eq:Vinf.1}
\end{equation}

Since $0\leq P(\rho;\,z),\,v_d(\rho)\leq 1$, we have\footnote{We use the notation $a\lesssim b$ to mean that there exists a constant $C>0$, whose exact value is unimportant, such that $a\leq Cb$.}
\begin{equation}
	\abs{V_\infty(\rho;\,z)-v_d(\rho)}
		=\frac{\abs*{P(\rho;\,z)+\left(P(\rho;\,z)+{(1-P(\rho;\,z))}^2\right)v_d(\rho)}}{P(\rho;\,z)+{(1-P(\rho;\,z))}^2+p^\ast}
			\lesssim\frac{1}{p^\ast},
	\label{eq:dist.Vinf-v0}
\end{equation}
which shows that the instantaneous control~\eqref{eq:ustar.1} has the effect of keeping the average traffic stream close to the suggested optimal speed. From this, we easily deduce the important consequence:
\begin{align}
	\begin{aligned}[b]
		\varsigma_\infty^2(\rho) &= \Var_z(V_\infty(\rho;\,z))=\Var_z(V_\infty(\rho;\,z)-v_d(\rho)) \\
		&\leq \E_z\left({(V_\infty(\rho;\,z)-v_d(\rho))}^2\right) \\
		&\lesssim \frac{1}{{(p^\ast)}^2},
	\end{aligned}
	\label{eq:varsigminf.est}
\end{align}
meaning that the control reduces indeed the scattering of the fundamental diagram caused by the uncertain parameter $z$. The magnitude of such a reduction is estimated by the inverse of ${(p^\ast)}^2$, hence it is higher for either a higher penetration rate $p$ or a lower control penalisation $\kappa$.

Interestingly, this result does not actually depend on the specific form~\eqref{eq:I} of the interaction function $I$. A fruitful generalisation is enlightened by the following theorem:
\begin{theorem}
\label{theo:I}
In~\eqref{eq:binary-u.1}, let the interaction function $I$ be bounded, i.e. let a constant $I_{\max}>0$ exist such that
$$ \abs{I(v,\,v_\ast;\,z)}\leq I_{\max} \qquad \forall\,v,\,v_\ast\in [0,\,1],\, z\in\R_+. $$
Assume, moreover, that $I$ is such that the interactions~\eqref{eq:binary-u.1} are all physically admissible in the sense discussed in Section~\ref{sect:phys.admiss-u.1}. Then
$$ \varsigma_\infty^2(\rho)\lesssim\frac{1}{{(p^\ast)}^2}. $$
\end{theorem}
\begin{proof}
We observe that it suffices to prove an estimate like~\eqref{eq:dist.Vinf-v0}, for then the thesis follows immediately from~\eqref{eq:varsigminf.est}, which holds independently of the specific form of $I$.

Since all interactions~\eqref{eq:binary-u.1} are physically admissible by assumption, we may use the Boltzmann-type equation~\eqref{eq:Boltz.Theta}. For $\varphi(v)=v$ we find then
$$ \frac{dV}{dt}=\frac{\gamma}{2\tau}\left[\int_0^1\int_0^1\left(1-\frac{p\gamma^2}{\nu+\gamma^2}\right)I(v,\,v_\ast;\,z)f(t,\,v;\,z)
	f(t,\,v_\ast;\,z)\,dv\,dv_\ast+\frac{p\gamma}{\nu+\gamma^2}\left(v_d(\rho)-V\right)\right], $$
which, in the quasi-invariant limit, becomes
$$ \frac{dV}{dt}=\int_0^1\int_0^1I(v,\,v_\ast;\,z)f(t,\,v;\,z)f(t,\,v_\ast;\,z)\,dv\,dv_\ast+p^\ast\left(v_d(\rho)-V\right). $$

Since $v_d(\rho)$ is time-independent, we have $\frac{dV}{dt}=\frac{d}{dt}\left(V-v_d(\rho)\right)$, therefore we may write
$$ \frac{d}{dt}\left(V-v_d(\rho)\right)+p^\ast\left(V-v_d(\rho)\right)
	=\int_0^1\int_0^1I(v,\,v_\ast;\,z)f(t,\,v;\,z)f(t,\,v_\ast;\,z)\,dv\,dv_\ast, $$
whence, integrating in time,
\begin{align*}
	V(t;\,z)-v_d(\rho) &= e^{-p^\ast t}\left(V_0(z)-v_d(\rho)\right) \\
	&\phantom{=} +\int_0^t e^{-p^\ast(t-s)}\int_0^1\int_0^1I(v,\,v_\ast;\,z)f(s,\,v;\,z)f(s,\,v_\ast;\,z)\,dv\,dv_\ast\,ds,
\end{align*}
with $V_0(z):=V(0;\,z)$. From the boundedness of $I$ we further obtain
$$ \abs*{\int_0^1\int_0^1I(v,\,v_\ast;\,z)f(s,\,v;\,z)f(s,\,v_\ast;\,z)\,dv\,dv_\ast}\leq I_{\max}, $$
thus we discover:
\begin{align*}
	\abs{V(t;\,z)-v_d(\rho)} &\leq e^{-p^\ast t}\left(V_0(z)-v_d(\rho)\right)+I_{\max}\int_0^t e^{-p^\ast(t-s)}\,ds \\
	&= e^{-p^\ast t}\left(V_0(z)-v_d(\rho)\right)+\frac{I_{\max}}{p^\ast}\left(1-e^{-p^\ast t}\right),
\end{align*}
whence finally
$$ \abs{V_\infty(\rho;\,z)-v_d(\rho)}=\lim_{t\to +\infty}\abs{V(t;\,z)-v_d(\rho)}\leq\frac{I_{\max}}{p^\ast} $$
and we are done.
\end{proof}

\begin{figure}[!t]
\centering
\subfigure[$z-1\sim\B\!\left(50,\,\frac{1}{50}\right)$, $p^\ast=1$]{\includegraphics[scale=0.4]{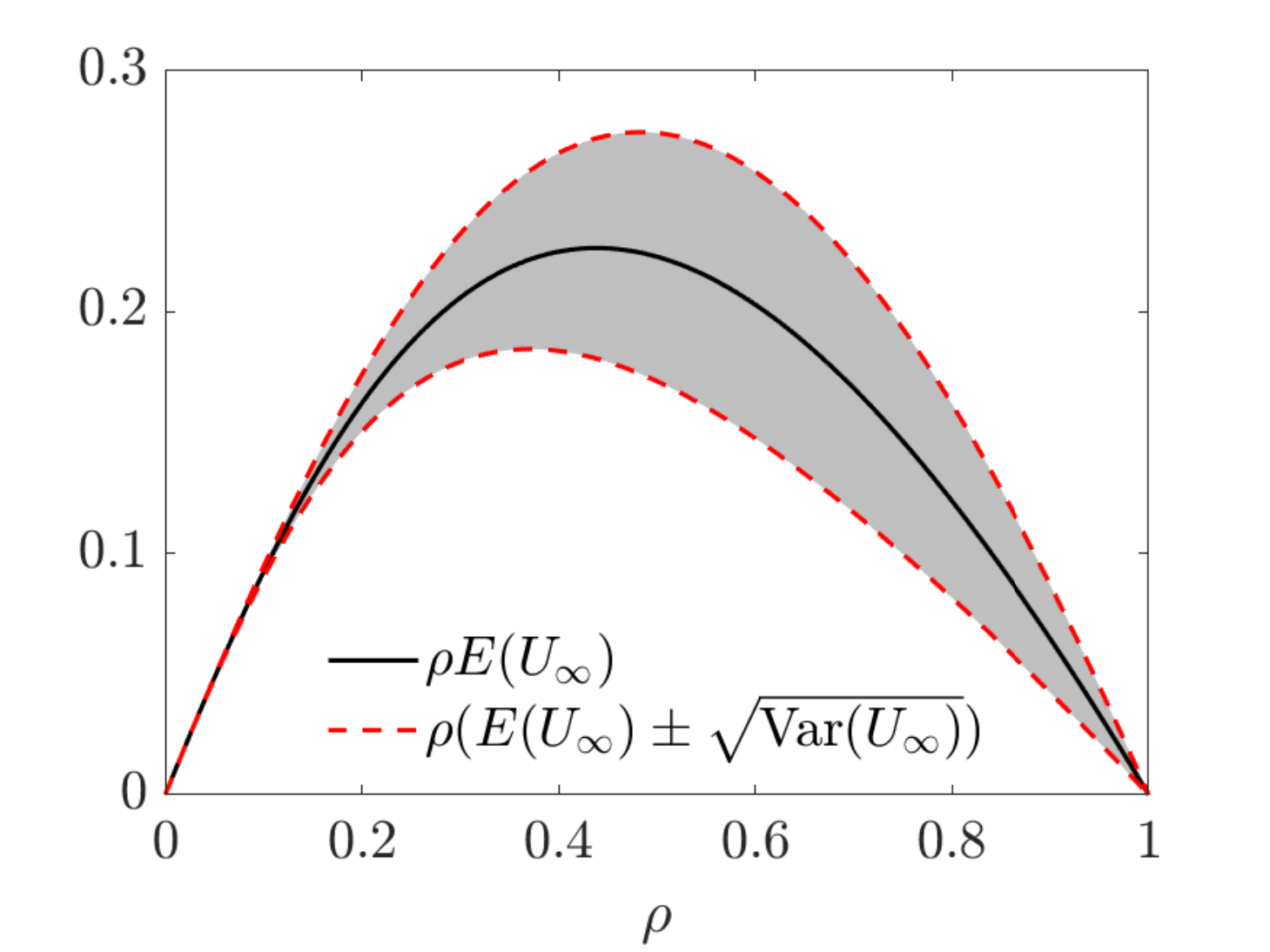}}
\subfigure[$z-1\sim\B\!\left(50,\,\frac{1}{50}\right)$, $p^\ast=10$]{\includegraphics[scale=0.4]{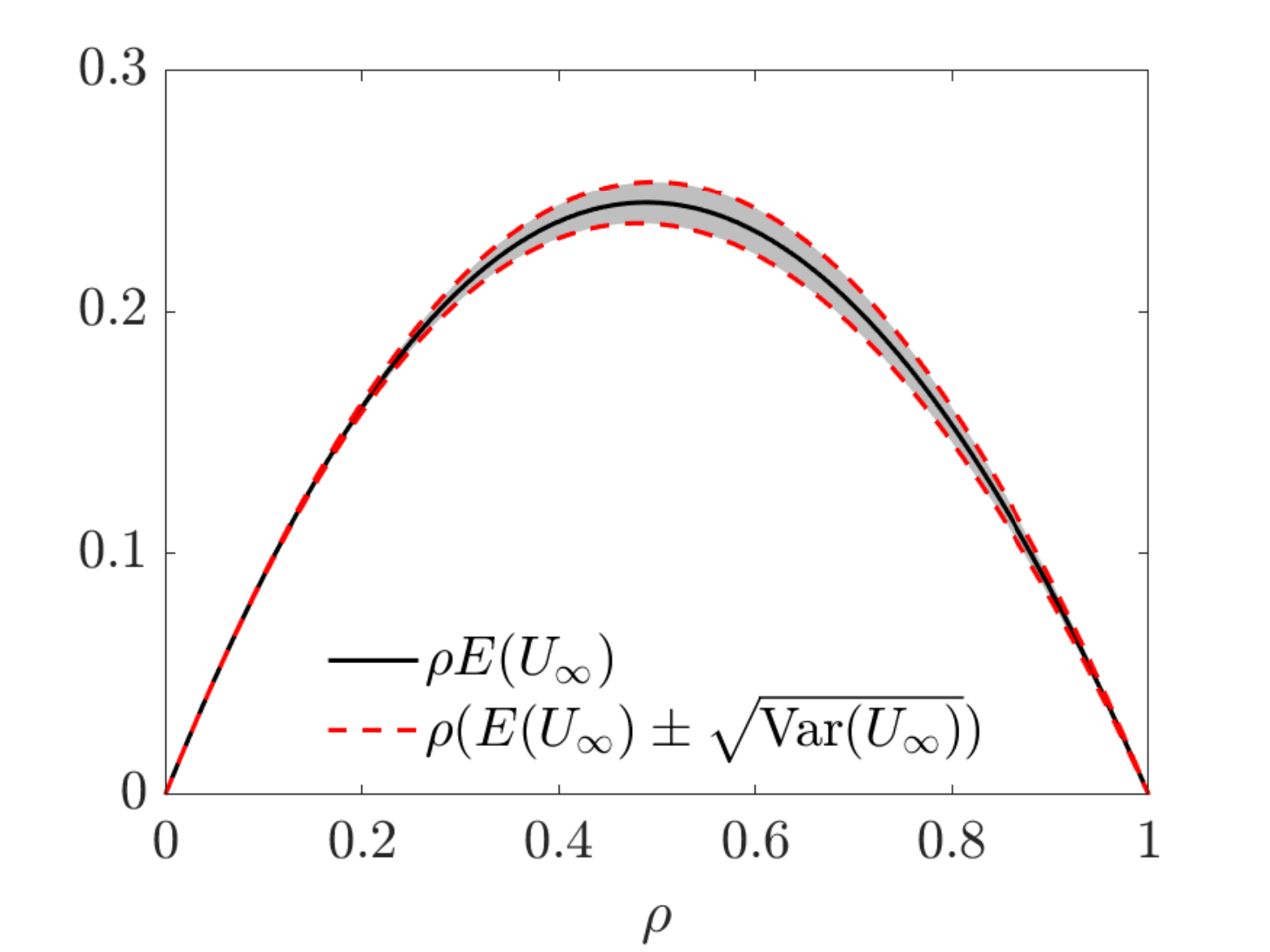}} \\
\subfigure[$z\sim\cU({[1,\,3]})$, $p^\ast=1$]{\includegraphics[scale=0.4]{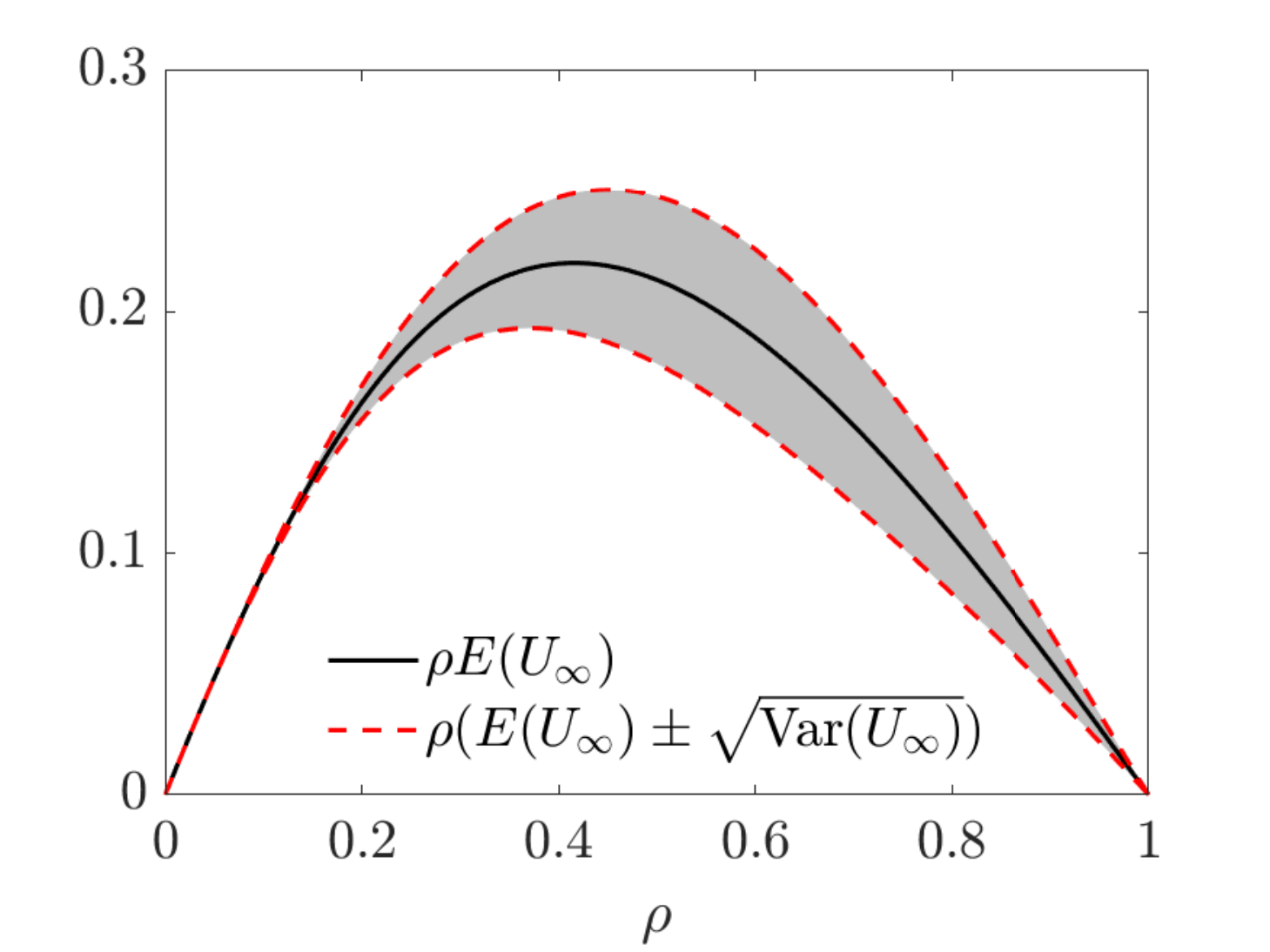}}
\subfigure[$z\sim\cU({[1,\,3]})$, $p^\ast=10$]{\includegraphics[scale=0.4]{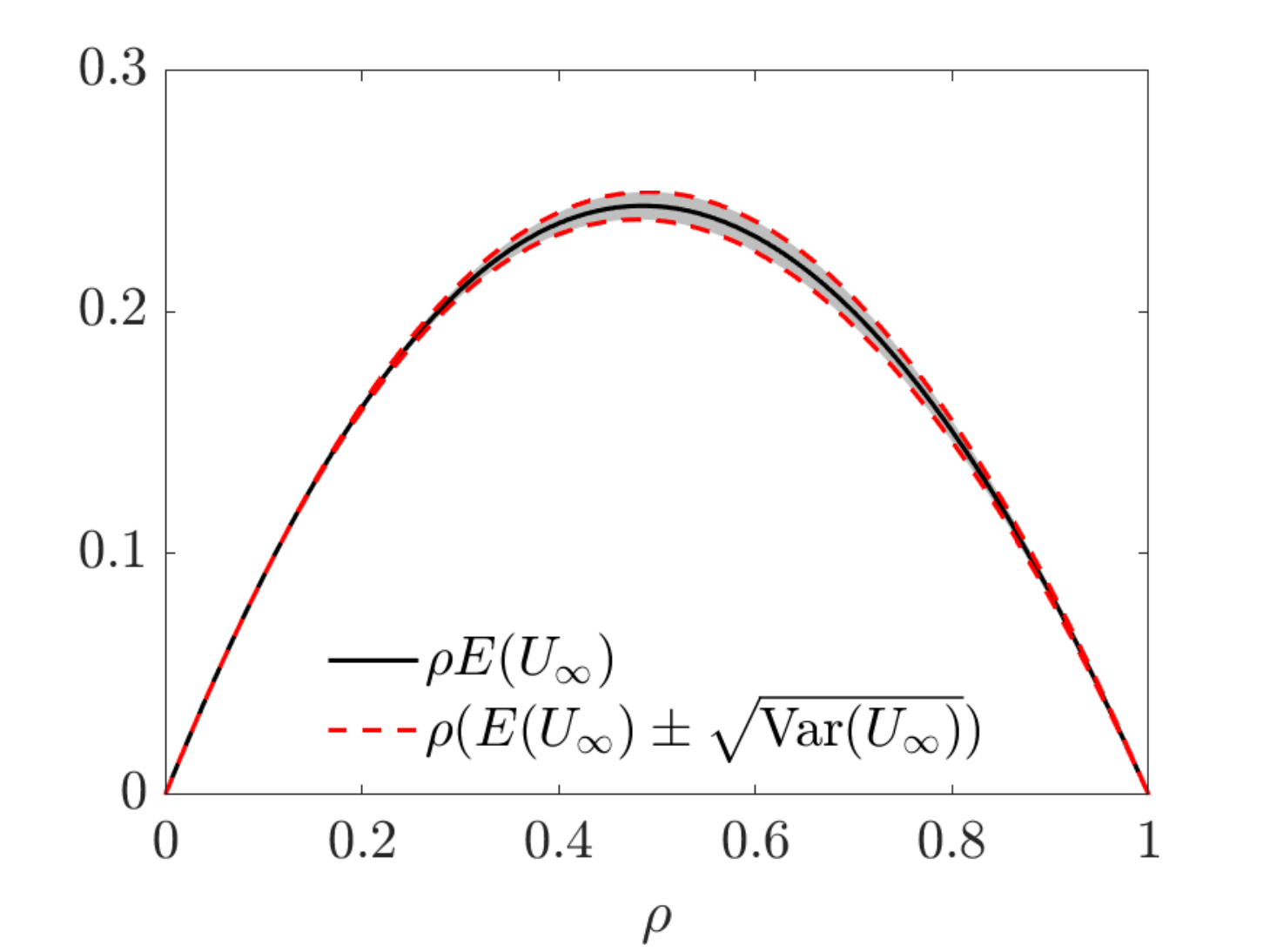}}
\caption{Action of the uncertain control~\eqref{eq:ustar.1} on the fundamental diagram and its scattering for two possible distributions of the uncertain parameter $z$ and two values of the effective penetration rate $p^\ast$. We considered $v_d(\rho)=1-\rho$ as optimal speed. For uniformly distributed $z$ (bottom row), a comparison is possible with Figure~\ref{fig:funddiag.Z_unif}, which illustrates the fundamental diagram and its scattering in the uncontrolled case.}
\label{fig:damping_diag}
\end{figure}

In Figure~\ref{fig:damping_diag}, we show the fundamental diagram $\rho\bar{V}_\infty(\rho)$, cf.~\eqref{eq:Vinf.1}, and its standard deviation $\rho\varsigma_\infty(\rho)$, cf.~\eqref{eq:varsigmainf}, for both a discrete (binomial) and a continuous (uniform) probability distribution of the uncertain parameter $z$ and increasing values of the effective penetration rate $p^\ast$. The curves have been obtained numerically by means of suitable quadrature formulas in $z$, starting from the analytical expression~\eqref{eq:Vinf.1}. As predicted by~\eqref{eq:varsigminf.est}, and more in general by Theorem~\ref{theo:I}, the results display a clear damping of the uncertainty at the macroscopic level, thanks to the action of the microscopic control. This is particularly evident in Figure~\ref{fig:damping_diag}(c)-(d), which may be compared with Figure~\ref{fig:funddiag.Z_unif} illustrating the related uncontrolled case.

\subsubsection{Fokker-Planck description and equilibria}
\label{sect:FP.1}
\begin{figure}[!t]
\centering
\subfigure[$z-1\sim\B\!\left(50,\,\frac{1}{50}\right)$, $p^\ast=1$]{\includegraphics[scale=0.4]{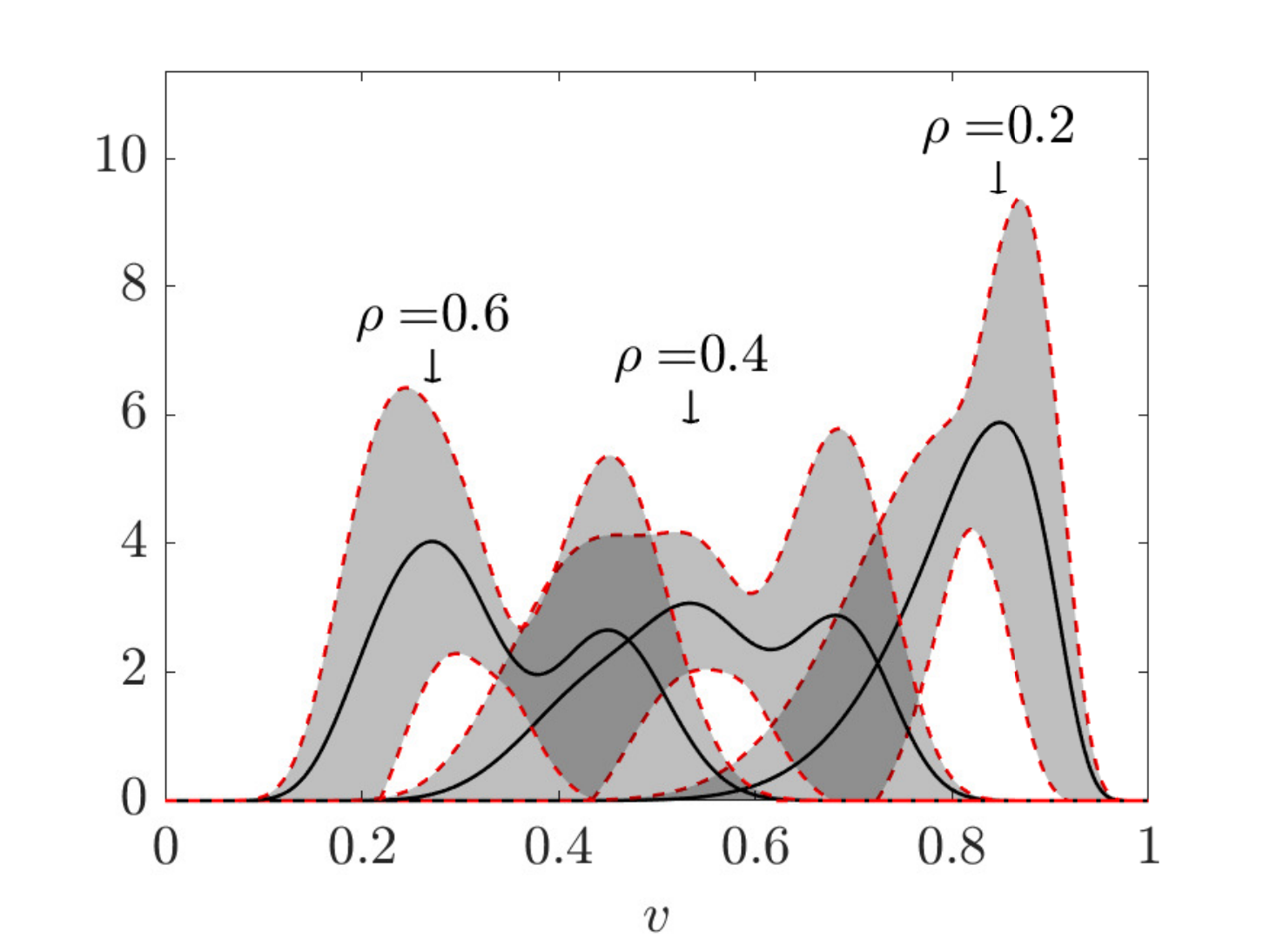}}
\subfigure[$z-1\sim\B\!\left(50,\,\frac{1}{50}\right)$, $p^\ast=10$]{\includegraphics[scale=0.4]{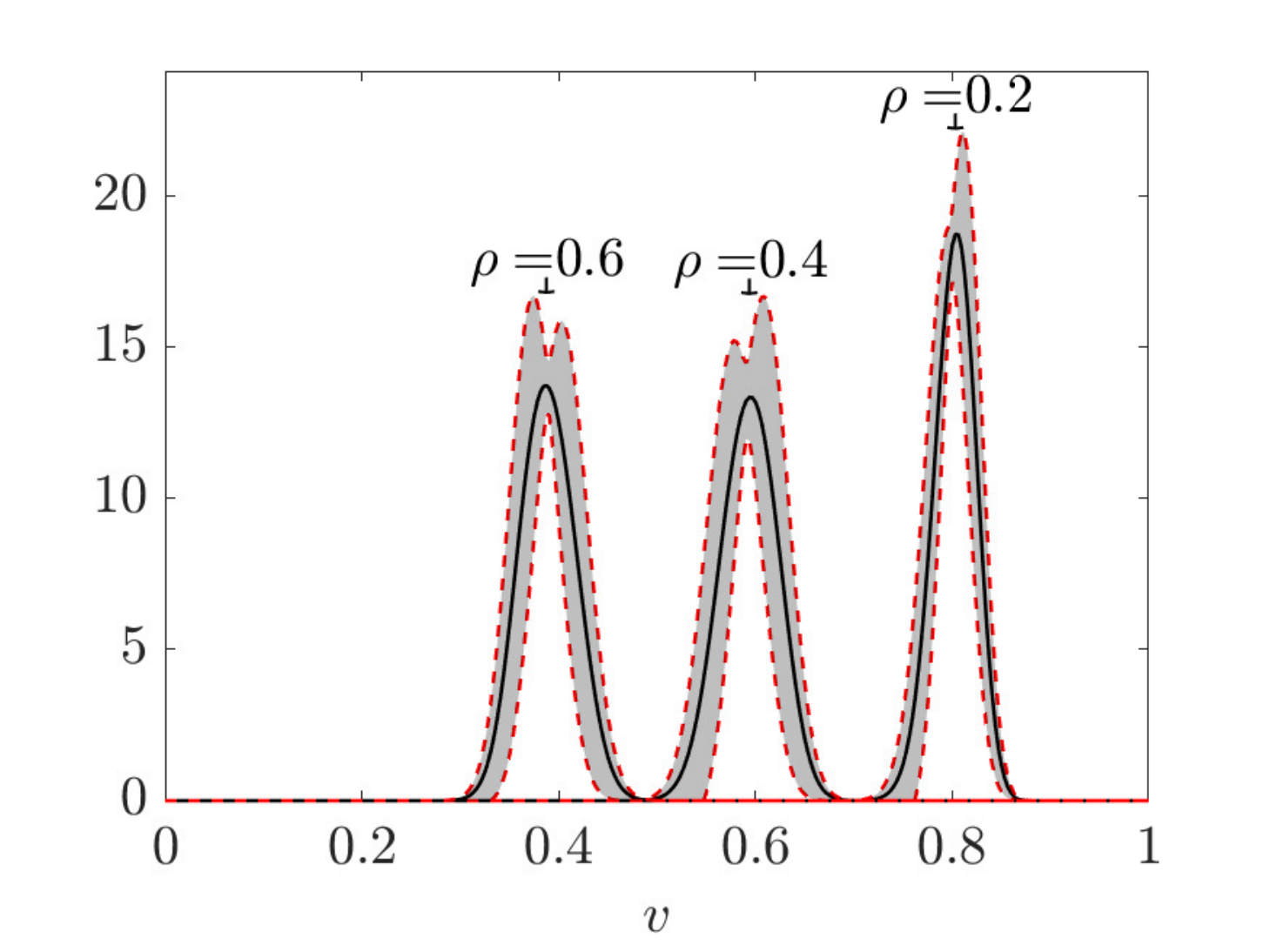}} \\
\subfigure[$z\sim\cU({[1,\,3]})$, $p^\ast=1$]{\includegraphics[scale=0.4]{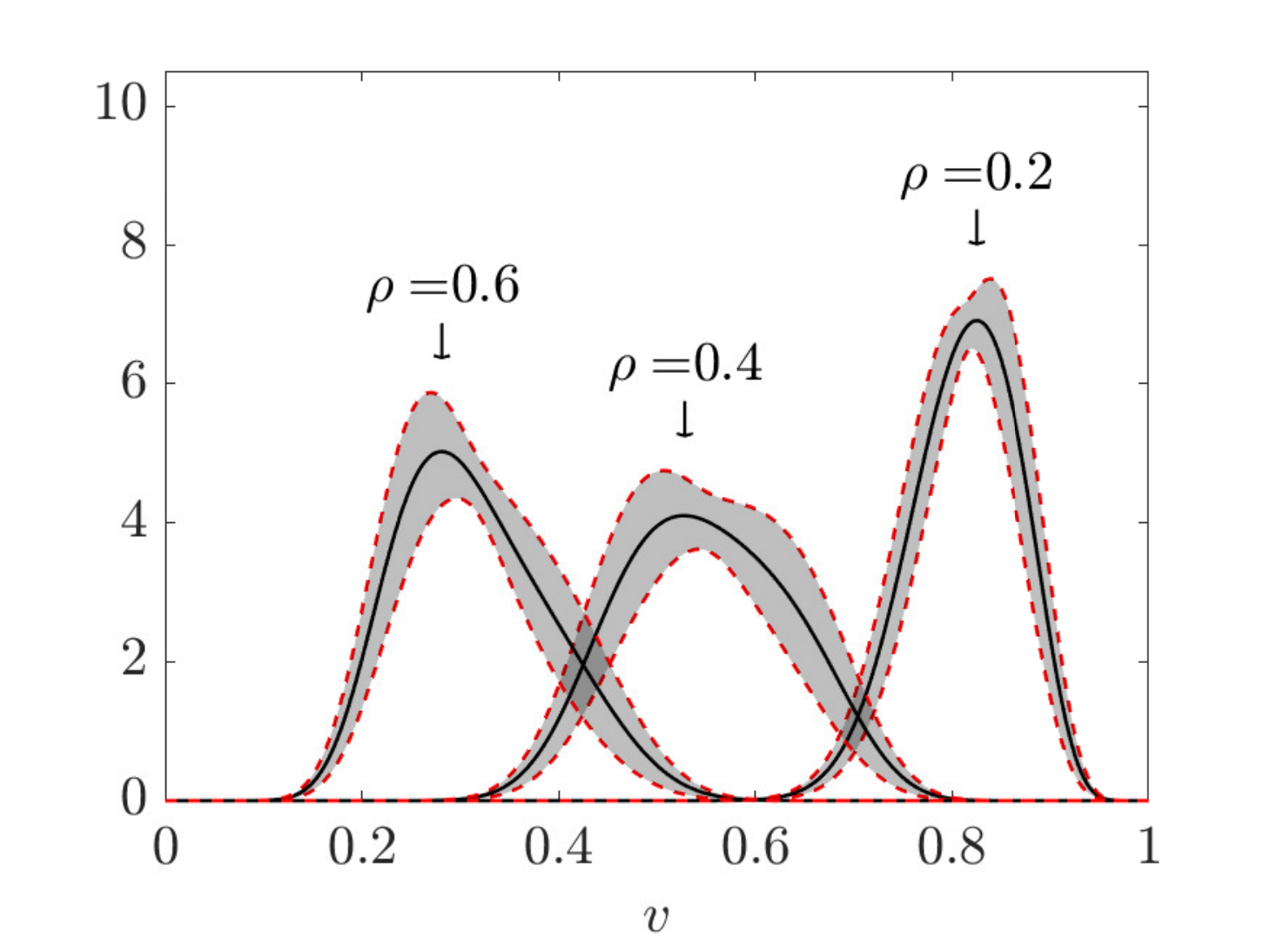}}
\subfigure[$z\sim\cU({[1,\,3]})$, $p^\ast=10$]{\includegraphics[scale=0.4]{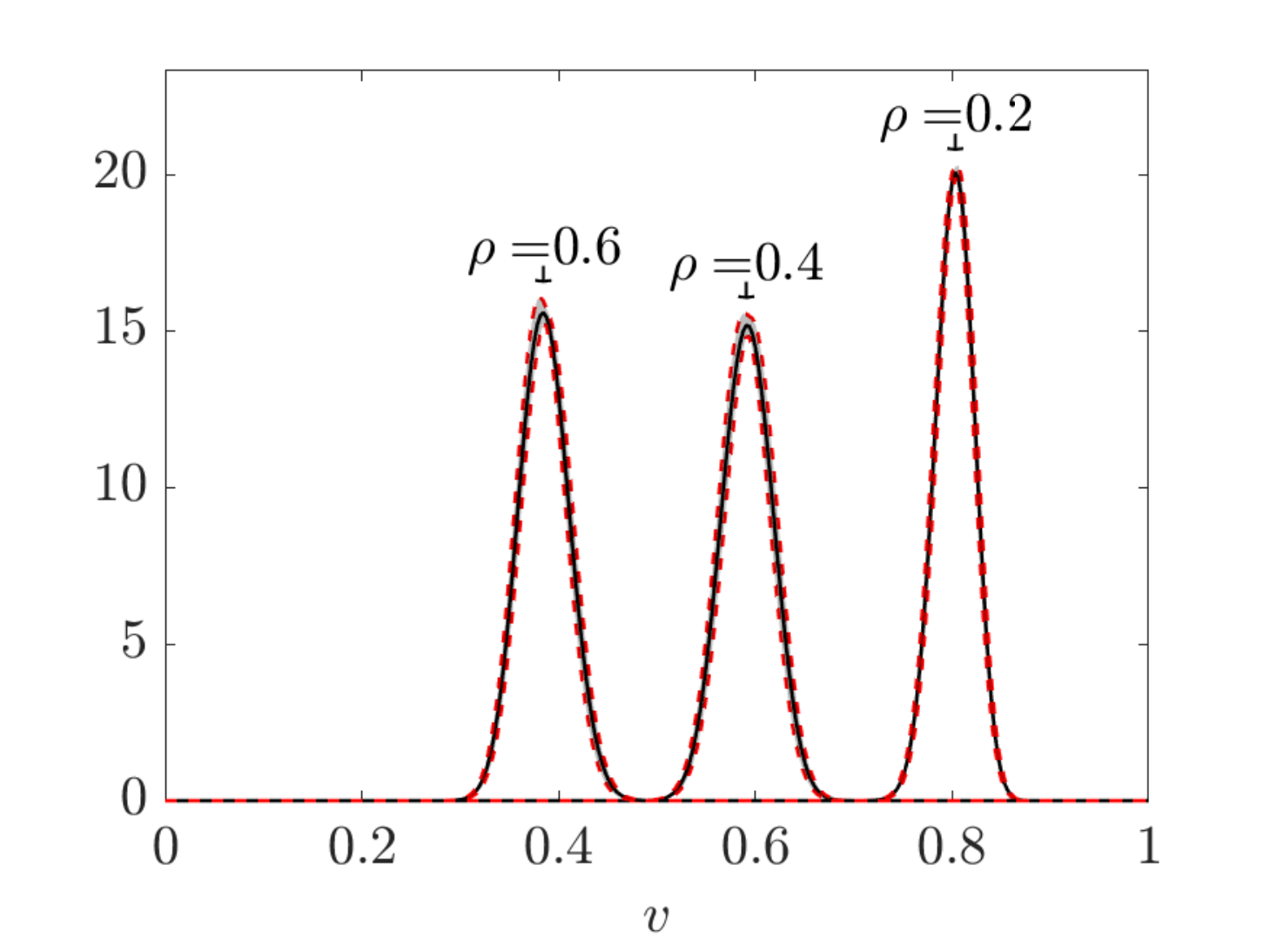}}
\caption{Representation of $\bar{f}_\infty(v)$ (solid lines) and of its $z$-standard deviation $\pm\sqrt{\Var_z(f_\infty(v;\,z))}$ (dashed lines) for different distributions of the uncertain parameter $z$, different effective penetration rates $p^\ast$ of the ADAS technology and various traffic densities. In (a)-(b), $z$ is such that $z-1$ has a binomial distribution. In (c)-(d), $z$ is uniformly distributed in $[1,\,3]$. We considered the optimal speed $v_d(\rho)=1-\rho$ and $\lambda=5\cdot 10^{-2}$ in~\eqref{eq:finf.1}.}
\label{fig:dis_pstar}
\end{figure}

Similarly to Section~\ref{sect:FP}, we may obtain more details on the large time trend emerging from the controlled interactions~\eqref{eq:binary-u.1} by studying the quasi-invariant limit of the full Boltzmann-type equation~\eqref{eq:Boltz.Theta}. We recall that, in this case, this amounts to considering the quasi-invariant scaling~\eqref{eq:scaling}-\eqref{eq:scaling.nu}. Proceeding like in Section~\ref{sect:FP}, we determine the  following Fokker-Planck equation with non-constant coefficients:
\begin{equation}
	\partial_tf=\frac{\lambda}{2}\partial_v^2\left(D^2(v)f\right)-\partial_v\Bigl[\Bigl(P(1+(1-P)V(t;\,z))+p^\ast v_d-(1+p^\ast)v\Bigr)f\Bigr],
	\label{eq:FP.1}
\end{equation}
where $V$ is the solution to~\eqref{eq:V.1}, whence, at the steady state,
$$ \frac{\lambda}{2}\partial_v\left(D^2(v)f_\infty\right)-(1+p^\ast)(V_\infty-v)f_\infty=0 $$
with $V_\infty$ given by~\eqref{eq:Vinf.1}. The unique distribution with unitary mass solving this equation for $D$ like in~\eqref{eq:D} is
\begin{equation}
	f_\infty(v;\,z)=\frac{v^{\frac{2(1+p^\ast)V_\infty(\rho;\,z)}{\lambda}-1}(1-v)^{\frac{2(1+p^\ast)(1-V_\infty(\rho;\,z))}{\lambda}-1}}
		{\operatorname{Beta}\!\left(\frac{2(1+p^\ast)V_\infty(\rho;\,z)}{\lambda},\,\frac{2(1+p^\ast)(1-V_\infty(\rho;\,z))}{\lambda}\right)},
	\label{eq:finf.1}
\end{equation}
namely again a beta probability density function parametrised by $z$ and incorporating the action of the microscopic control through the parameter $p^\ast$.

In Figure~\ref{fig:dis_pstar}, we show the mean equilibrium distribution $\bar{f}_\infty$, together with the uncertainty bands produced by its standard deviation $\pm\sqrt{\Var_z(f_\infty(v;\,z))}$, for a discrete (binomial) and a continuous (uniform) probability distribution of $z$ and increasing values of the effective penetration rate $p^\ast$. The curves have been obtained numerically, like in the previous cases, by means of suitable quadrature formulas in $z$, starting from the analytical expression~\eqref{eq:finf.1}. A comparison of Figure~\ref{fig:dis_pstar}(a)-(b) with Figure~\ref{fig:ginf.binom} and of Figure~\ref{fig:dis_pstar}(c)-(d) with Figure~\ref{fig:ginf.unif-gamma}(a) clearly shows that, by increasing $p^\ast$, the mean value of the expected equilibrium distribution $\bar{f}_\infty$ tends to the desired speed $v_d(\rho)$, cf.~\eqref{eq:dist.Vinf-v0}, and, moreover, that the uncertainty affecting $\bar{f}_\infty$ is globally dampened with respect to the uncontrolled case.

\subsection{Uncertainty control on average}
\label{sect:average}
The instantaneous optimal control~\eqref{eq:ustar.1} has a theoretical interest, as confirmed by Theorem~\ref{theo:I}, but it may not be straightforwardly implementable in practice. Indeed, it is an \textit{uncertain} control, because it depends on the specific value of the uncertain parameter $z$. In this section, we focus on the construction of a \textit{deterministic} control, i.e. one independent of $z$, and we prove that it actually performs analogously to the uncertain one as far as the damping of the $z$-uncertainty is concerned.

To this purpose, we consider the following new functional to be minimised:
$$ J_2(v',\,u):=\E_z(J_1(v',\,u;\,z))=\frac{1}{2}\E_z\ave{{\left(v_d(\rho)-v'\right)}^2+\nu u^2} $$
and we form the corresponding Lagrangian:
$$ \cL_2(v',\,u,\,\mu)=J_2(v',\,u)+\mu\E_z\ave{v'-v-\gamma\left(I(v,\,v_\ast;\,z)+\Theta u\right)-D(v)\eta}. $$
The main difference with respect to the case considered in Section~\ref{sect:pointwise} is that here we perform an optimisation \textit{on average}, rather than pointwise, in $z$. Consequently, we also require that the constraint given by the interaction dynamics~\eqref{eq:binary-u.gen} be satisfied on average, as it is clear from the expression of $\cL_2$ above.

The optimality conditions take now the form:
\begin{equation*}
	\begin{cases}
		\partial_{v'}\cL_2=\mu-\E_z\ave{v_d(\rho)-v'}=0 \\
		\partial_u\cL_2=\nu u-\gamma\Theta\mu=0 \\
		\partial_\mu\cL_2=\E_z\ave{v'-v-\gamma\left(I(v,\,v_\ast;\,z)+\Theta u\right)-D(v)\eta}=0,
	\end{cases}
\end{equation*}
whence the instantaneous optimal control turns out to be
\begin{equation}
	u^\ast(v,\,v_\ast)=\frac{\gamma\Theta}{\nu+\gamma^2\Theta^2}\Bigl[v_d(\rho)-v-\gamma\E_z(I(v,\,v_\ast;\,z))\Bigr].
	\label{eq:ustar.2}
\end{equation}
We observe that this control coincides with the $z$-average of the uncertain control~\eqref{eq:ustar.1}. Moreover, also~\eqref{eq:ustar.2} is a feedback control.

Plugging into~\eqref{eq:binary-u.gen}, we obtain the interaction rules
\begin{align}
	\begin{aligned}[c]
		v' &= v+\gamma\left(I(v,\,v_\ast;\,z)-\frac{\gamma^2\Theta^2}{\nu+\gamma^2\Theta^2}\E_z(I(v,\,v_\ast;\,z))\right)
			+\frac{\gamma^2\Theta^2}{\nu+\gamma^2\Theta^2}\left(v_d(\rho)-v\right)+D(v)\eta \\
		v_\ast' &= v_\ast.
	\end{aligned}
	\label{eq:binary-u.2}
\end{align}

In this case, it is not as straightforward as before to individuate sufficient conditions on $\eta$ and $D(v)$, like~\eqref{eq:eta.D} or~\eqref{eq:eta.D-contr.1}, which guarantee the physical admissibility of these interactions. The reason is that, when attempting to ensure $v'\in [0,\,1]$ by bounding the random fluctuation and the diffusion coefficient, it is not easy to compare effectively $I(v,\,v_\ast;\,z)$ and $\E_z(I(v,\,v_\ast;\,z))$ for a possibly generic probability distribution of $z$, not even if we refer to the specific interaction function~\eqref{eq:I}. Consequently, we have to admit that, in general, some interactions~\eqref{eq:binary-u.2} may not be physically admissible (i.e., they may produce either $v'<0$ or $v'>1$) and need therefore to be discarded from the statistical description of the system.

\subsubsection{Boltzmann-type description with cut-off}
\label{sect:Boltzmann.2}
In order to discard the possibly unphysical microscopic interactions~\eqref{eq:binary-u.2}, we consider the following Boltzmann-type equation \textit{with cut-off}:
\begin{align}
	\begin{aligned}[b]
		\frac{d}{dt}\int_0^1\varphi(v)f(t,\,v;\,z)\,dv &= \frac{1}{2\tau}\E_\Theta\left[\int_0^1\int_0^1
			\ave*{\chi(0\leq v'\leq 1)\bigl(\varphi(v')-\varphi(v)\bigr)}\right. \\
		&\phantom{=} \left.\phantom{\int_0^1}\times f(t,\,v;\,z)f(t,\,v_\ast;\,z)\,dv\,dv_\ast\right],
	\end{aligned}
	\label{eq:Boltz.B}
\end{align}
which differs from~\eqref{eq:Boltz.Theta} because of the non-constant \textit{interaction kernel}
$$ \chi(0\leq v'\leq 1)=
	\begin{cases}
		1 & \text{if } v'\in [0,\,1] \\
		0 & \text{otherwise}.
	\end{cases} $$
This makes the study of~\eqref{eq:Boltz.B} more challenging than in the previous cases, because it introduces additional non-linearities. For instance, the trend of the mean speed $V=V(t;\,z)$ is given, for $\varphi(v)=v$, by the equation
$$ \frac{dV}{dt}=\frac{1}{2\tau}\E_\Theta\left[\int_0^1\int_0^1\ave*{\chi(0\leq v'\leq 1)(v'-v)}f(t,\,v;\,z)f(t,\,v_\ast;\,z)\,dv\,dv_\ast\right], $$
whose right-hand side can no longer be easily written in terms of $V$ itself, like in~\eqref{eq:V} or~\eqref{eq:V.1}. To tackle this case, we may fruitfully work in the quasi-invariant regime, taking inspiration from the technique presented in~\cite{cordier2005JSP}. In particular, we will show in the next Section~\ref{sect:FP.non-maxw} that the Fokker-Planck equation originating from~\eqref{eq:Boltz.B} in such a regime is still~\eqref{eq:FP.1}, which implies in particular that the properties of the equilibria remain unchanged with respect to the case of the pointwise uncertain control discussed in Section~\ref{sect:pointwise}.

\subsubsection{Fokker-Planck asymptotics}
\label{sect:FP.non-maxw}
Similarly to~\eqref{eq:Boltz.eps}, under the quasi-invariant scaling~\eqref{eq:scaling}-\eqref{eq:scaling.nu} equation~\eqref{eq:Boltz.B} may be rewritten as
\begin{align*}
	\frac{d}{dt}\int_0^1\varphi(v)f_\epsilon(t,\,v;\,z)\,dv &= \frac{1}{\epsilon}\E_\Theta\left[\int_0^1\int_0^1
		\ave*{\chi(0\leq v'\leq 1)\bigl(\varphi(v')-\varphi(v)\bigr)}\right. \\
	&\phantom{=} \left.\phantom{\int_0^1}\times f_\epsilon(t,\,v;\,z)f_\epsilon(t,\,v_\ast;\,z)\,dv\,dv_\ast\right] \\
	&= \frac{1}{\epsilon}\E_\Theta\left[\int_0^1\int_0^1\ave{\varphi(v')-\varphi(v)}f_\epsilon(t,\,v;\,z)f_\epsilon(t,\,v_\ast;\,z)\,dv\,dv_\ast\right] \\
	&\phantom{=} -\frac{1}{\epsilon}\E_\Theta\left[\int_0^1\int_0^1\ave{(1-\chi(0\leq v'\leq 1))(\varphi(v')-\varphi(v))}\right. \\
	&\phantom{=} \left.\phantom{\int_0^1}\times f_\epsilon(t,\,v;\,z)f_\epsilon(t,\,v_\ast;\,z)\,dv\,dv_\ast\right] \\
	&=: A_\epsilon(f_\epsilon,\,f_\epsilon)[\varphi](t;\,z)+R_\epsilon(f_\epsilon,\,f_\epsilon)[\varphi](t;\,z).
\end{align*}
We make the same assumptions as in Section~\ref{sect:FP} concerning the existence and uniqueness of the solution to this equation for every $\epsilon>0$ and the convergence, up to subsequences, of $f_\epsilon(\cdot,\,\cdot;\,z)$ to some $f(\cdot,\,\cdot;\,z)\in C(\R_+;\,L^1(0,\,1))$ as $\epsilon\to 0^+$. Moreover, from now on we take $\varphi\in C^\infty_c(0,\,1)$ and, inspired by~\eqref{eq:eta.D} and~\eqref{eq:eta.D-contr.1}, we assume $cD(v)\leq\min\{v,\,1-v\}$ for some $c>0$ along with $D(v)>0$ for all $v\in (0,\,1)$.

In the quasi-invariant limit, the term $A_\epsilon(f_\epsilon,\,f_\epsilon)[\varphi]$ produces the right-hand side, in weak form, of the Fokker-Planck equation~\eqref{eq:FP.1}. Indeed, the additional term with $\E_z(I(v,\,v_\ast;\,z))$ contained in the interaction rules~\eqref{eq:binary-u.2}, which is absent in~\eqref{eq:binary-u.1}, can be easily shown to be of order $\epsilon$, hence infinitesimal in the limit.

Therefore, the goal is to show that the remainder $R_\epsilon(f_\epsilon,\,f_\epsilon)[\varphi]$ vanishes for $\epsilon\to 0^+$. To this purposes, we observe preliminarily that
\begin{align}
	\begin{aligned}[b]
		\abs{R_\epsilon(f_\epsilon,\,f_\epsilon)[\varphi](t;\,z)} &\leq \frac{1}{\epsilon}\E_\Theta\left[\int_0^1\int_0^1\ave{(1-\chi(0\leq v'\leq 1))\abs{\varphi(v')-\varphi(v)}}\right. \\
		&\phantom{=} \left.\phantom{\int_0^1}\times f_\epsilon(t,\,v;\,z)f_\epsilon(t,\,v_\ast;\,z)\,dv\,dv_\ast\right]
	\end{aligned}
	\label{eq:R}
\end{align}
and that, in view of the scaling~\eqref{eq:scaling} of $\sigma^2$, we may represent $\eta=\sqrt{\lambda\epsilon}Y$, where $Y$ is a random variable with $\ave{Y}=0$ and $\ave{Y^2}=1$. We assume moreover that $\eta$, hence also $Y$, has finite third order moment, i.e. $\ave{\abs{Y}^3}<+\infty$. Then, upon scaling the other parameters of the interactions~\eqref{eq:binary-u.2} according to~\eqref{eq:scaling}-\eqref{eq:scaling.nu} and expanding $\varphi(v')$ in Taylor's series about $v$, we estimate\footnote{Here and henceforth we use the notation $a\lesssim b$ to mean that there exists a constant $C>0$, independent of $\epsilon$ and whose specific value is unimportant, such that $a\leq Cb$.}
\begin{equation}
	\abs{\varphi(v')-\varphi(v)}\lesssim\sum_{i=0}^{3}\epsilon^{i/2}c_i(v,\,v')\abs{Y}^i,
	\label{eq:phi.Taylor}
\end{equation}
where $\overline{D}:=\sup_{v\in [0,\,1]}D(v)$ and moreover
\begin{align*}
	c_0(v,\,v') &:= \alpha\abs{\varphi'(v)}+\frac{\alpha^2}{2}\abs{\varphi''(v)}+\frac{\alpha^3}{6}\abs{\varphi'''(\bar{v})} \\
	c_1(v,\,v') &:= \abs{\varphi'(v)}+\alpha\abs{\varphi''(v)}+\frac{\alpha^2}{2}\abs{\varphi'''(\bar{v})} \\
	c_2(v,\,v') &:= \frac{1}{2}\abs{\varphi''(v)}+\frac{\alpha}{2}\abs{\varphi'''(\bar{v})} \\
	c_3(v,\,v') &:= \frac{1}{6}\abs{\varphi'''(\bar{v})}
\end{align*}
with
$$ \bar{v}:=\vartheta v+(1-\vartheta)v'\ \text{for some } \vartheta\in [0,\,1], \qquad \alpha:=\epsilon\left(1+\frac{1+\epsilon}{\kappa}\right). $$
Plugging~\eqref{eq:phi.Taylor} into~\eqref{eq:R} and considering that $1-\chi(0\leq v'\leq 1)\leq 1$ and that, for $\epsilon$ sufficiently small, it results $\alpha\lesssim\epsilon$, we compute:
\begin{align}
	\begin{aligned}[b]
		\abs{R_\epsilon(f_\epsilon,\,f_\epsilon)[\varphi](t;\,z)} &\lesssim \int_0^1\ave{1-\chi(0\leq v'\leq 1)}\abs{\varphi'(v)}f_\epsilon(t,\,v;\,z)\,dv \\
		&\phantom{\lesssim} +\frac{1}{\sqrt{\epsilon}}\int_0^1\ave{(1-\chi(0\leq v'\leq 1))\abs{Y}}\abs{\varphi'(v)}f_\epsilon(t,\,v;\,z)\,dv \\
		&\phantom{\lesssim} +\int_0^1\ave{(1-\chi(0\leq v'\leq 1))Y^2}\abs{\varphi''(v)}f_\epsilon(t,\,v;\,z)\,dv \\
		&\phantom{\lesssim} \left(\norm{\varphi''}\ave{\abs{Y}}+\norm{\varphi'''}\ave{\abs{Y}^3}\right)\sqrt{\epsilon}
			+\left(\norm{\varphi''}+\norm{\varphi'''}\right)\epsilon \\
		&\phantom{\lesssim} +\norm{\varphi'''}\ave{\abs{Y}}\epsilon^{\frac{3}{2}}+\norm{\varphi'''}\epsilon^2
	\end{aligned}
	\label{eq:R.first_est}
\end{align}
The last six terms on the right-hand side are bounded and tend to zero as $\epsilon\to 0^+$. Conversely, the first three terms need to be estimated more carefully.

From the interactions~\eqref{eq:binary-u.2} scaled according to~\eqref{eq:scaling}-\eqref{eq:scaling.nu}, we see that a sufficient condition for $v'\geq 0$ is
$$ Y\geq \frac{1}{\sqrt{\lambda\epsilon}}\left[\epsilon\left(1+\frac{1}{\kappa}\right)-1\right]\frac{v}{D(v)}+\frac{2\epsilon\sqrt{\epsilon}}{\sqrt{\lambda}\kappa D(v)}. $$
Let us take $\epsilon<\frac{\kappa}{\kappa+1}$, so that $\epsilon\left(1+\frac{1}{\kappa}\right)-1<0$. Then, invoking the assumptions~\eqref{eq:eta.D},~\eqref{eq:eta.D-contr.1} on $D$, a further sufficient condition for $v'\geq 0$ is
\begin{equation}
	Y\geq \frac{c}{\sqrt{\lambda\epsilon}}\left[\epsilon\left(1+\frac{1}{\kappa}\right)-1\right]+\frac{2\epsilon\sqrt{\epsilon}}{\sqrt{\lambda}\kappa D(v)}.
	\label{eq:Y.below}
\end{equation}
If we restrict $v$ to any compact subset $U$ of $(0,\,1)$ then there exists a constant $\underline{D}>0$ such that $D(v)\geq\underline{D}$ for all $v\in U$. Consequently, the right-hand side of~\eqref{eq:Y.below} is in $U$ a continuous function of $\epsilon$, which tends to $-\infty$ when $\epsilon\to 0^+$. Hence, it is negative for $\epsilon$ small enough, say $\epsilon<\epsilon_0$. On the whole, we may define
$$ b_\epsilon(v):=\frac{c}{\sqrt{\lambda\epsilon}}\left[1-\epsilon\left(1+\frac{1}{\kappa}\right)\right]-\frac{2\epsilon\sqrt{\epsilon}}{\sqrt{\lambda}\kappa D(v)}>0,
	\qquad (v\in U,\,\epsilon<\epsilon_0) $$
and state that, for $v\in U$, a sufficient condition for $v'\geq 0$ is $Y\geq -b_\epsilon(v)$ provided $\epsilon<\min\{\frac{\kappa}{\kappa+1},\,\epsilon_0\}$. Arguing analogously, we deduce that, for $v\in U$, a sufficient condition for $v'\leq 1$ is $Y\leq b_\epsilon(v)$ with the same restriction on $\epsilon$ as above, so that finally
$$ \abs{Y}\leq b_\epsilon(v) \quad \Rightarrow \quad 0\leq v'\leq 1 \qquad \left(v\in U,\,\epsilon<\min\left\{\frac{\kappa}{\kappa+1},\,\epsilon_0\right\}\right). $$
It follows
$$ 1-\chi(0\leq v'\leq 1)\leq 1-\chi(\abs{Y}\leq b_\epsilon(v))=\chi(\abs{Y}>b_\epsilon(v)), $$
whence, from~\eqref{eq:R.first_est} with $U:=\supp{\varphi'}\cup\supp{\varphi''}\subset\subset (0,\,1)$,
\begin{align*}
	\abs{R_\epsilon(f_\epsilon,\,f_\epsilon)[\varphi](t;\,z)} &\lesssim \int_U\ave{\chi(\abs{Y}>b_\epsilon(v))}\abs{\varphi'(v)}f_\epsilon(t,\,v;\,z)\,dv \\
	&\phantom{\lesssim} +\frac{1}{\sqrt{\epsilon}}\int_U\ave{\chi(\abs{Y}>b_\epsilon(v))\abs{Y}}\abs{\varphi'(v)}f_\epsilon(t,\,v;\,z)\,dv \\
	&\phantom{\lesssim} +\int_U\ave{\chi(\abs{Y}>b_\epsilon(v))Y^2}\abs{\varphi''(v)}f_\epsilon(t,\,v;\,z)\,dv
		+O(\sqrt{\epsilon}).
\end{align*}

Let now $p,\,q\in [1,\,+\infty]$ be such that $\frac{1}{p}+\frac{1}{q}=1$. H\"{o}lder's and Chebyshev's inequalities imply
$$ \ave{\chi(\abs{Y}>b_\epsilon(v))\abs{Y}^k}\leq\ave{\abs{Y}^{kp}}^{\frac{1}{p}}\P(\abs{Y}>b_\epsilon(v))^{\frac{1}{q}}
	\leq\frac{\ave{\abs{Y}^{kp}}^{\frac{1}{p}}}{b_\epsilon(v)^{\frac{2}{q}}}
		\lesssim\ave{\abs{Y}^{kp}}^{\frac{1}{p}}\epsilon^{\frac{1}{q}}, $$
where the last one holds for $\epsilon<\min\{\frac{\kappa}{\kappa+1},\,\epsilon_0\}$. Therefore, we may deduce:
\begin{align*}
	& \ave{\chi(\abs{Y}>b_\epsilon(v))}\lesssim\epsilon & (k=0,\,p=+\infty,\,q=1) \\
	& \ave{\chi(\abs{Y}>b_\epsilon(v))\abs{Y}}\lesssim\ave{\abs{Y}^{\frac{5}{2}}}^{\frac{2}{5}}\epsilon^{\frac{3}{5}} & \left(k=1,\,p=\frac{5}{2},\,q=\frac{5}{3}\right) \\
	& \ave{\chi(\abs{Y}>b_\epsilon(v))Y^2}\lesssim\ave{\abs{Y}^3}^{\frac{2}{3}}\epsilon^{\frac{1}{3}} & \left(k=2,\,p=\frac{3}{2},\,q=3\right),
\end{align*}
and finally
$$ \abs{R_\epsilon(f_\epsilon,\,f_\epsilon)[\varphi](t;\,z)}\lesssim
	\norm{\varphi'}\epsilon+\norm{\varphi'}\ave{\abs{Y}^{\frac{5}{2}}}^{\frac{2}{5}}\epsilon^{\frac{1}{10}}
		+\norm{\varphi''}\ave{\abs{Y}^3}^{\frac{2}{3}}\epsilon^{\frac{1}{3}}+O(\sqrt{\epsilon}), $$
which implies $R_\epsilon(f_\epsilon,\,f_\epsilon)[\varphi]\to 0$ for $\epsilon\to 0^+$. Notice, in particular, that $\ave{\abs{Y}^{\frac{5}{2}}}^{\frac{2}{5}}\leq\ave{\abs{Y}^3}^{\frac{1}{3}}<+\infty$ by H\"{o}lder's inequality again.

Consequently, as already anticipated, the large time aggregate effect of the deterministic control~\eqref{eq:ustar.2} is actually the same as that of the uncertain control~\eqref{eq:ustar.1}. Indeed, thanks to the result just obtained, the quasi-invariant limit for $\varphi(v)=v$ yields equation~\eqref{eq:V.1}, whence we deduce that the asymptotic mean speed is again~\eqref{eq:Vinf.1} and that all the conclusions drawn in Section~\ref{sect:pointwise}, in particular~\eqref{eq:dist.Vinf-v0},~\eqref{eq:varsigminf.est},~\eqref{eq:finf.1} and Theorem~\ref{theo:I}, hold true also in this case.

\begin{figure}[!t]
\centering
\subfigure[$\rho=0.2$]{\includegraphics[scale=0.4]{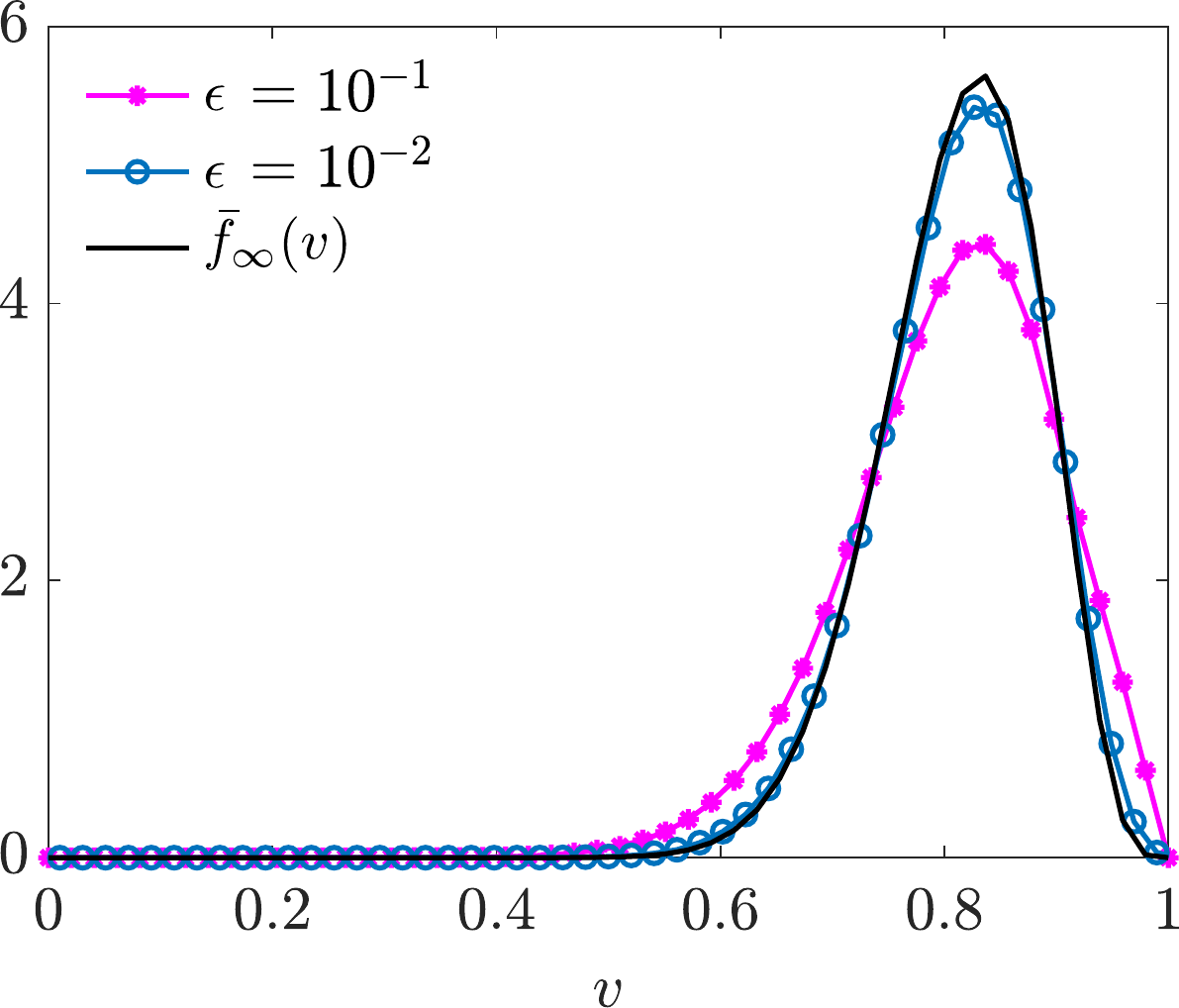}}
\subfigure[$\rho=0.4$]{\includegraphics[scale=0.4]{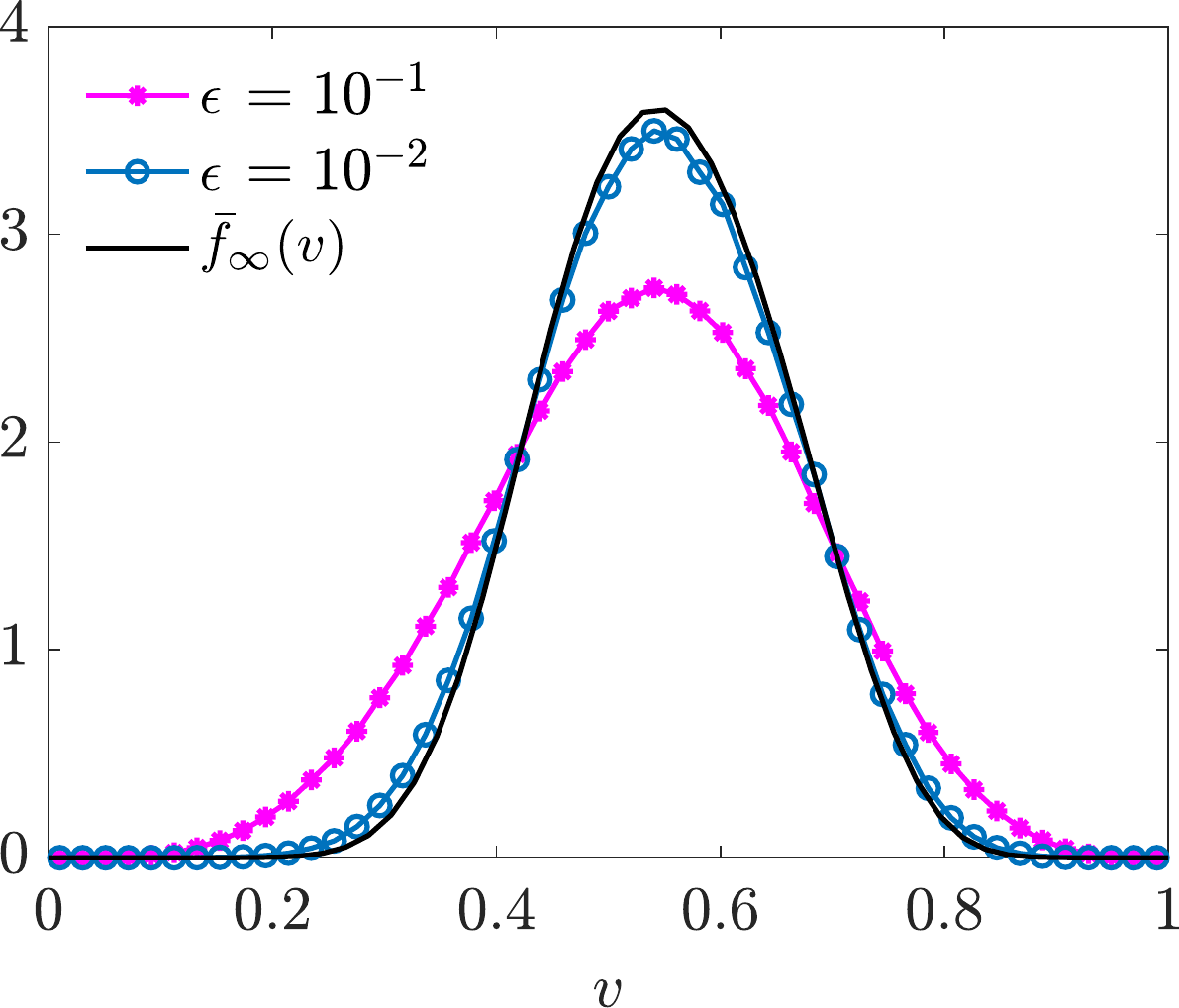}}
\subfigure[$\rho=0.6$]{\includegraphics[scale=0.4]{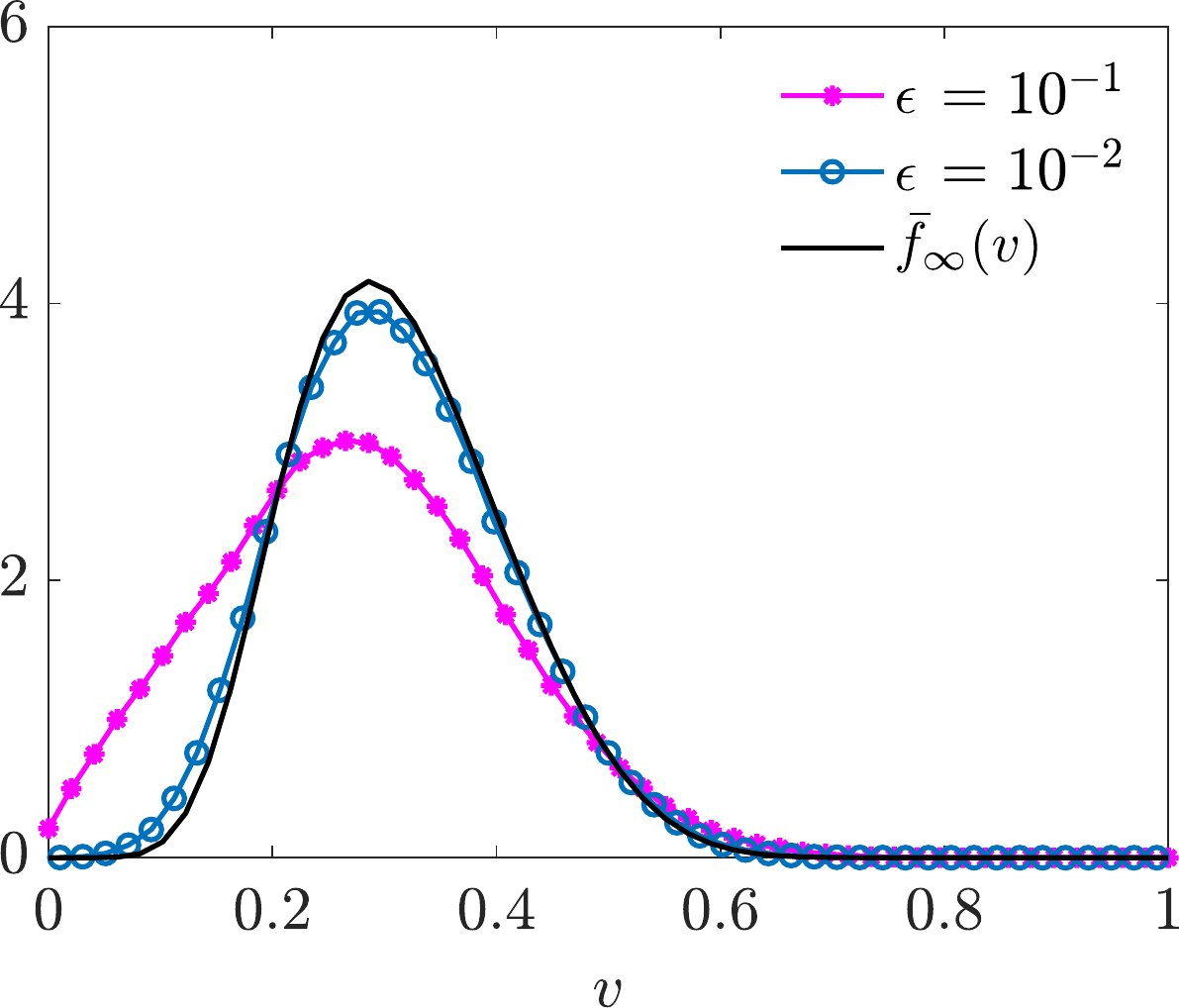}}
\caption{Comparison between the large time $z$-averaged solution to the Boltzmann-type equation~\eqref{eq:Boltz.B} with non-constant collision kernel (bulleted lines) and the average equilibrium solution to the Fokker-Planck equation~\eqref{eq:FP.1} (solid line) for decreasing $\epsilon$, mimicking the quasi-invariant limit $\epsilon\to 0^+$. The relevant parameters are $\rho=0.2,\,0.4,\,0.6$, $p^\ast=1$, $z\sim\cU([1,\,3])$.}
\label{fig:boltzmann_average}
\end{figure}

In Figure~\ref{fig:boltzmann_average}, we exemplify the validity of the theoretical result just obtained. We show the numerical approximation of the large time average solution $\bar{f}_\infty(v)$ to the Boltzmann-type equation~\eqref{eq:Boltz.B}, computed by means of a standard Monte Carlo scheme with rejection in the variable $v$ and a collocation method for uncertainty quantification in the variable $z$ in the case $z\sim\cU([1,\,3])$, cf.~\cite{pareschi2001ESAIMP,pareschi2013BOOK} and~\cite{Dimarco2017,tosin2018CMS}. By comparing, for decreasing $\epsilon$, such a solution with the $\bar{f}_\infty(v)$ obtained from~\eqref{eq:finf.1} by means of the Gauss-Legendre quadrature formula in $z$, we observe the expected consistency of the Boltzmann-type model~\eqref{eq:Boltz.B} with the Fokker-Planck model~\eqref{eq:FP.1} in the quasi-invariant limit.

\section{Numerical tests}
\label{sect:numerics}
In this section, we present several numerical tests, which confirm the theoretical findings of the previous sections. First, we focus on the Fokker-Planck-type modelling derived in Section~\ref{sect:FP} in the case of uncontrolled microscopic interactions. Next, we show the effectiveness of the control approaches proposed in Sections~\ref{sect:pointwise},~\ref{sect:average}.

\subsection{Generalised polynomial chaos for uncertainty quantification}
Numerical methods for Uncertainty Quantification (UQ) gained a lot of momentum in recent times, thanks to their spectral convergence on the random field obtained under suitable regularity assumptions~\cite{Xiu2010,Xiu2002}. The development of UQ methods for kinetic equations encoding uncertain quantities presents several difficulties, which are essentially due to the increased dimensionality and the intrinsic structural properties of the solution. Recently, a lot of efforts have been devoted to the study of kinetic equations with uncertainties, see e.g.~\cite{Dimarco2017,Hu2017,Zhu2017}. Despite their spectral accuracy, Stochastic Galerking-generalised Polynomial Chaos (SG-gPC) methods typically lead to the loss of physical properties of the solution, such as the non-negativity of the distribution function, the entropy dissipation and the conservation of the steady states. We refer the interested reader to~\cite{Carrillo2019,carrillo2019VJM,zanella2020MCS} for further insights into the preservation of the structural properties of the solutions to kinetic equations for the collective behaviour of multi-agent systems.

For completeness, in the following we report some fundamental notions on SG-gPC methods, in order to introduce the numerical scheme which we use for our numerical tests.

Let $(\Omega,\,\cF,\,\P)$ be a probability space, where $\Omega$ is an abstract sample space, $\cF$ a $\sigma$-algebra of subsets of $\Omega$ and $\P$ a probability measure on $\cF$. Let us introduce, furthermore, a random variable 
$$ z:(\Omega,\,\cF)\mapsto (I_Z,\,\cB(\R)), $$
where $I_Z\subseteq\R$ is the range of $z$ and $\cB(\R)$ is the Borel $\sigma$-algebra of subsets of $\R$. We assume that the distribution of $z$ is expressed by a known probability density function $\Psi:I_Z\to\R_+$.

\begin{table}[!t]
\caption{Choices of the $\Phi_k$'s depending on the distribution $\Psi$ of the uncertain parameter}
\begin{center}
\begin{tabular}{|c|c|c|}
\hline
Distribution $\Psi$ & Polynomials $\Phi_k$ & Support of the $\Phi_k$'s \\
\hline
\hline
Uniform & Legendre & Compact interval \\
Beta & Jacobi & Compact interval \\
Gamma & Laguerre & $\R_+$ \\
Binomial & Krawtchouk & $\N$ \\
\hline
\end{tabular}
\end{center}
\label{tab:pol}
\end{table}

Given $M\in\N$, the SG-gPC method is based on introducing the linear space of polynomials of degree at most $M$, generated by an orthonormal set $\{\Phi_k\}_{k=0}^{M}$ satisfying
$$ \E\bigl(\Phi_h(z)\Phi_k(z)\bigr)=\int_{I_Z}\Phi_h(z)\Phi_k(z)\Psi(z)\,dz=\delta_{hk}, $$
where
$$ \delta_{hk}=
	\begin{cases}
		1 & \text{if } h=k \\
		0 & \text{otherwise}
	\end{cases} $$
is the Kronecker delta. The choice of the $\Phi_k$'s depends clearly on the density function $\Psi$. This aspect has been extensively investigated in the UQ literature in relation with the so-called Wiener-Askey scheme for orthogonal polynomials, see~\cite{Xiu2010,Xiu2002}. In the following, we present some examples where the law of $z$ and the related set of orthonormal polynomials are among those reported in Table~\ref{tab:pol}.

Assume that the kinetic distribution function $f$ solving either Fokker-Planck equation~\eqref{eq:FP},~\eqref{eq:FP.1} is sufficiently smooth, in particular $f(t,\,v;\,\cdot)\in L^2(I_Z)$. Then, we can introduce the expansion
\begin{equation}
	f(t,\,v;\,z)\approx f^M(t,\,v;\,z):=\sum_{k=0}^M\hat{f}_k(t,\,v)\Phi_k(z),
	\label{eq:gPC_expansion}
\end{equation}
where $\hat{f}_k(t,\,v):=\E\bigl(f(t,\,v;\,z)\Phi_k(z)\bigr)$ for every $k=0,\,\dots,\,M$. Plugging the approximation $f^M$ of $f$ into the Fokker-Planck equation~\eqref{eq:FP}, we obtain
\begin{multline*}
	\partial_tf^M= \\
	\frac{\lambda}{2}\partial_v^2\left(D^2(v)f^M\right)-\partial_v\left\{\left[P(\rho;\,z)+\left(1+(1-P(\rho;\,z))\int_0^1 vf^M(t,\,v;\,z)\,dv\right)-v\right]f^M\right\}.
\end{multline*}
Representing $f^M$ by means of the expansion above, multiplying the equation by $\Phi_h$ and using the orthonormality of the set $\{\Phi_k\}_{k=0}^{M}$ yields the following coupled system of $M+1$ deterministic partial differential equations ruling the evolution of the projections $\{\hat{f}_h\}_{h=0}^{M}$ of $f$ onto the polynomial space:
\begin{equation}
	\partial_t\hat{f}_h=\frac{\lambda}{2}\partial_v^2(D^2(v)\hat{f}_h)
		-\partial_v\sum_{k=0}^{M}E_{hk}[f^M](t,\,v)\hat{f}_k, \qquad h=0,\,\dots,\,M,
	\label{eq:FP_h.uncontrolled}
\end{equation}
where
\begin{multline}
	E_{hk}[f^M](t,\,v):= \\
	\int_{I_Z}\left[P(\rho;\,z)+\left(1+(1-P(\rho;\,z))\int_0^1 vf^M(t,\,v;\,z)\,dv\right)-v\right]\Phi_h(z)\Phi_k(z)\Psi(z)\,dz.
	\label{eq:Ehk.uncontrolled}
\end{multline}

By introducing the vector $\hat{\vecf}(t,\,v):=(\hat{f}_0(t,\,v),\,\dots,\,\hat{f}_M(t,\,v))$ and the $(M+1)\times (M+1)$-matrix $\bE[\hat{\vecf}](t,\,v):=\{E_{hk}[f^M](t,\,v)\}_{h,k=0}^M$, whose entries are given by~\eqref{eq:Ehk.uncontrolled}, problem~\eqref{eq:FP_h.uncontrolled} may be recast in vector form. In particular, we obtain the following deterministic system of $M+1$ Fokker-Planck equations:
\begin{equation}
	\partial_t\hat{\vecf}=\frac{\lambda}{2}\partial_v^2(D^2(v)\hat{\vecf})-\partial_v(\bE[\hat{\vecf}]\hat{\vecf}).
	\label{eq:FP_vec.uncontrolled}
\end{equation}

In the case of controlled microscopic interactions, we apply the same procedure to the Fokker-Planck equation~\eqref{eq:FP.1} and, in place of~\eqref{eq:FP_h.uncontrolled}, we get:
\begin{equation}
	\partial_t\hat{f}_h=\frac{\lambda}{2}\partial_v^2(D^2(v)\hat{f}_h)
		-\partial_v\left[\sum_{k=0}^M\Bigl(E_{hk}[f^M](t,\,v)+p^\ast v_d(\rho)\delta_{hk}\Bigr)\hat{f}_k\right],
			\qquad h=0,\,\dots,\,M,
	\label{eq:FP_h.controlled}
\end{equation}
where now
\begin{multline}
	E_{hk}[f^M](t,\,v):= \\
	\int_{I_Z}\left[P(\rho;\,z)+\left(1+(1-P(\rho;\,z))\int_0^1 vf^M(t,\,v;\,z)\,dv\right)-(1+p^\ast)v\right]\Phi_h(z)\Phi_k(z)\Psi(z)\,dz.
	\label{eq:Ehk.controlled}
\end{multline}
Consequently, problem~\eqref{eq:FP_h.controlled} can be recast in vector form as
\begin{equation}
	\partial_t\hat{\vecf}=\frac{\lambda}{2}\partial_v^2(D^2(v)\hat{\vecf})-\partial_v((\bE[\hat{\vecf}]+p^\ast v_d(\rho)\mathbf{I})\hat{\vecf}),
	\label{eq:FP_vec.controlled}
\end{equation}
where $\bE[\hat{\vecf}](t,\,v):=\{E_{hk}[f^M](t,\,v)\}_{h,k=0}^M$ is the $(M+1)\times (M+1)$ matrix whose entries are given by~\eqref{eq:Ehk.controlled} and $\mathbf{I}$ is the $(M+1)\times (M+1)$ identity matrix.

In particular, for our purposes, it is of paramount importance to compute accurately the large time trend of the kinetic distribution function $f$. Aiming at preserving the steady state, we solve numerically~\eqref{eq:FP_vec.uncontrolled},~\eqref{eq:FP_vec.controlled} by means of a micro-macro gPC scheme, see~\cite{Dimarco2017,Pareschi2017}, together with a second order semi-implicit time integration, see \cite{Boscarino2016,Pareschi2018}.

\subsection{The uncontrolled case}
We begin by considering the Fokker-Planck model~\eqref{eq:FP_vec.uncontrolled}, i.e., in particular, the one in which microscopic interactions are not controlled, cf.~\eqref{eq:FP}, complemented with a deterministic initial condition independent of the traffic density $\rho$:
\begin{equation}
	f_0(v)=\frac{e^{-{\left(v-\frac{1}{2}\right)}^2}}{\sqrt{\pi}\Erf{\frac{1}{2}}}.
	\label{eq:f_init}
\end{equation}
Concerning the distribution of the uncertain parameter $z$, we explore mainly the following two cases:
\begin{enumerate}[label=(\roman*)]
\item uniform: $z\sim\cU([1,\,3])$, corresponding to a Legendre polynomial chaos expansion;
\item binomial: $z$ is such that $z-1\sim\B\!\left(50,\,\frac{1}{50}\right)$, corresponding to a Krawtchouk polynomial chaos expansion,
\end{enumerate}
cf. Table~\ref{tab:pol}. Notice that, in both cases, the mean value of the exponent of the probability of acceleration~\eqref{eq:P} is $\E_z(z)=2$.

\begin{figure}[!t]
\centering
\subfigure[]{\includegraphics[scale=0.5]{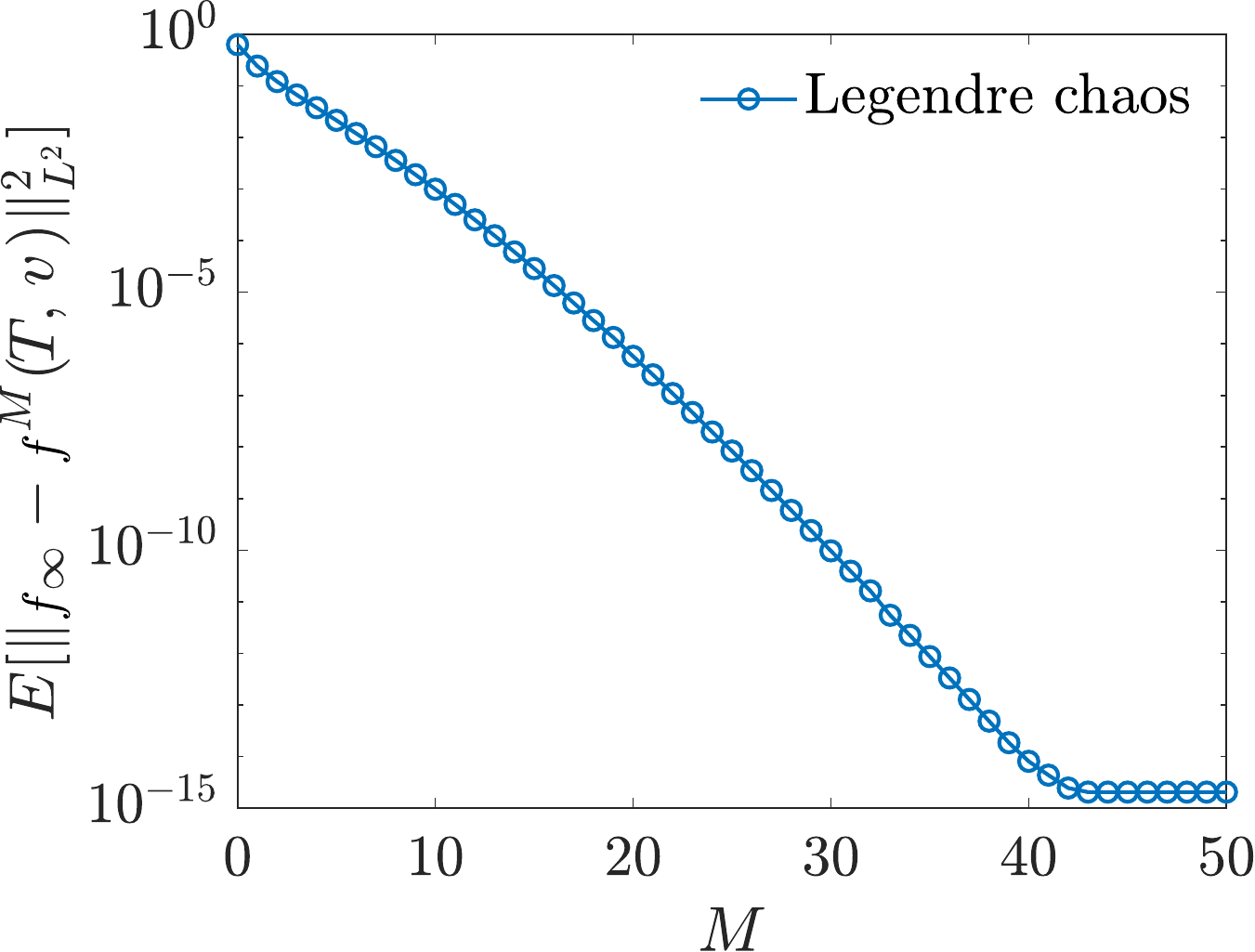}}
\subfigure[]{\includegraphics[scale=0.5]{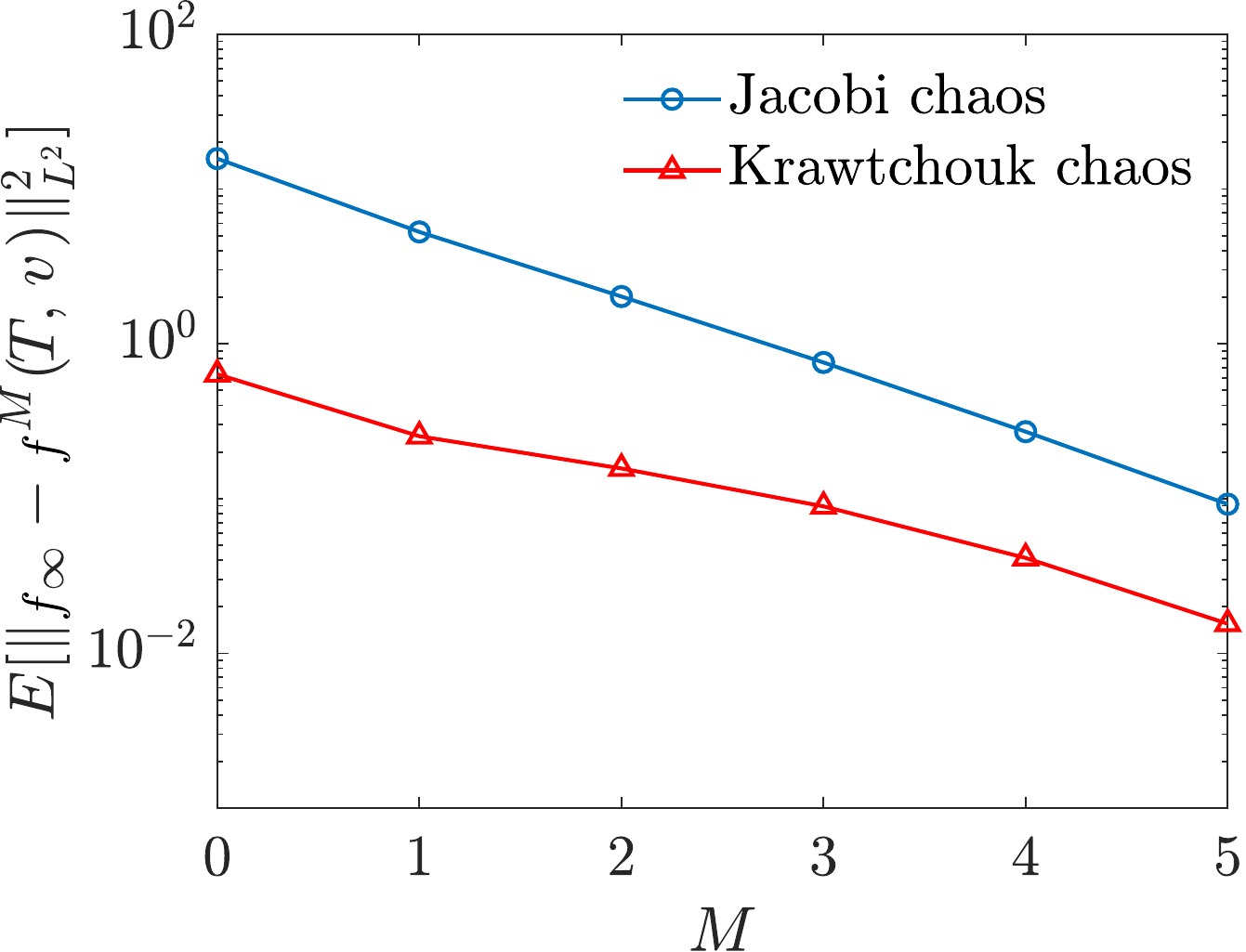}}
\caption{\textbf{Uncontrolled case}. Convergence of the $L^2$-numerical error with respect to the exact solution~\eqref{eq:ginf} of the Fokker-Planck equation~\eqref{eq:FP} for: (a) $z\sim\cU([1,\,3])$ (circular markers); (b) $z$ such that $z-1\sim\B\!\left(50,\,\frac{1}{50}\right)$ (triangular markers) and $z$ with beta distribution in $I_Z=[1,\,3]$ with zero mean and variance equal to $\frac{1}{3}$ (circular markers).}
\label{fig:convergence}
\end{figure}

In Figure~\ref{fig:convergence}, we show the trend of the $L^2$-error produced by the numerical scheme with respect to the analytical steady state $f_\infty(v;\,z)$ given by~\eqref{eq:ginf} for an increasing number $M$ of modes of the gPC expansion~\eqref{eq:gPC_expansion}. For this convergence test, we considered a uniform discretisation of the variable $v\in [0,\,1]$ by $41$ gridpoints and a uniform discretisation of the time interval $[0,\,T]$, with final time $T=60$, by a time step $\Delta{t}=1$. Moreover, we fixed the traffic density to $\rho=0.4$. From Figure~\ref{fig:convergence}(a), which refers to the uniformly distributed $z$ with Legendre polynomial chaos expansion, we observe that we reach essentially the machine precision with a relatively small number of modes. The same spectral convergence to zero of the error may be appreciated in Figure~\ref{fig:convergence}(b), which illustrates the case of the binomial distribution with Krawtchouck polynomial chaos expansion and, for completeness, also the case in which $z$ has beta distribution on the interval $I_Z=[1,\,3]$ with zero mean and variance equal to $\frac{1}{3}$. Notice that this case corresponds actually to $z\sim\cU([1,\,3])$, using however the Jacobi polynomial chaos expansion.

Figures~\ref{fig:uniform_unconstrained},~\ref{fig:binomial_unconstrained} show the contours of the mean $\bar{f}(t,\,v)$ and of the variance $\Var_z(f(t,\,v;\,z))$ of the solution $f$ to the Fokker-Planck equation~\eqref{eq:FP} for $\rho=0.2,\,0.4,\,0.6$, computed by means of the polynomial chaos expansion as
$$ \bar{f}(t,\,v)=\E_z(f(t,\,v;\,z))\approx\hat{f}_0(t,\,v), \qquad
	\operatorname{Var}_z(f(t,\,v;\,z))\approx\sum_{k=1}^{M}\hat{f}_k^2(t,\,v). $$
We used $M=20$ modes and the same discretisations of $v$ and $t$ described above. It is evident that the distribution $\Psi$ of $z$ strongly affects the results. In particular, we observe substantial differences in the instantaneous variance, which highlights the regions where the solution to~\eqref{eq:FP} is more sensitive to the uncertainty introduced by $z$.

\begin{figure}[!t]
\centering
\subfigure[$\bar{f}(t,\,v)$, $\rho=0.2$]{\includegraphics[scale=0.3]{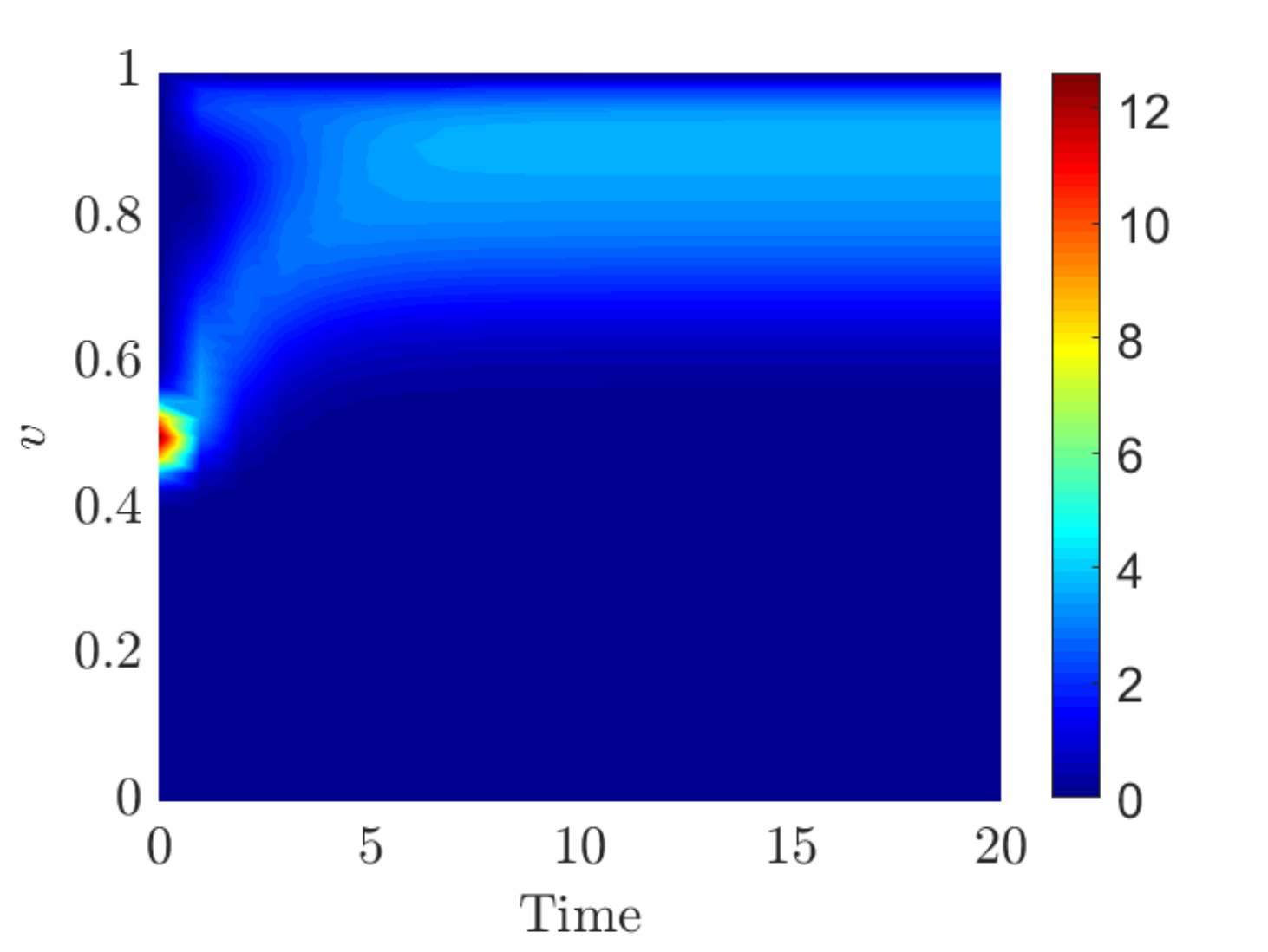}}
\subfigure[$\bar{f}(t,\,v)$, $\rho=0.4$]{\includegraphics[scale=0.3]{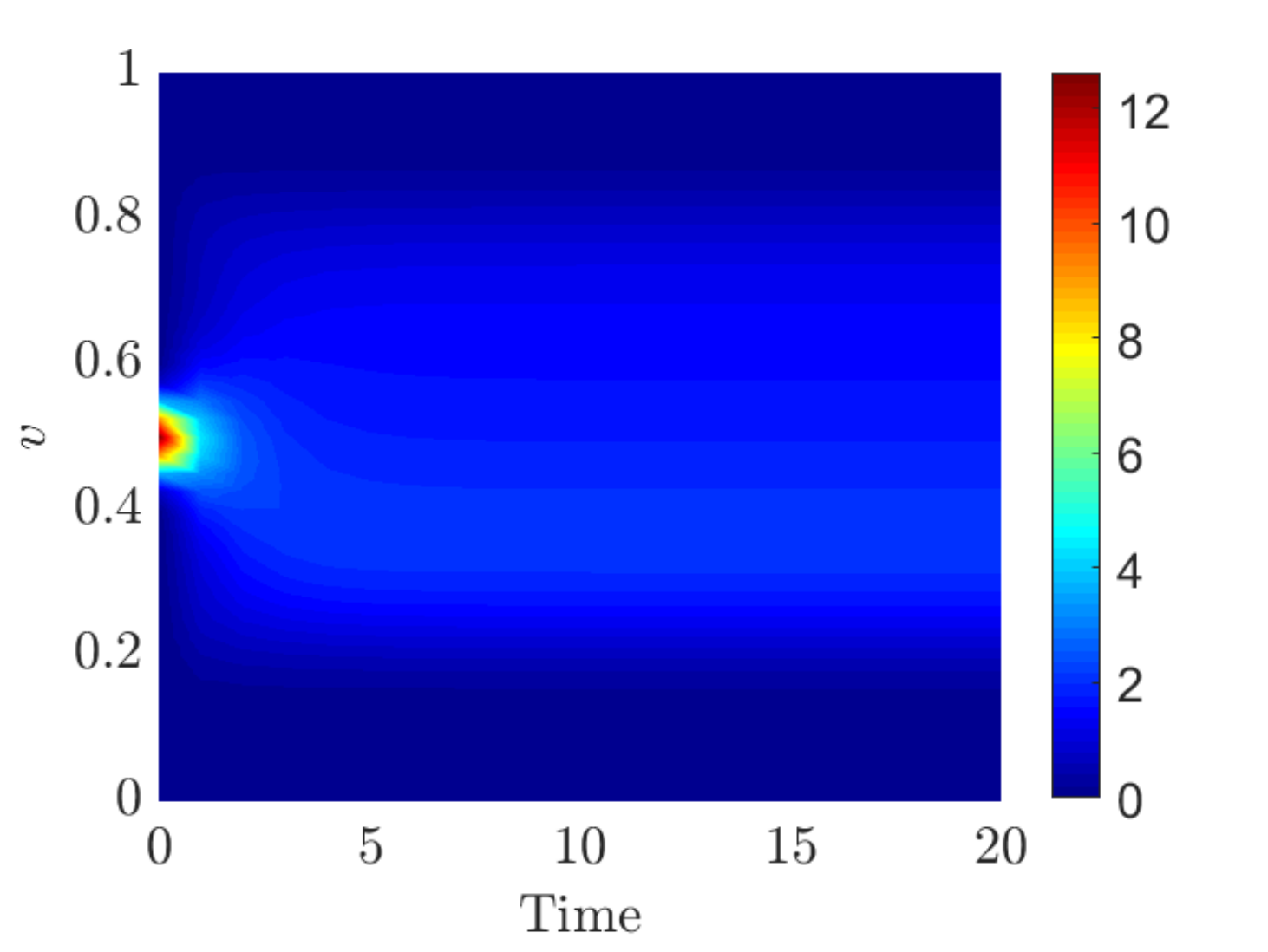}}
\subfigure[$\bar{f}(t,\,v)$, $\rho=0.6$]{\includegraphics[scale=0.3]{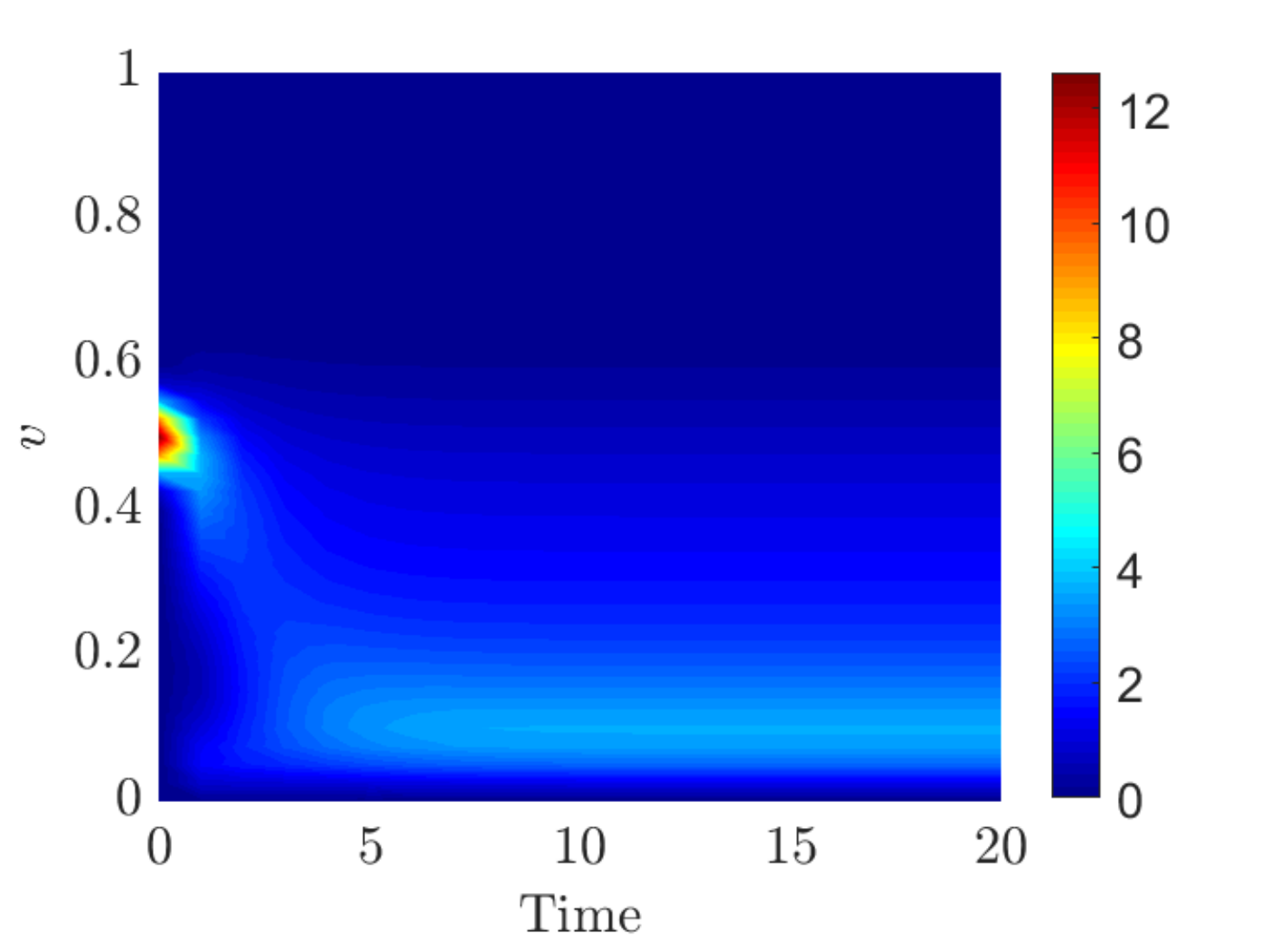}} \\
\subfigure[$\Var_z(f(t,\,v;\,z))$, $\rho=0.2$]{\includegraphics[scale=0.3]{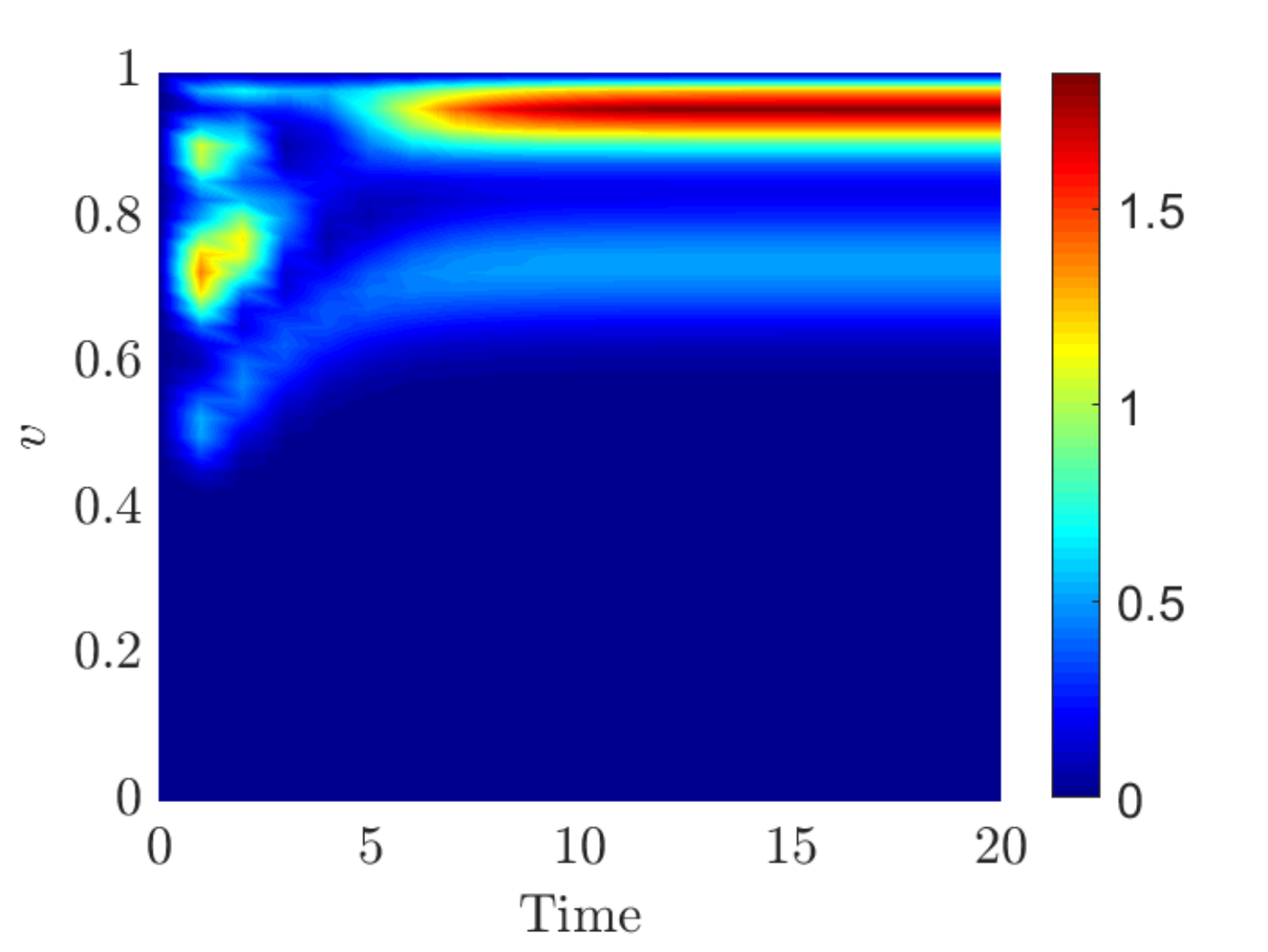}}
\subfigure[$\Var_z(f(t,\,v;\,z))$, $\rho=0.4$]{\includegraphics[scale=0.3]{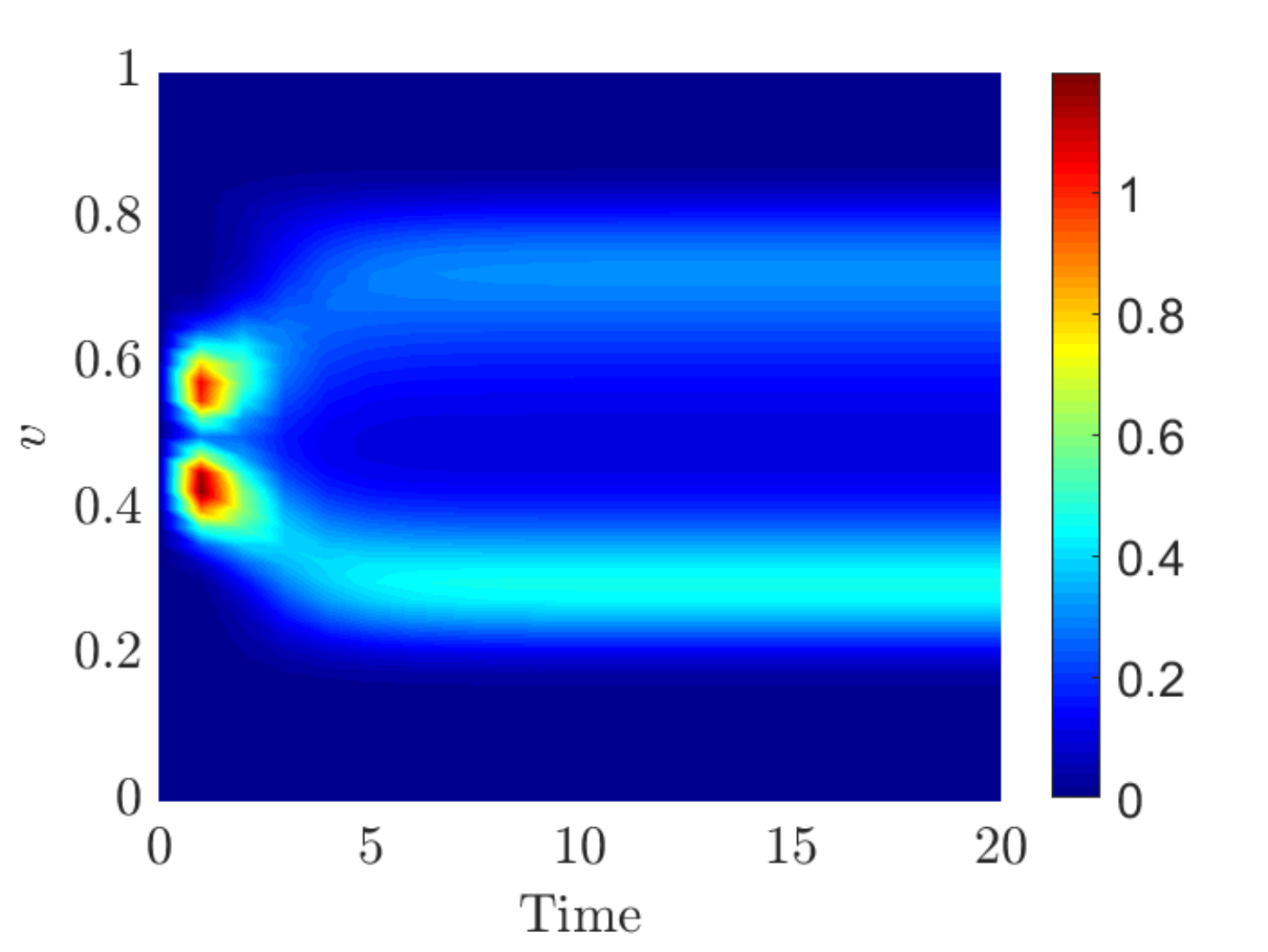}}
\subfigure[$\Var_z(f(t,\,v;\,z))$, $\rho=0.6$]{\includegraphics[scale=0.3]{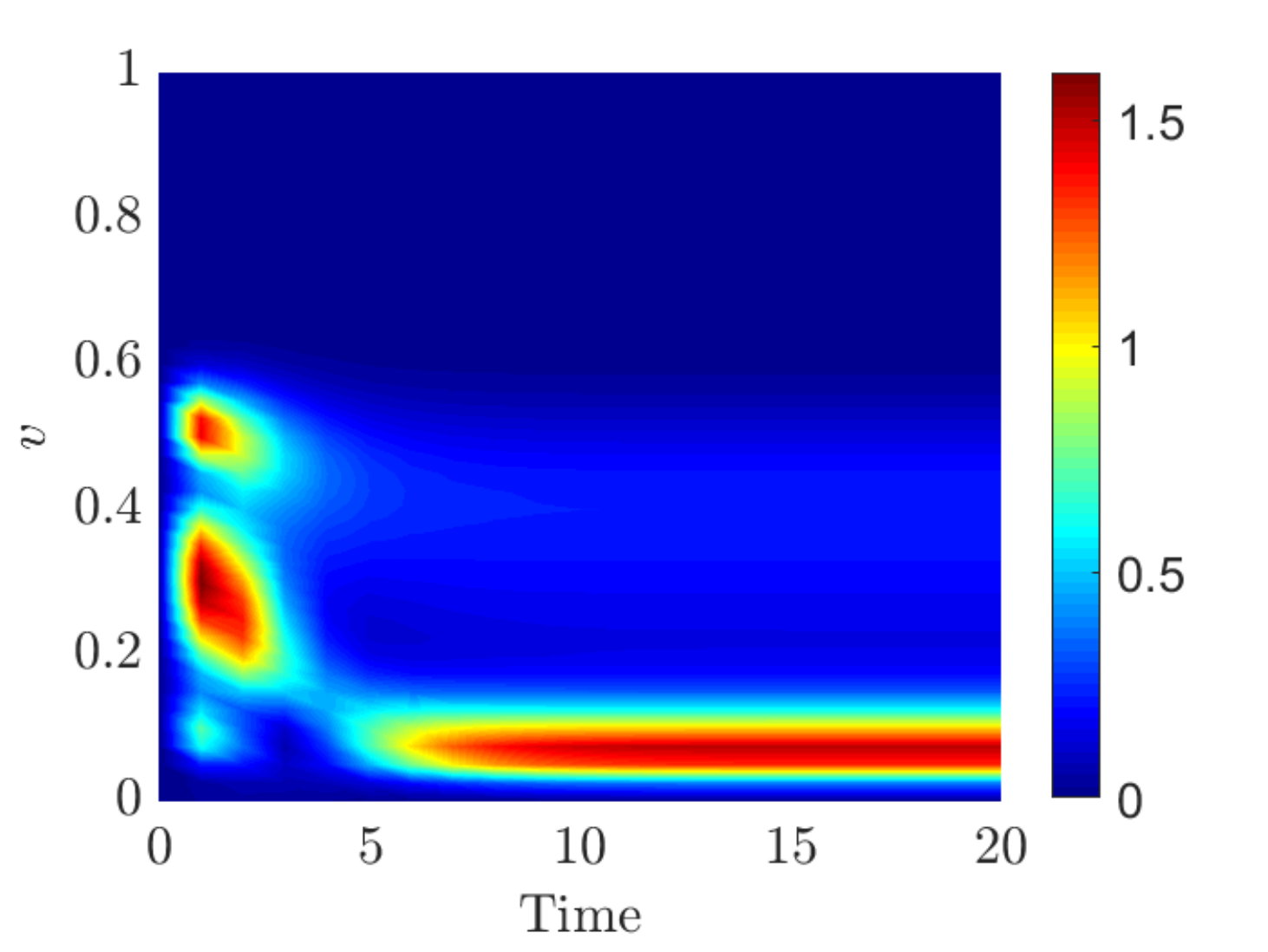}}
\caption{\textbf{Uncontrolled case}, $\boldsymbol{z\sim\cU([1,\,3])}$. Contours of $\bar{f}(t,\,v)=\E_z(f(t,\,v;\,z))$ (top row) and of $\Var_z(f(t,\,v;\,z))$ (bottom row), where $f$ is the solution to~\eqref{eq:FP} with $\lambda=5\cdot 10^{-2}$ issuing from the initial datum~\eqref{eq:f_init}, for $t\in [0,\,20]$ and $\rho=0.2,\,0.4,\,0.6$ in the case of uniformly distributed $z$.}
\label{fig:uniform_unconstrained}
\end{figure}
\begin{figure}[!t]
\centering
\subfigure[$\bar{f}(t,\,v)$, $\rho=0.2$]{\includegraphics[scale=0.3]{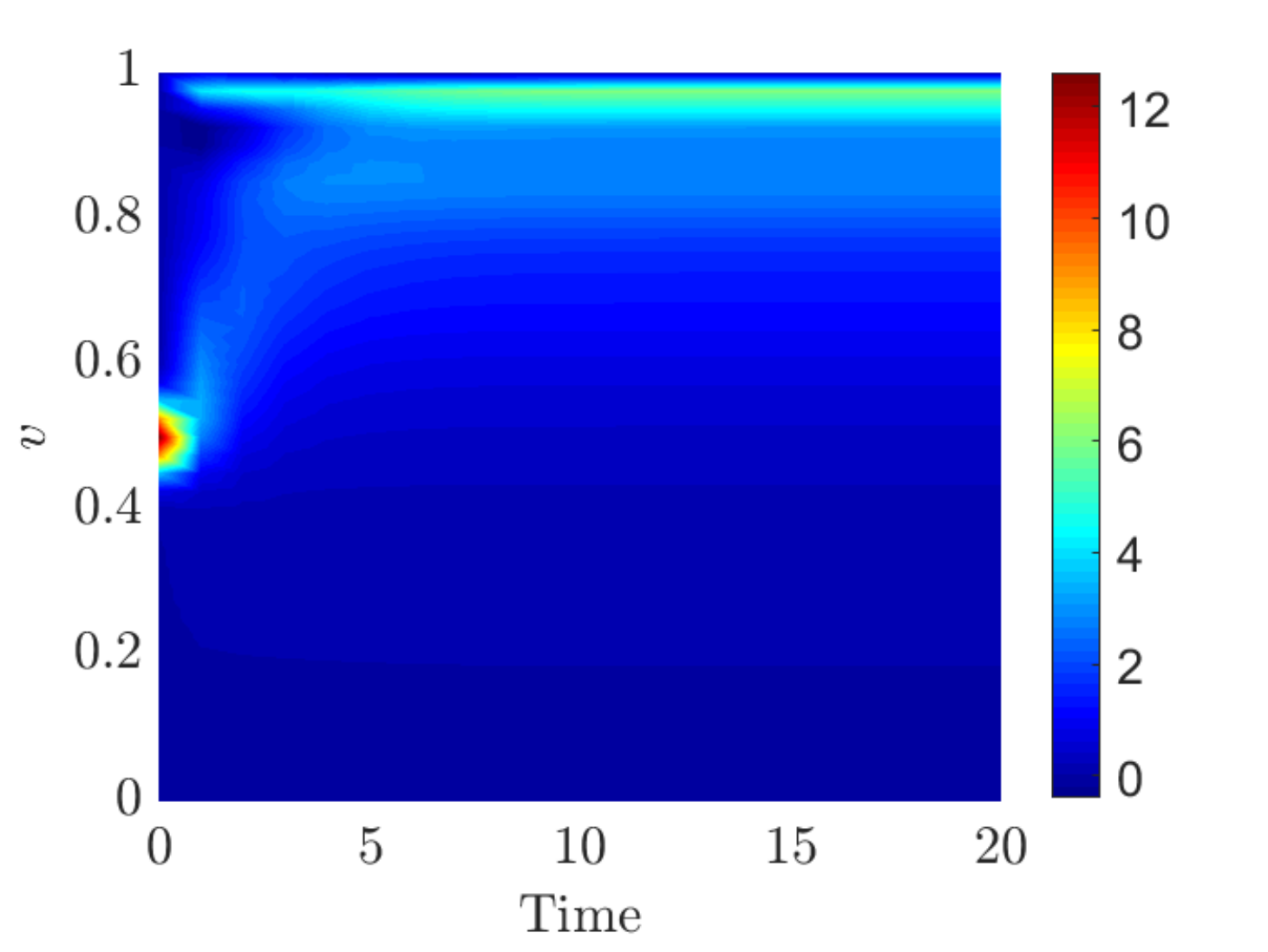}}
\subfigure[$\bar{f}(t,\,v)$, $\rho=0.4$]{\includegraphics[scale=0.3]{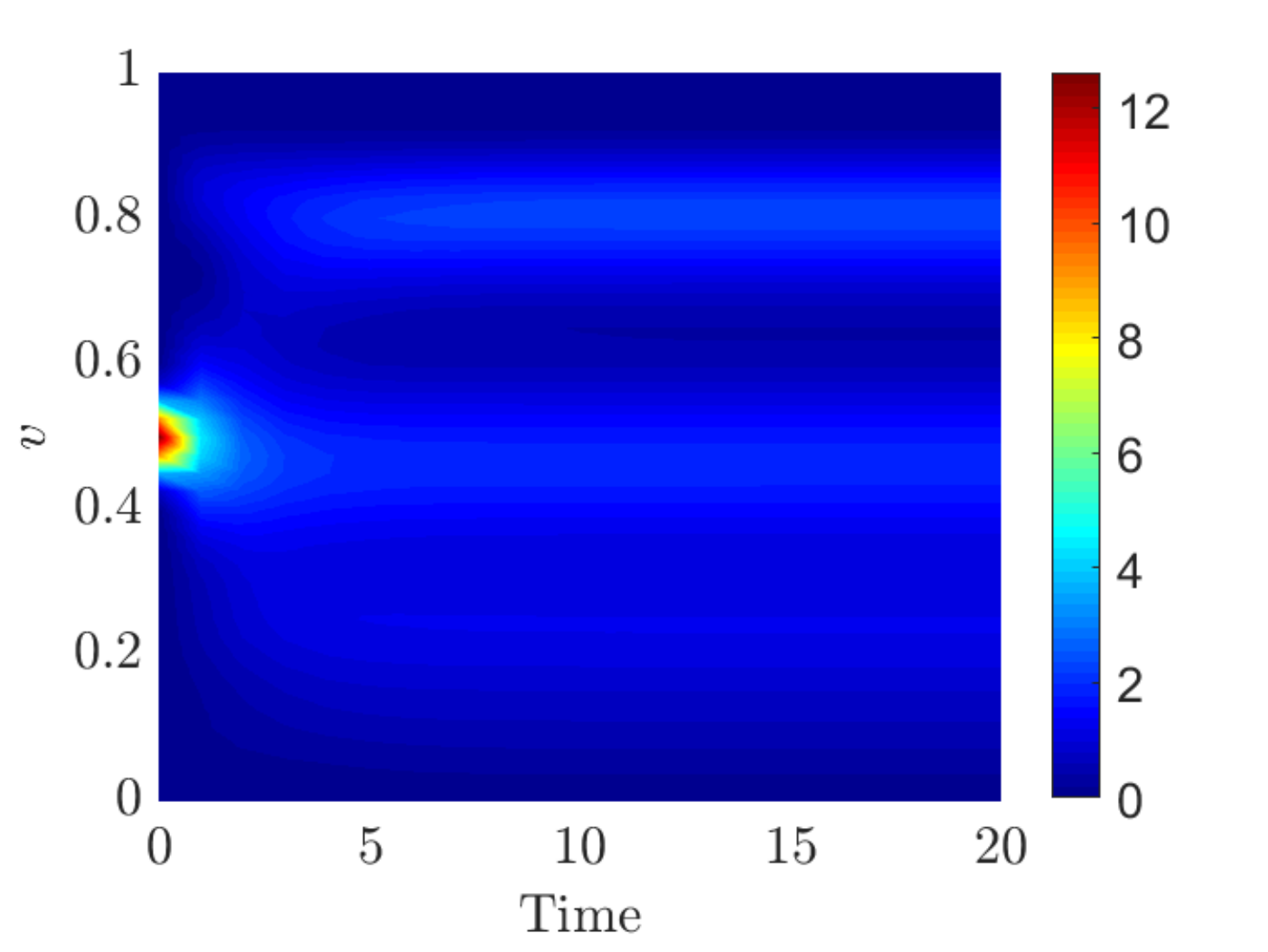}}
\subfigure[$\bar{f}(t,\,v)$, $\rho=0.6$]{\includegraphics[scale=0.3]{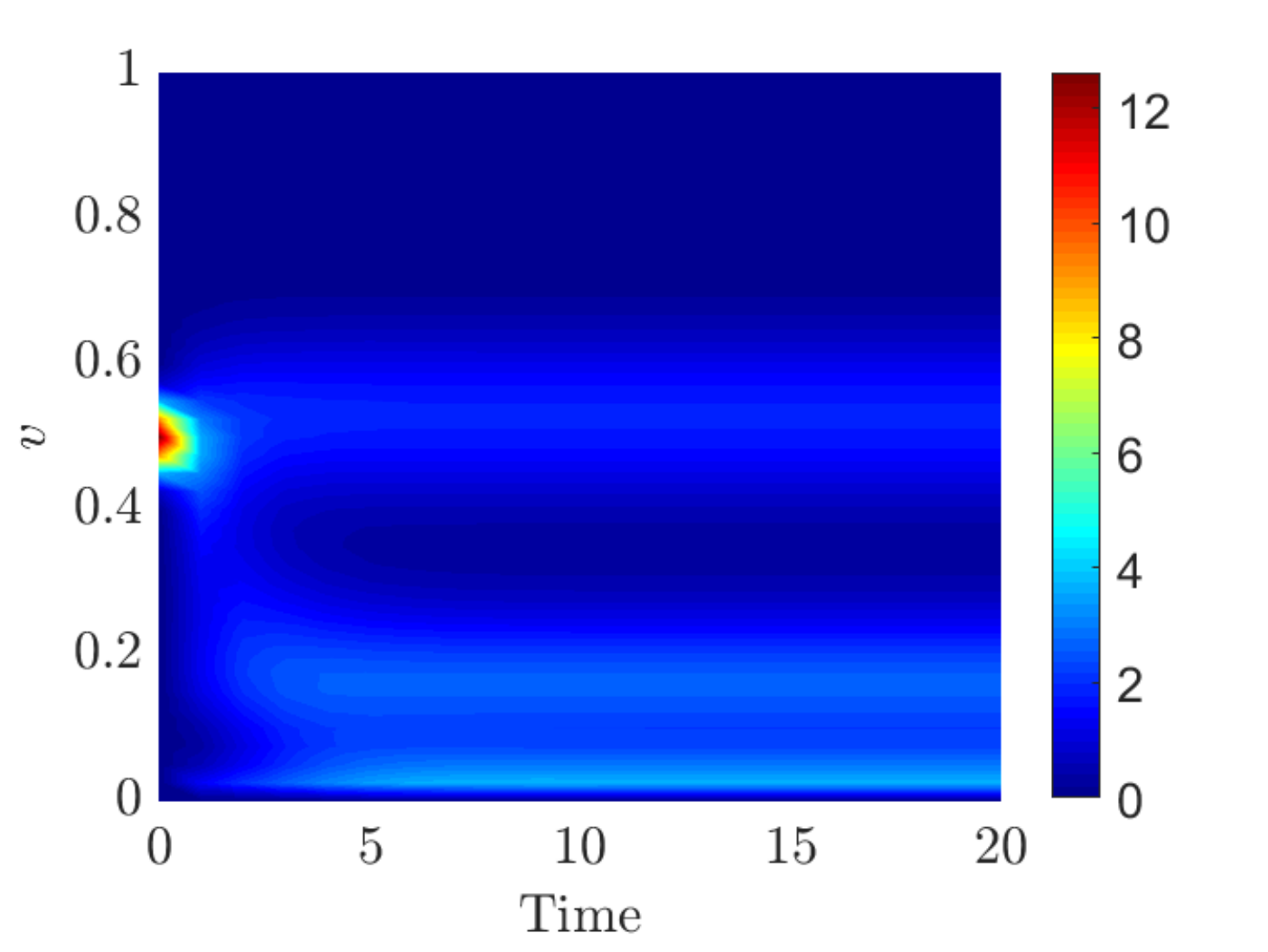}} \\
\subfigure[$\Var_z(f(t,\,v;\,z))$, $\rho=0.2$]{\includegraphics[scale=0.3]{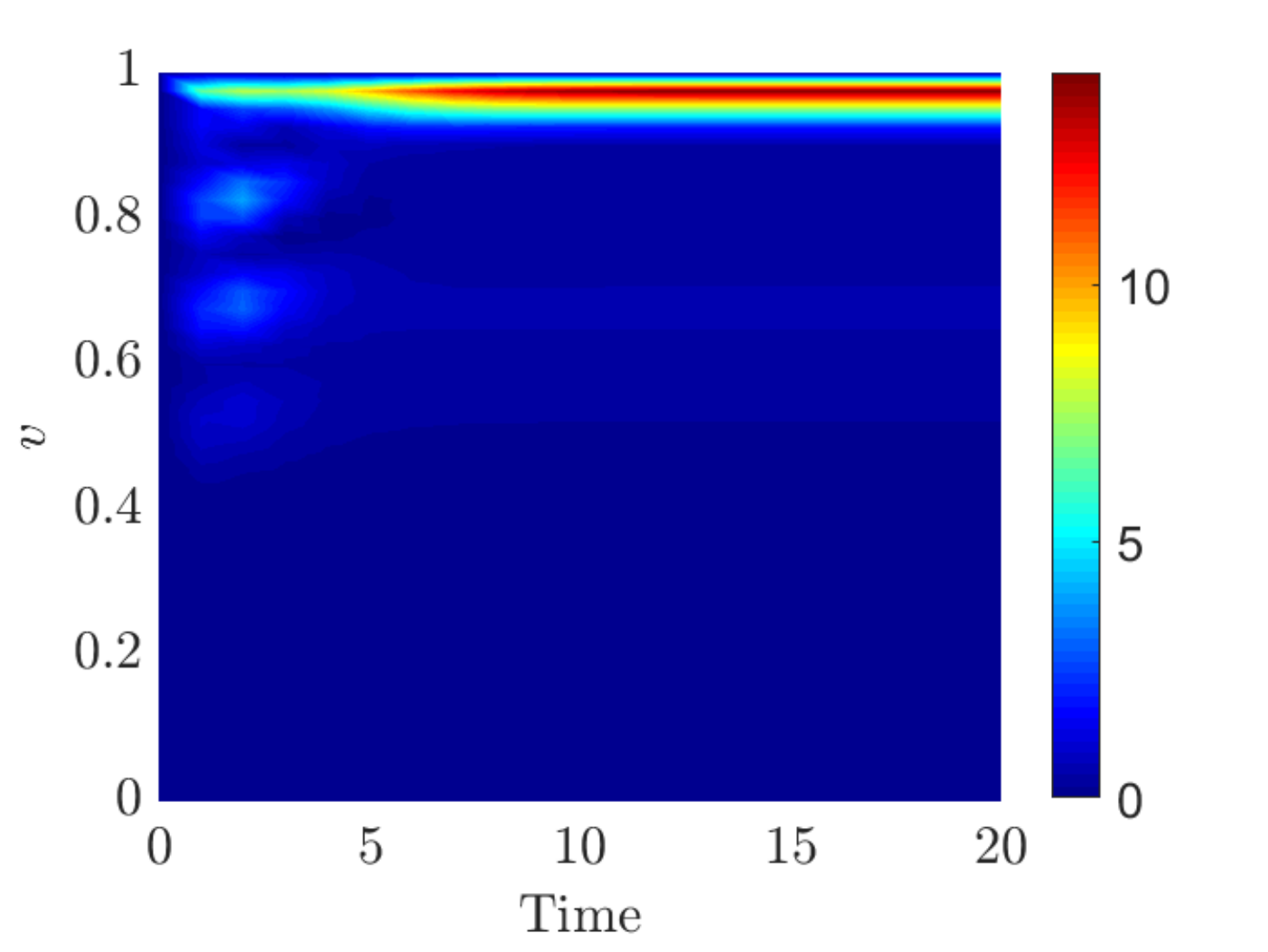}}
\subfigure[$\Var_z(f(t,\,v;\,z))$, $\rho=0.4$]{\includegraphics[scale=0.3]{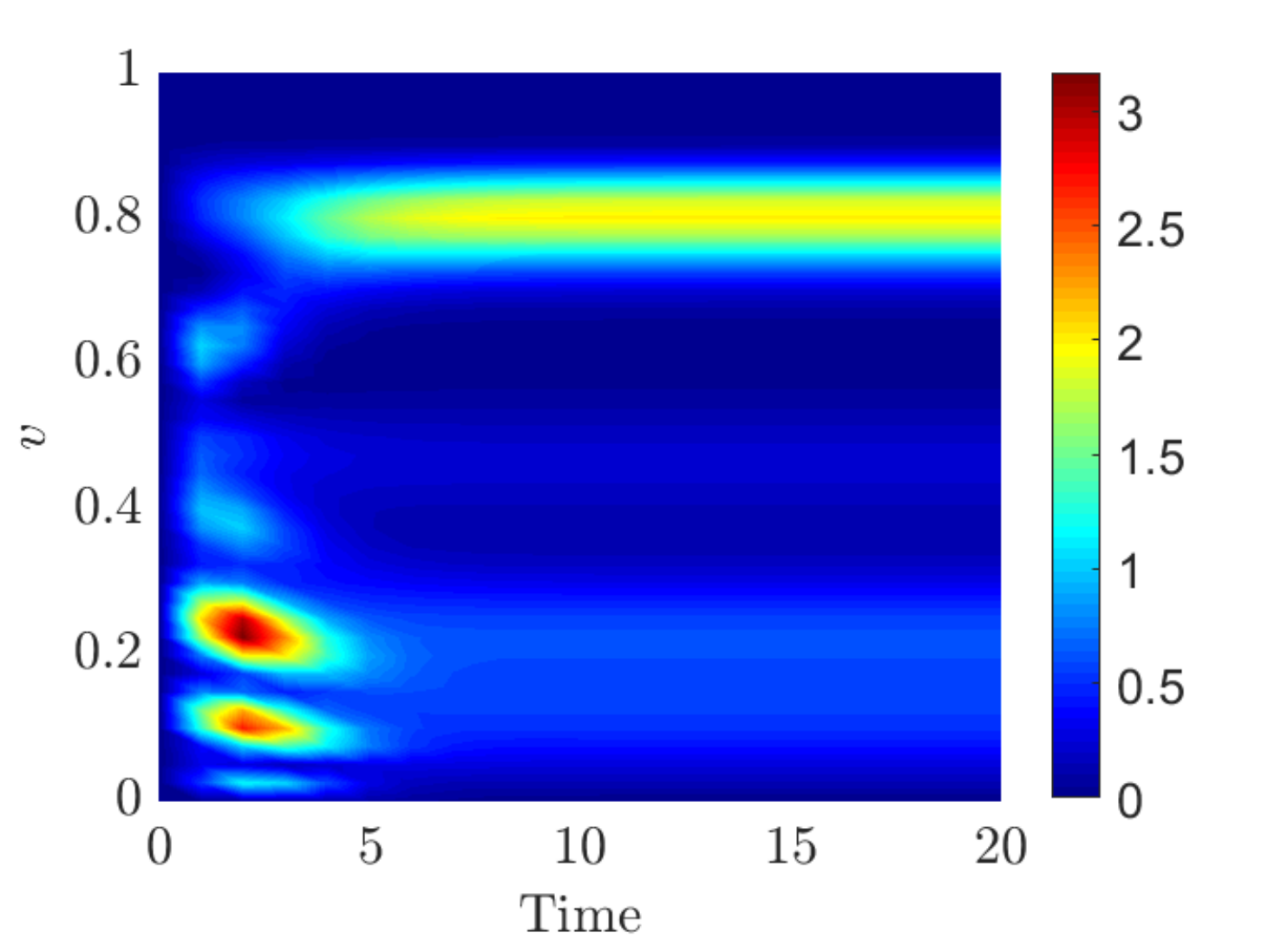}}
\subfigure[$\Var_z(f(t,\,v;\,z))$, $\rho=0.6$]{\includegraphics[scale=0.3]{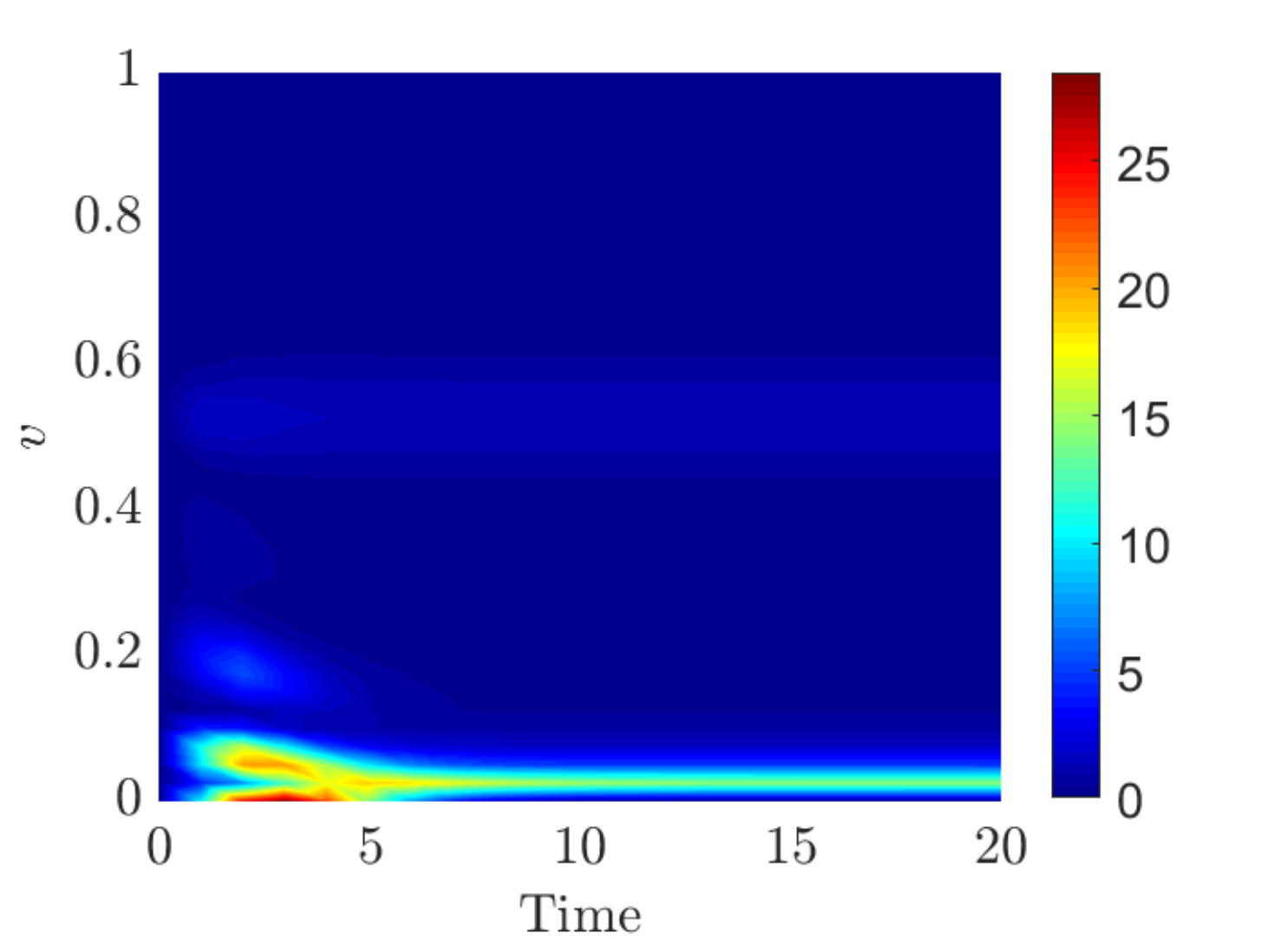}}
\caption{\textbf{Uncontrolled case}, $\boldsymbol{z-1\sim\B\!\left(50,\,\frac{1}{50}\right)}$. Contours of $\bar{f}(t,\,v)=\E_z(f(t,\,v;\,z))$ (top row) and of $\Var_z(f(t,\,v;\,z))$ (bottom row), where $f$ is the solution to~\eqref{eq:FP}, when $z$ is such that $z-1$ has binomial distribution. All the parameters are like in Figure~\ref{fig:uniform_unconstrained}.}
\label{fig:binomial_unconstrained}
\end{figure}

\subsection{The controlled case: uncertainty damping by control methods}
We consider now the Fokker-Planck model~\eqref{eq:FP_h.controlled}, in which the microscopic interactions~\eqref{eq:binary-u.1} are controlled by the ADAS technology through the uncertain control~\eqref{eq:ustar.1}. Owing to the results of Section~\ref{sect:Boltzmann.2}, we know that the observed large time trends are also produced by the microscopic interactions~\eqref{eq:binary-u.2} implementing the deterministic control~\eqref{eq:ustar.2}.

In particular, we prescribe the same deterministic initial distribution~\eqref{eq:f_init} as in the uncontrolled case and, as far as the distribution of the uncertainty is concerned, we explore the same cases $z\sim\cU([1,\,3])$ and $z-1\sim\B\!\left(50,\,\frac{1}{50}\right)$ already considered in the uncontrolled case. Moreover, we fix the penetration rate to $p=0.1$, meaning that $10\%$ of the vehicles in the traffic stream are equipped with the ADAS technology, and we consider various penalisation coefficients $\kappa>0$ of the control. We remind that, in~\eqref{eq:FP.1}, these two factors are linked through the parameter $p^\ast$. Finally, we set the optimal speed to $v_d(\rho)=1-\rho$, like in the cases illustrated in Figures~\ref{fig:damping_diag},~\ref{fig:dis_pstar}.

\begin{figure}[!t]
\centering
\subfigure[$p^\ast=1$, $\rho=0.2$]{\includegraphics[scale=0.3]{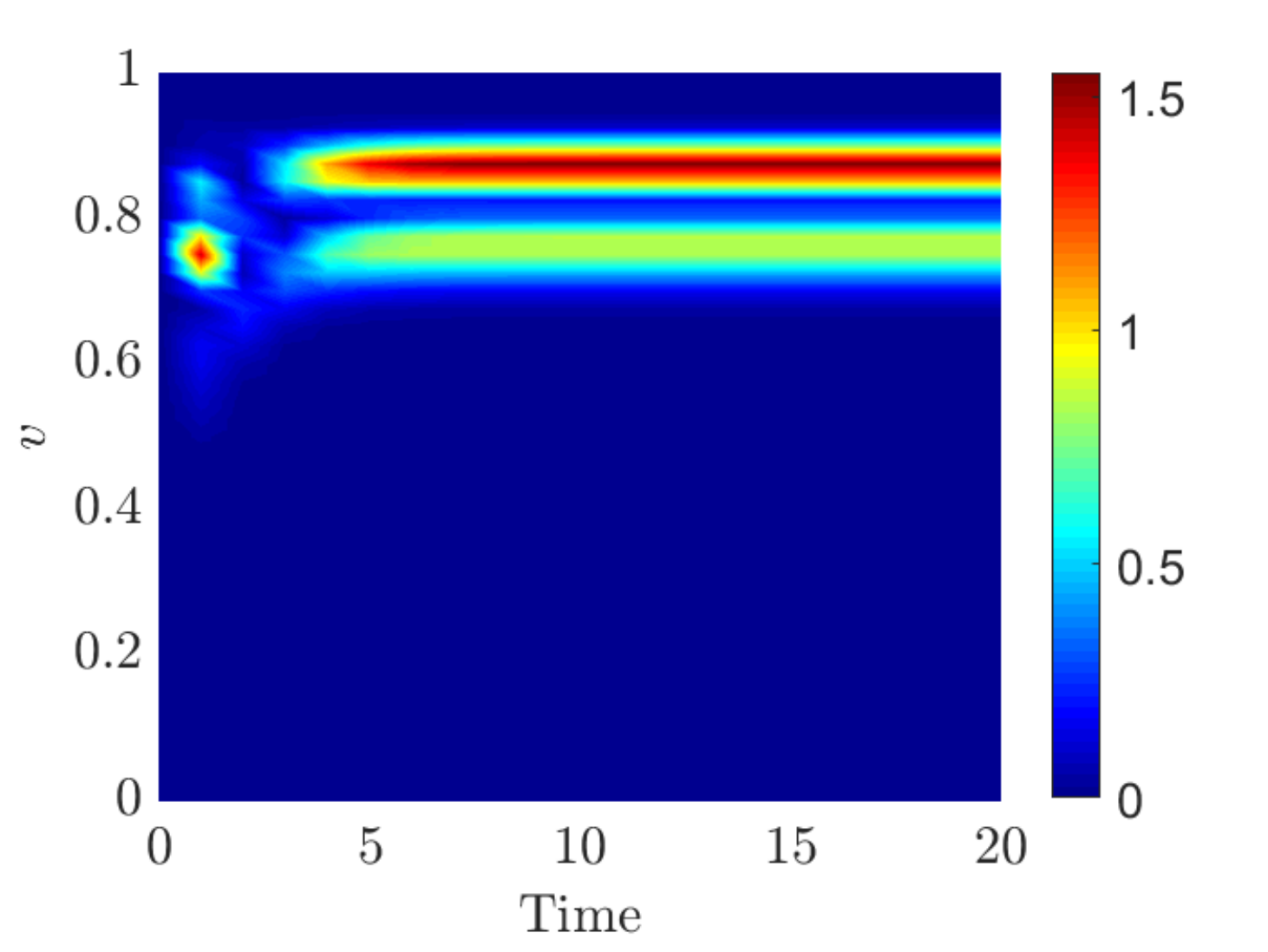}}
\subfigure[$p^\ast=1$, $\rho=0.4$]{\includegraphics[scale=0.3]{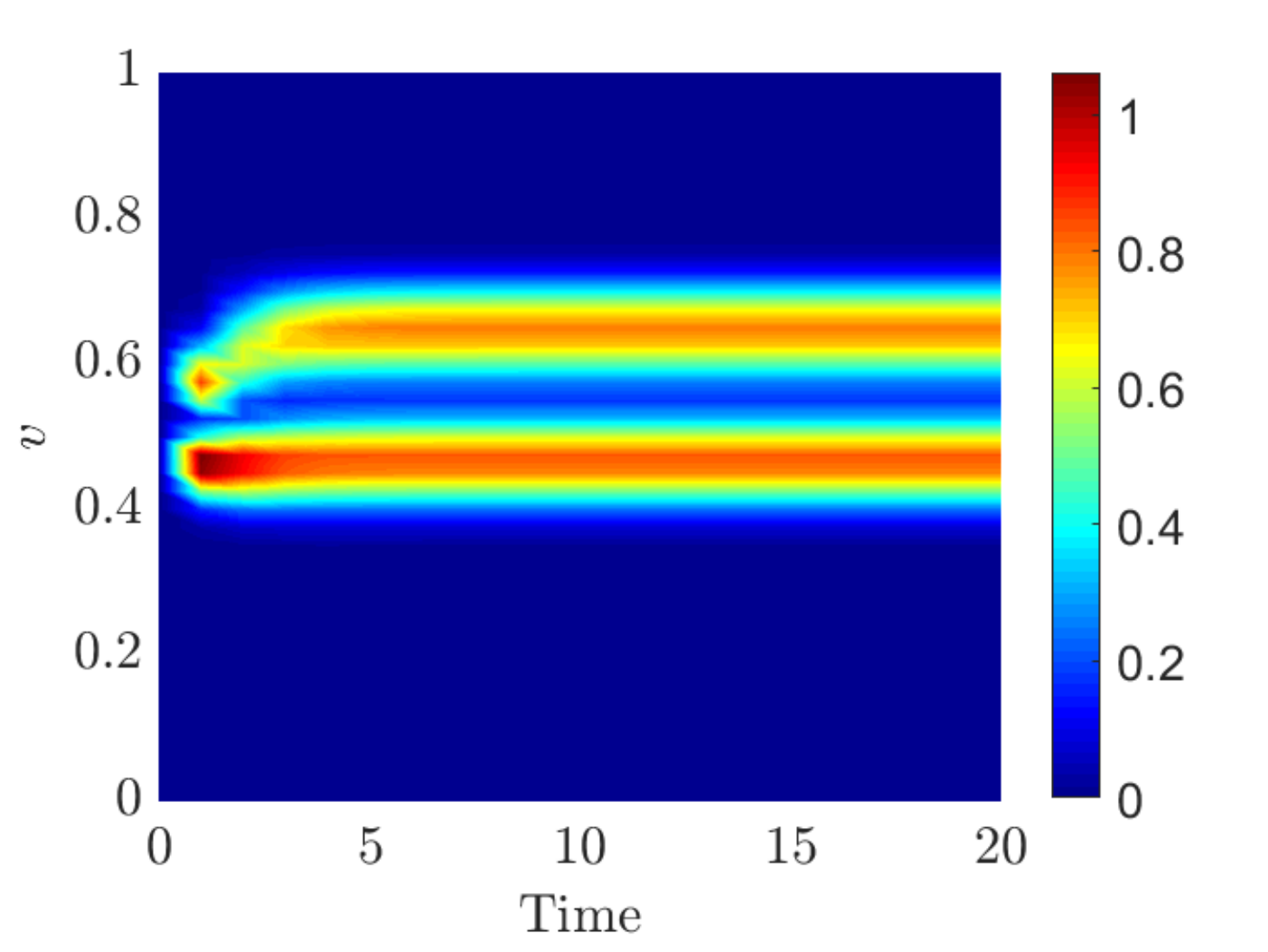}}
\subfigure[$p^\ast=1$, $\rho=0.6$]{\includegraphics[scale=0.3]{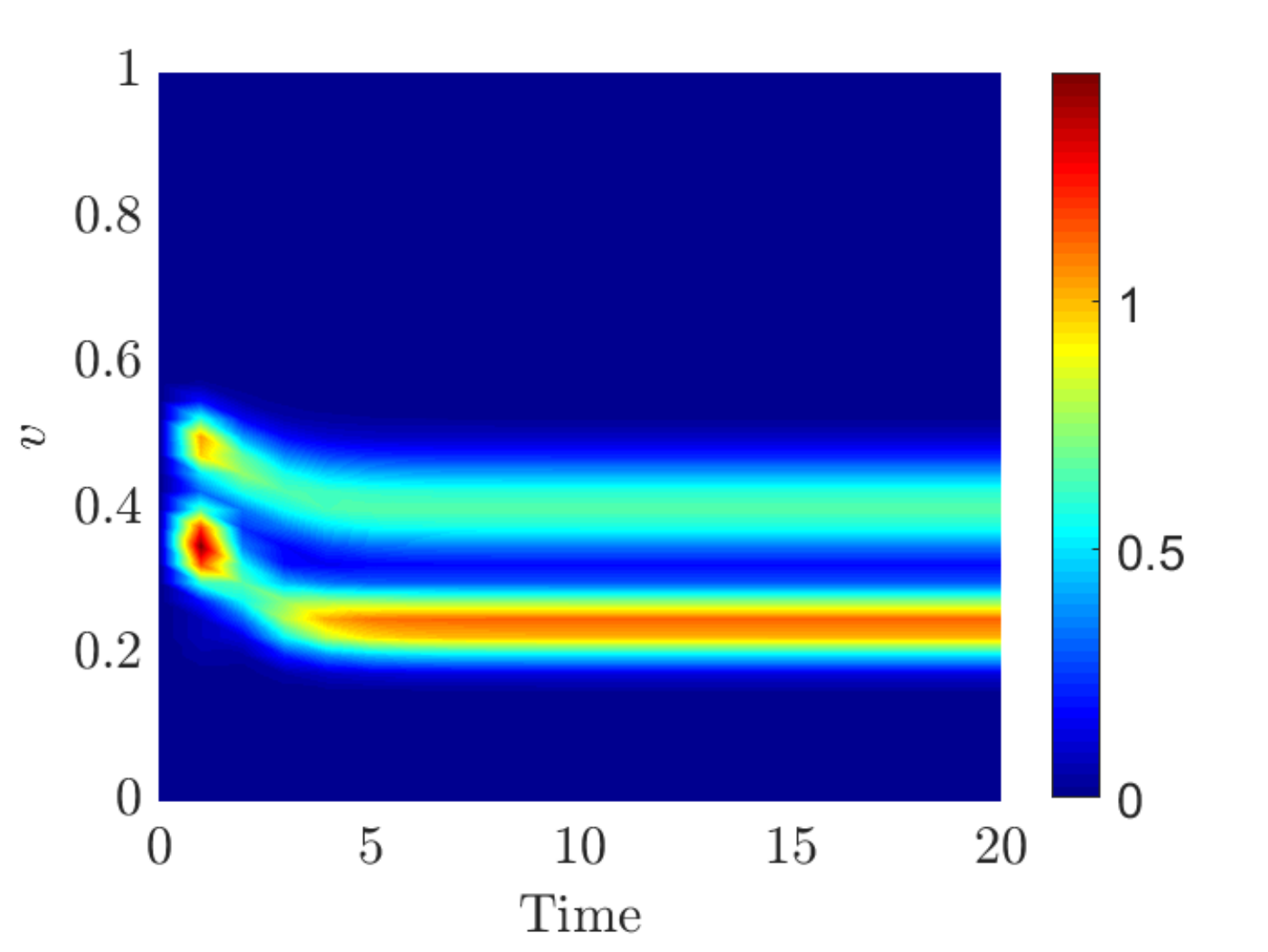}} \\
\subfigure[$p^\ast=10$, $\rho=0.2$]{\includegraphics[scale=0.3]{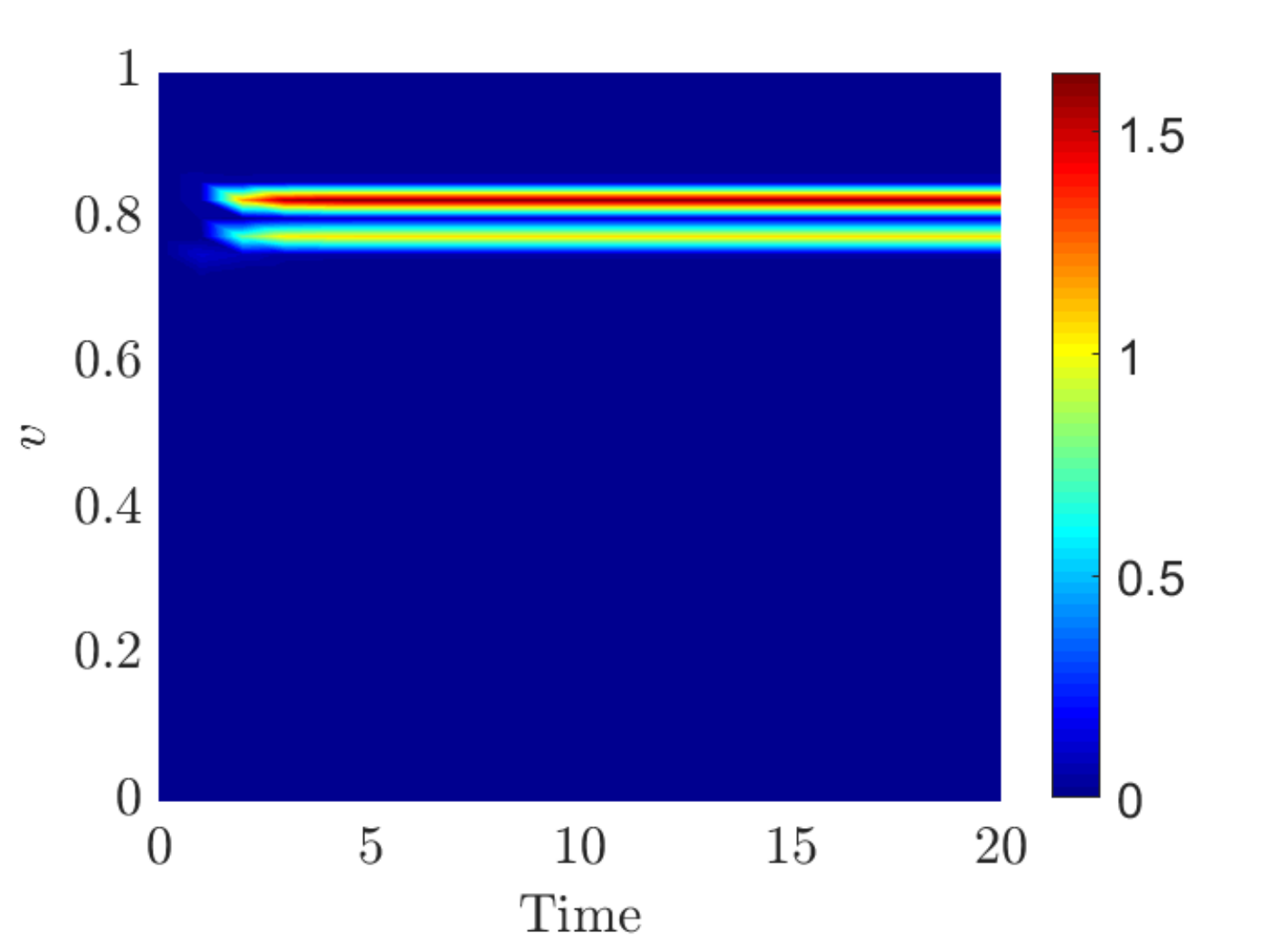}}
\subfigure[$p^\ast=10$, $\rho=0.4$]{\includegraphics[scale=0.3]{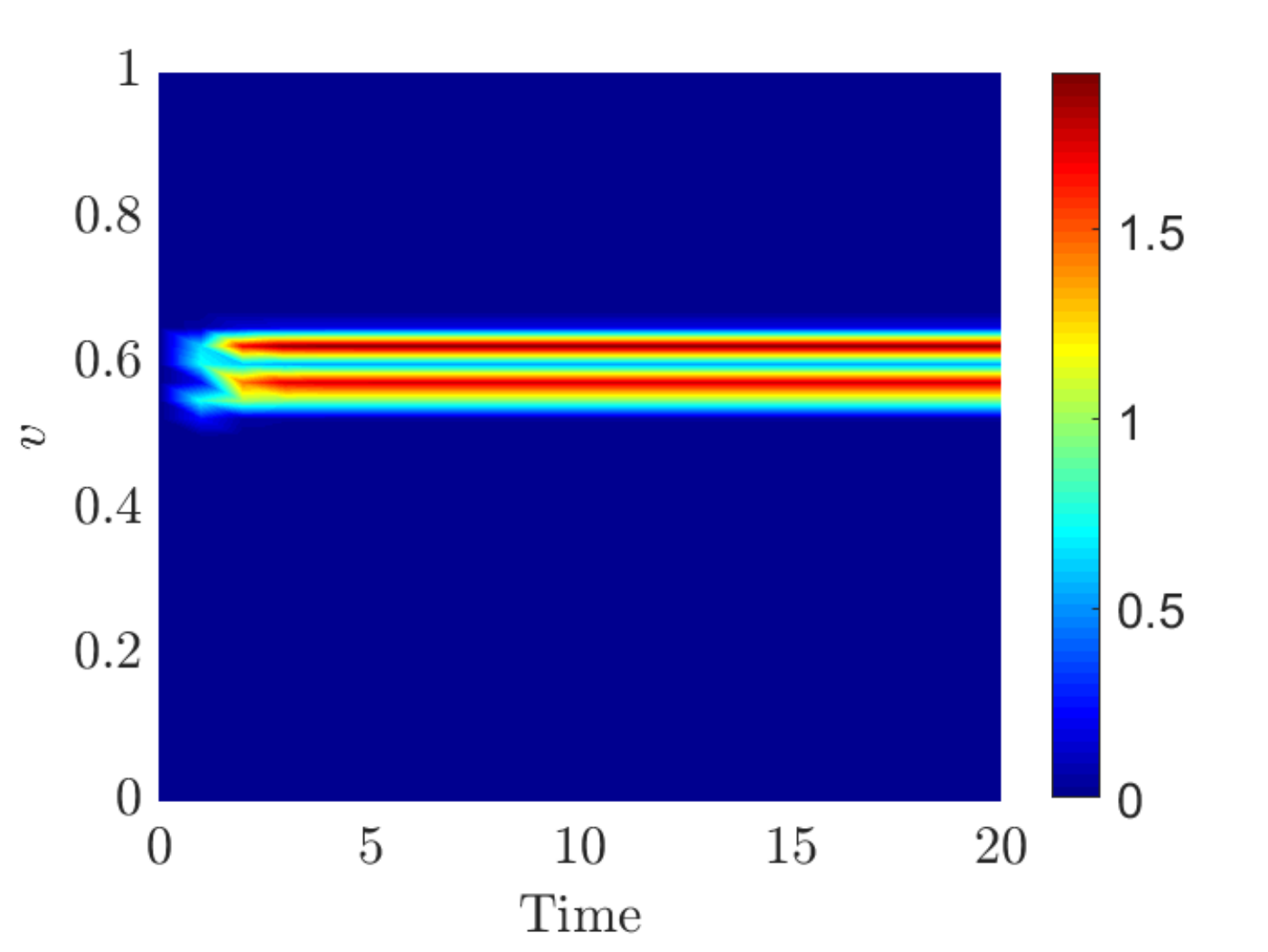}}
\subfigure[$p^\ast=10$, $\rho=0.6$]{\includegraphics[scale=0.3]{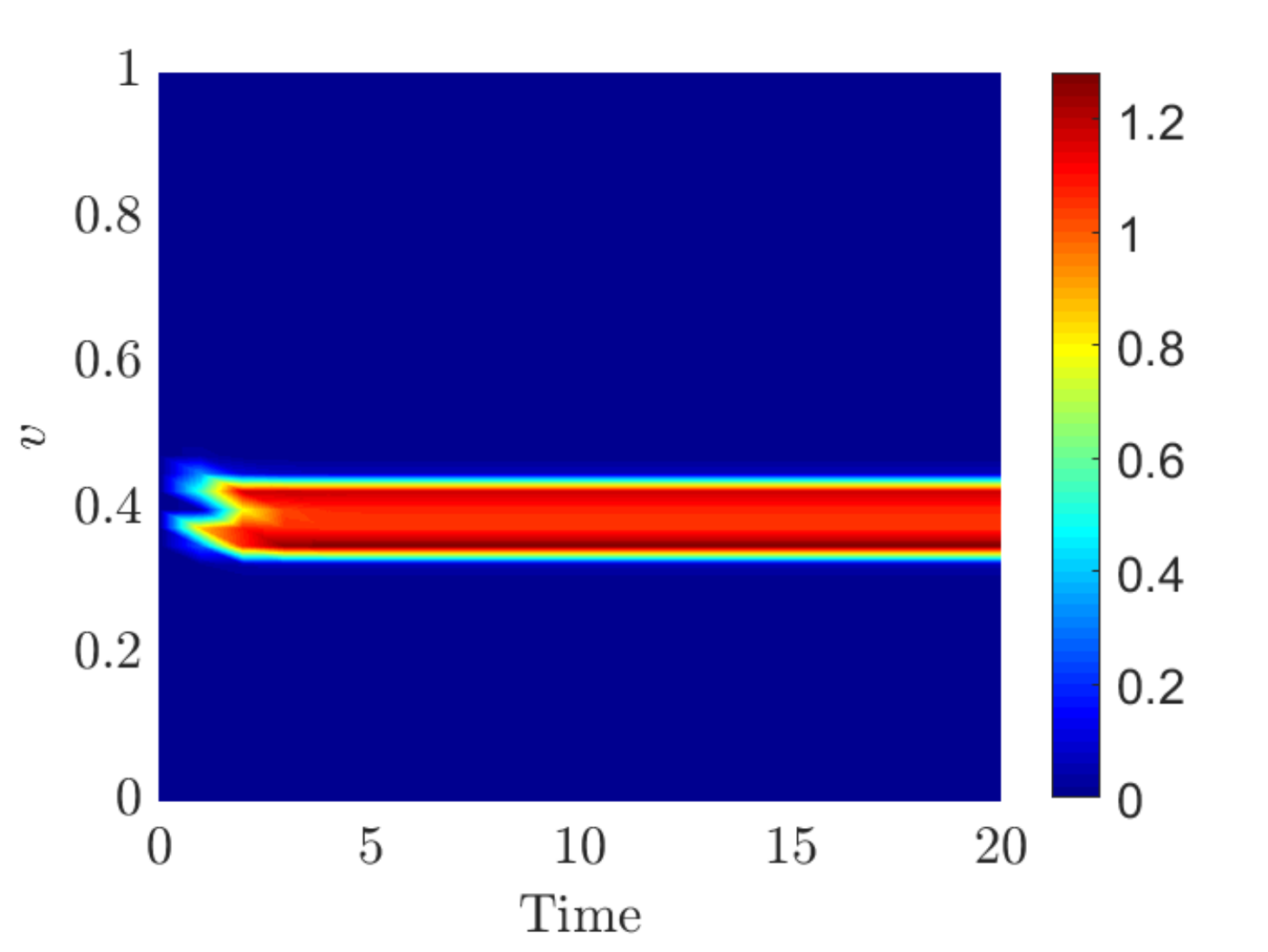}} 
\caption{\textbf{Controlled case}, $\boldsymbol{z\sim\cU([1,\,3])}$. Contours of $\Var_z(f(t,\,v;\,z))$, where $f$ is the solution to~\eqref{eq:FP.1} with $\lambda=5\cdot 10^{-2}$ issuing from the initial datum~\eqref{eq:f_init}, for $t\in [0,\,20]$ and $\rho=0.2,\,0.4,\,0.6$ in the case of uniformly distributed $z$. Top row: $p^\ast=1$; bottom row: $p^\ast=10$.}
\label{fig:uni_cons}
\end{figure}
\begin{figure}[!t]
\centering
\subfigure[$p^\ast=1$, $\rho=0.2$]{\includegraphics[scale=0.3]{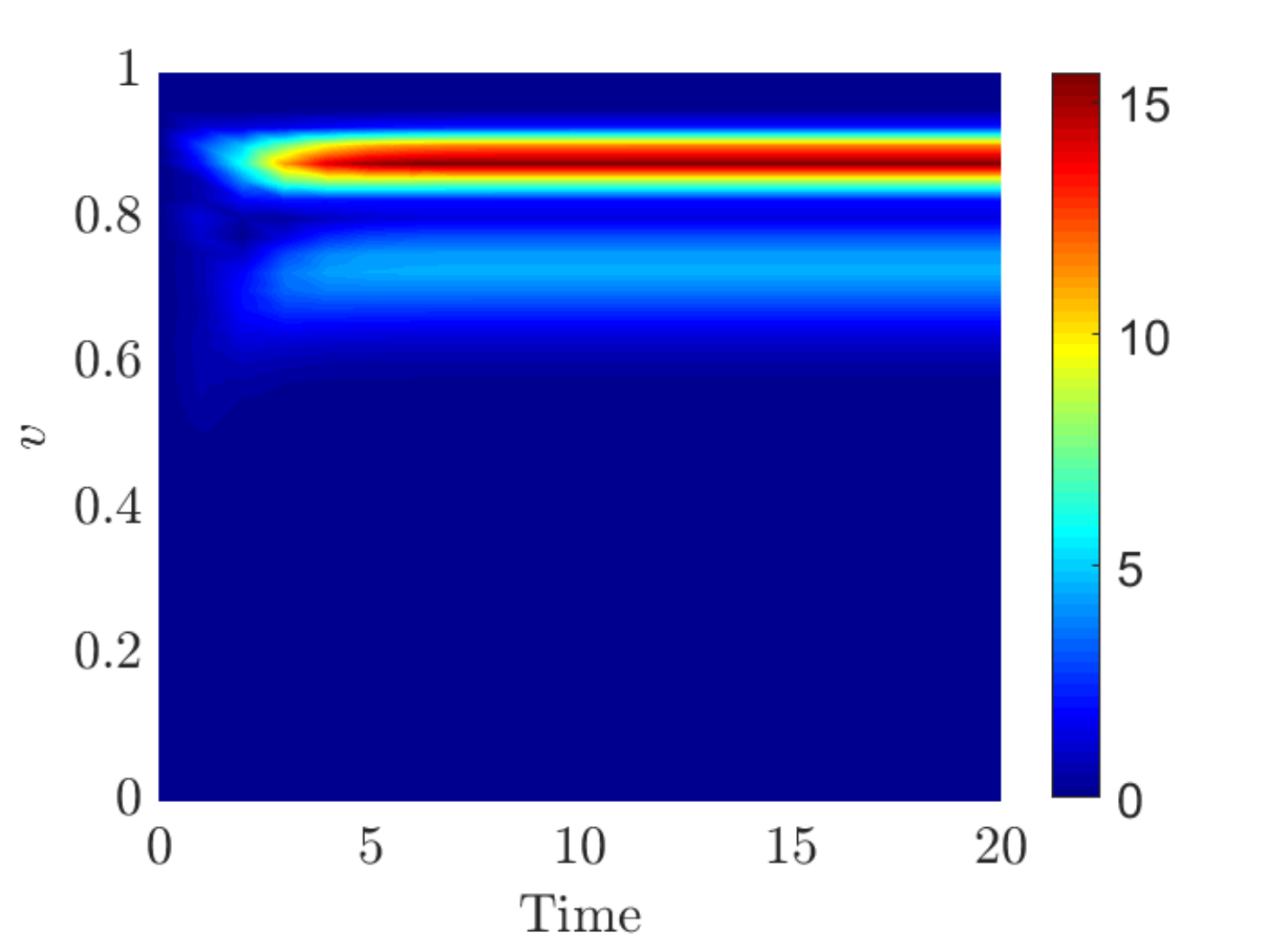}}
\subfigure[$p^\ast=1$, $\rho=0.4$]{\includegraphics[scale=0.3]{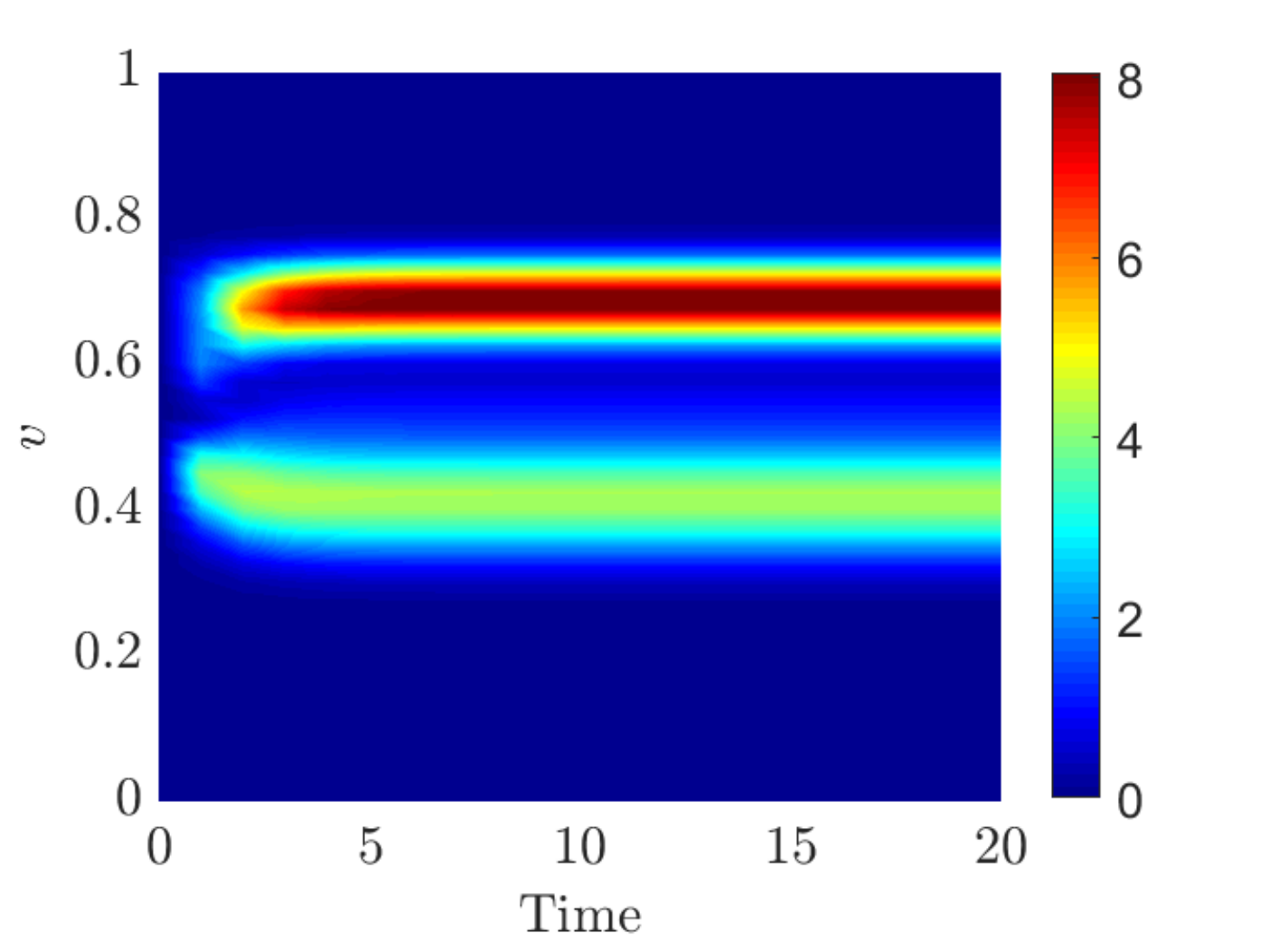}}
\subfigure[$p^\ast=1$, $\rho=0.6$]{\includegraphics[scale=0.3]{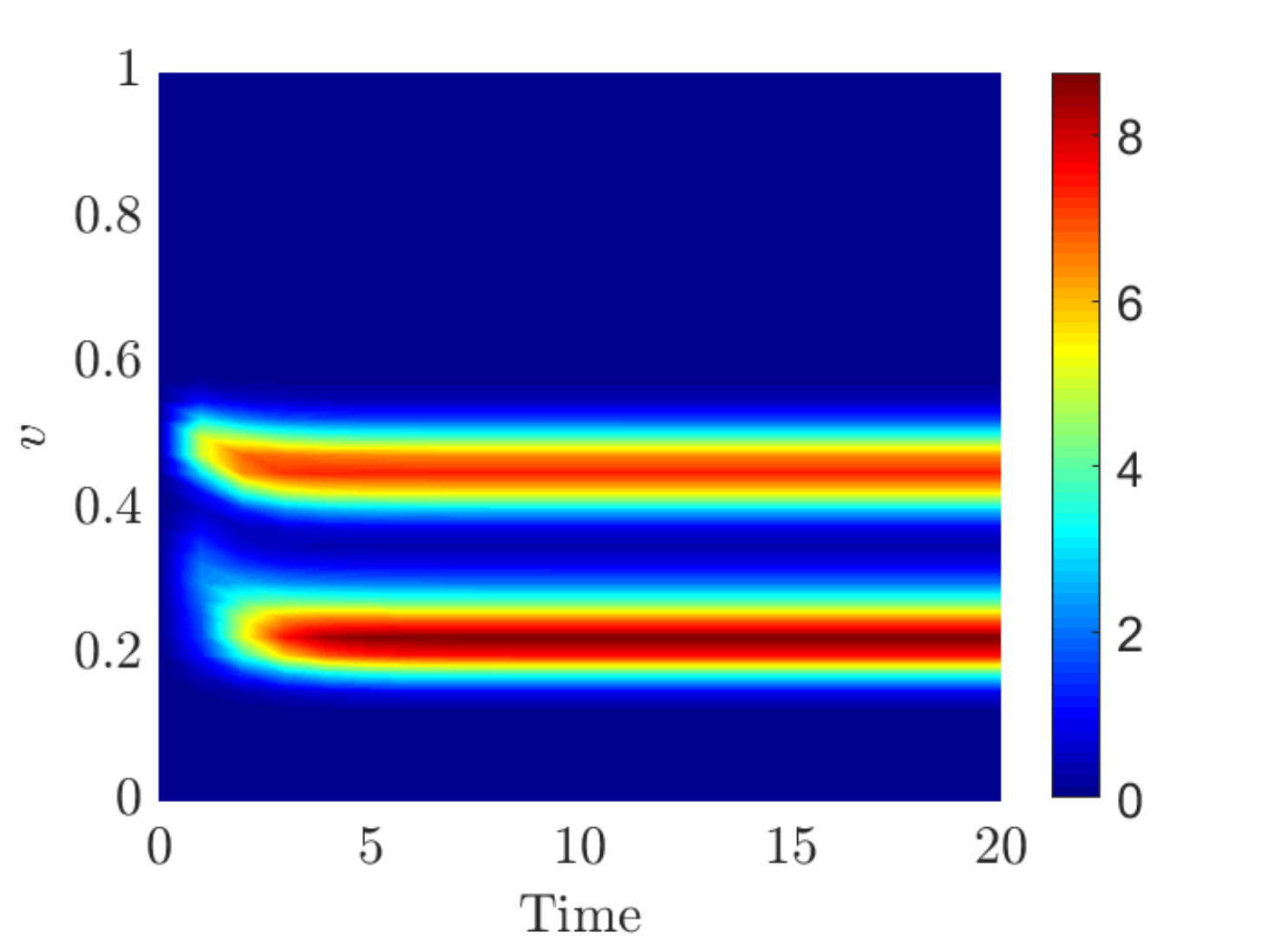}} \\
\subfigure[$p^\ast=10$, $\rho=0.2$]{\includegraphics[scale=0.3]{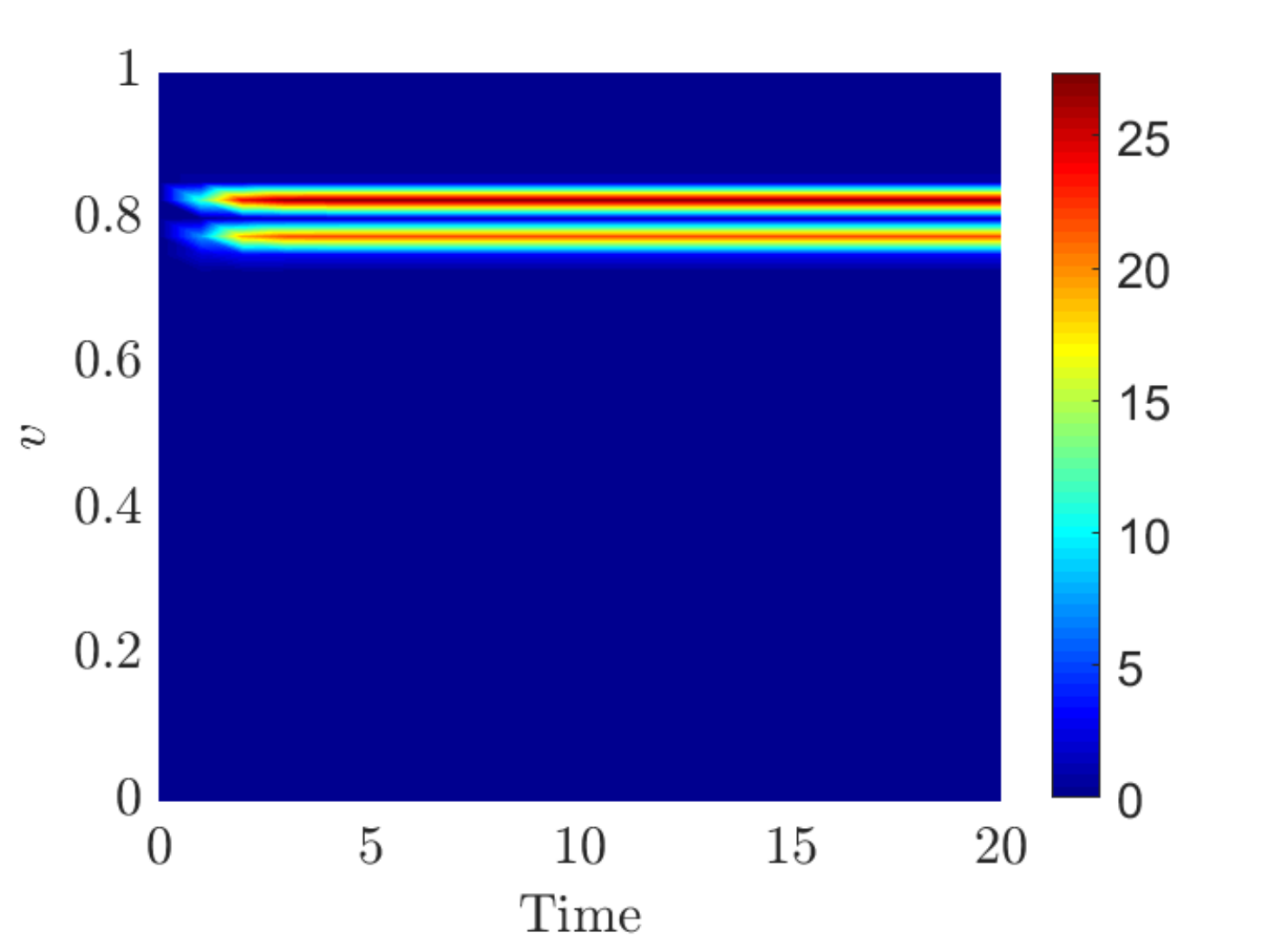}}
\subfigure[$p^\ast=10$, $\rho=0.4$]{\includegraphics[scale=0.3]{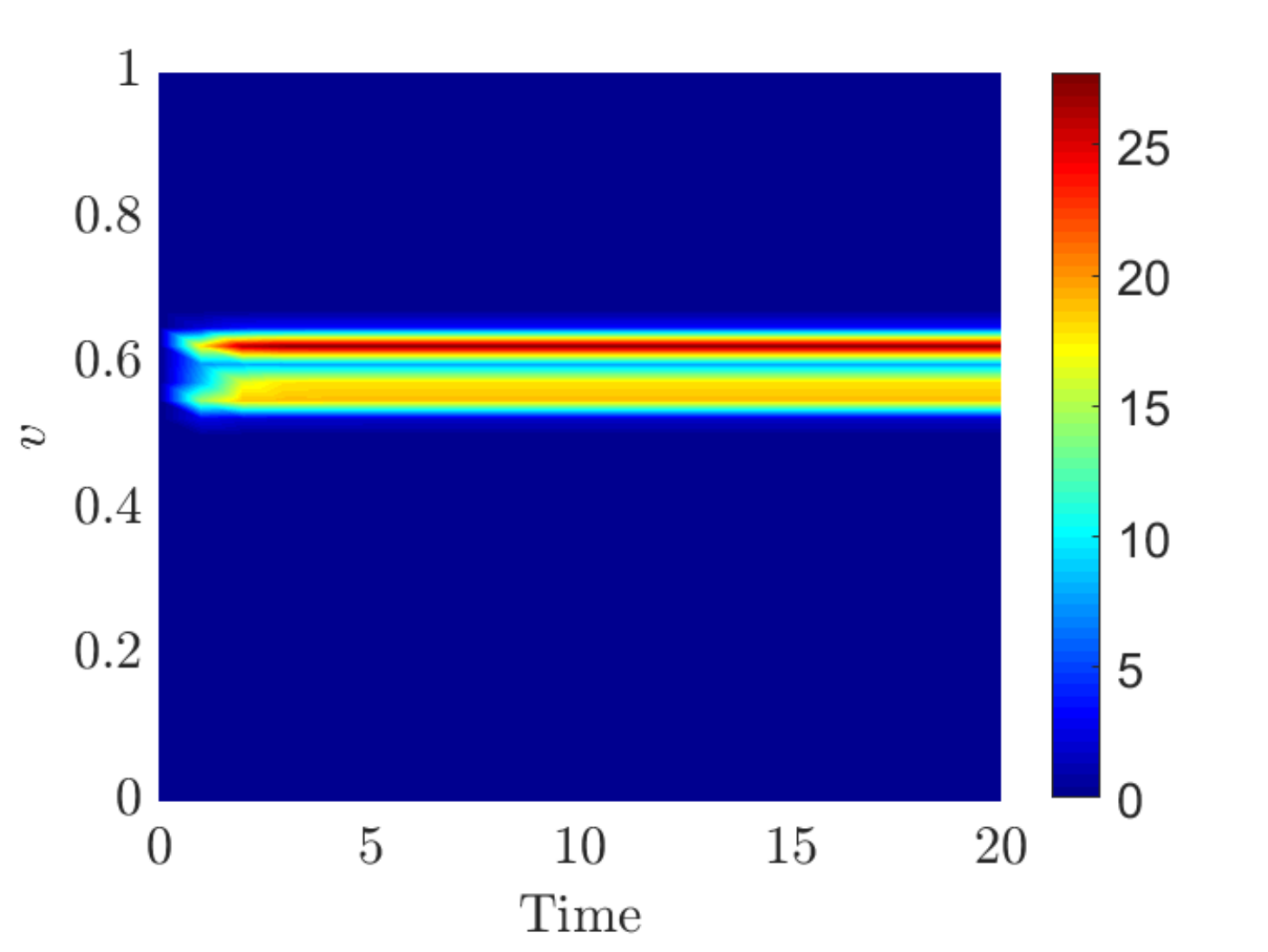}}
\subfigure[$p^\ast=10$, $\rho=0.6$]{\includegraphics[scale=0.3]{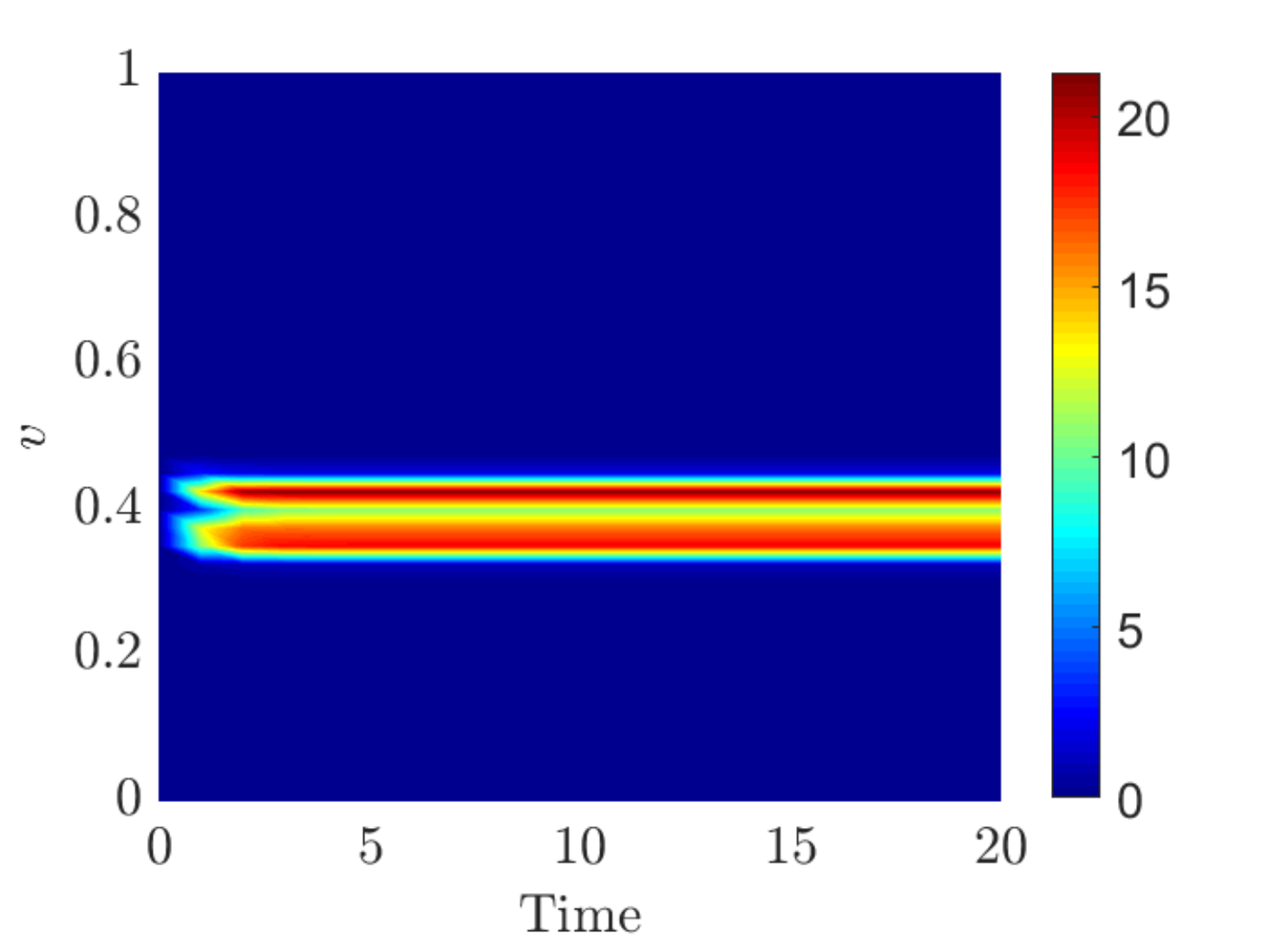}}
\caption{\textbf{Controlled case}, $\boldsymbol{z-1\sim\B\!\left(50,\,\frac{1}{50}\right)}$. Contours of $\bar{f}(t,\,v)=\E_z(f(t,\,v;\,z))$ (top row) and of $\Var_z(f(t,\,v;\,z))$ (bottom row), where $f$ is the solution to~\eqref{eq:FP.1}, when $z$ is such that $z-1$ has binomial distribution. All the parameters are like in Figure~\ref{fig:uni_cons}.}
\label{fig:bin_cons}
\end{figure}

With the aim of investigating the ability and the effectiveness of the control method to dampen the structural uncertainty of the model, in Figures~\ref{fig:uni_cons},~\ref{fig:bin_cons} we show the variance $\Var_z(f(t,\,v;\,z))$ of the solution $f$ to the Fokker-Planck equation~\eqref{eq:FP.1} for $\rho=0.2,\,0.4,\,0.6$ and for a decreasing penalisation $\kappa=10^{-1},\,10^{-2}$ of the control, which implies an effective penetration rate increasing from $p^\ast=1$ to $p^\ast=10$. We observe that the control produces in time a confinement of the $z$-variance of $g$ around the optimal speed $v_d(\rho)$. In particular, such a confinement is tighter and is reached more quickly with a high effective penetration rate $p^\ast$.

\begin{figure}[!t]
\centering
\subfigure[$z\sim\cU({[1,\,3]})$]{\includegraphics[scale=0.4]{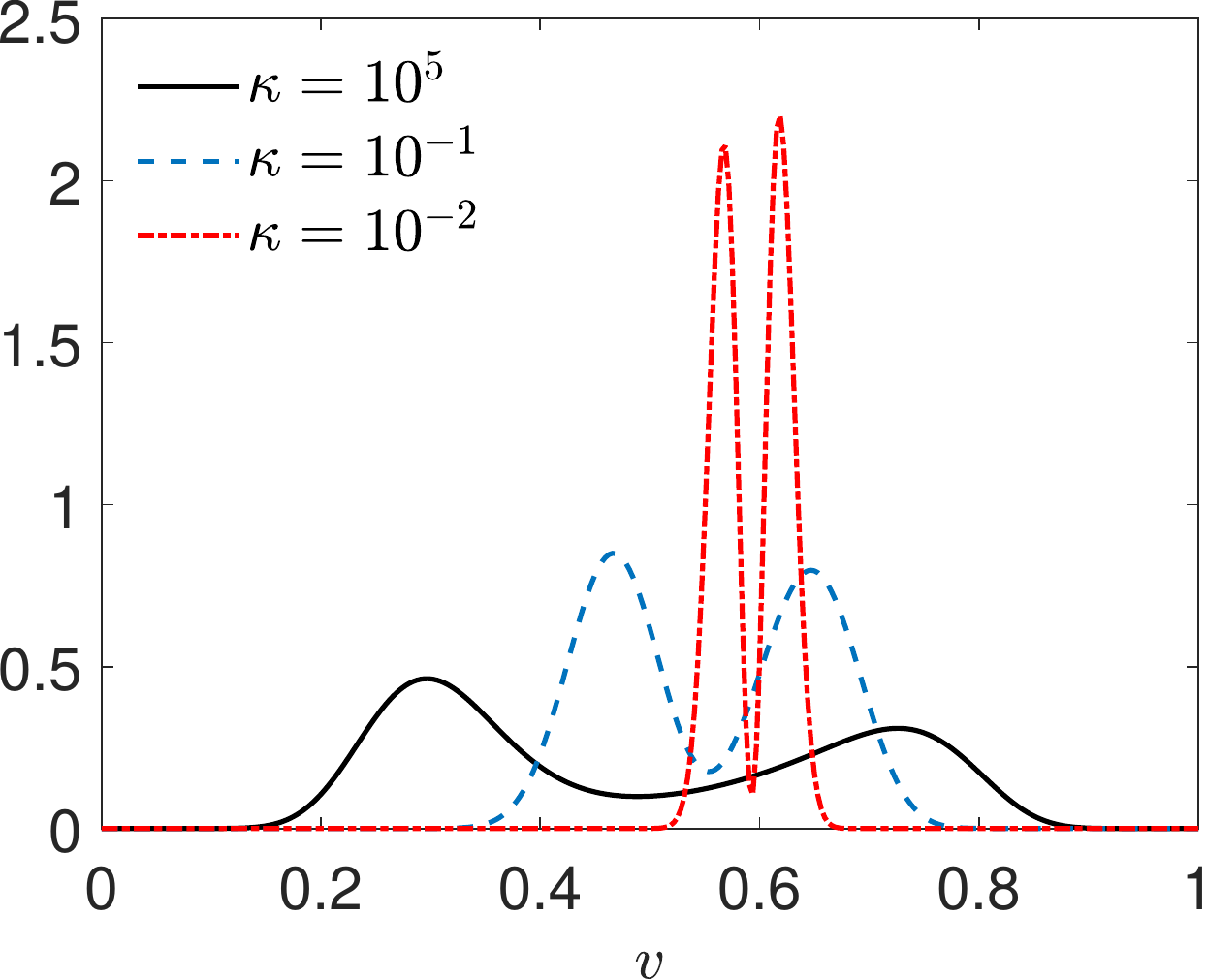}} \qquad
\subfigure[$z-1\sim\B\!\left(50,\,\frac{1}{50}\right)$]{\includegraphics[scale=0.4]{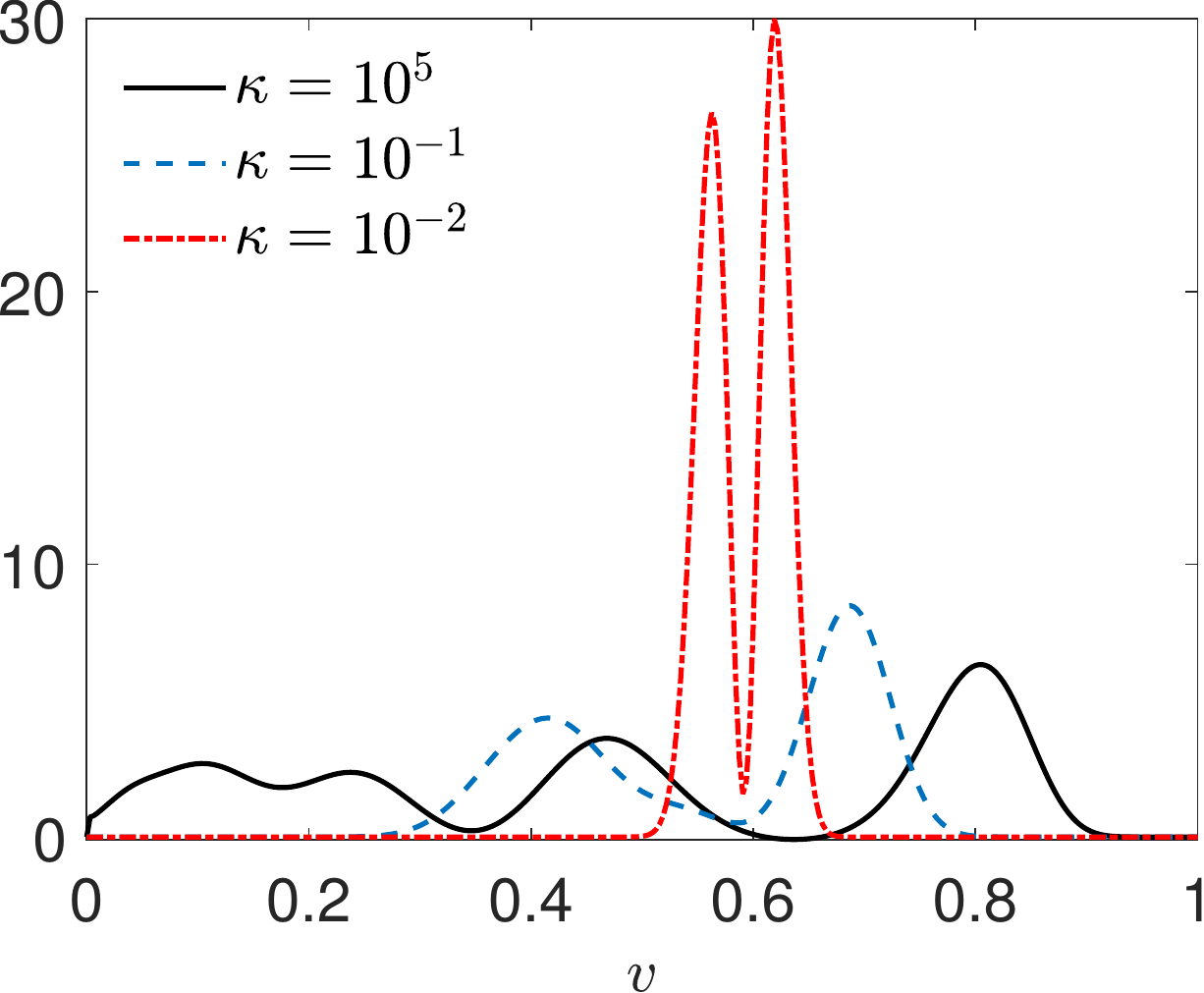}}
\caption{\textbf{Controlled case}. Asymptotic variance $\Var_z(f_\infty(v;\,z))$, where $f_\infty(v;\,z)$ is like in~\eqref{eq:finf.1}, obtained with a uniform and a binomial distribution of the uncertainty for a decreasing penalisation of the control: from $\kappa=10^5$, corresponding to a virtually uncontrolled setting, to $\kappa=10^{-1},\,10^{-2}$.}
\label{fig:comp_var}
\end{figure}

To further stress this fact, in Figure~\ref{fig:comp_var} we show a numerical approximation of the asymptotic variance $\Var_z(f_\infty(v;\,z))$, where $f_\infty(v;\,z)$ is given by~\eqref{eq:finf.1}, computed at a sufficiently large time ensuring that the equilibrium has been numerically reached. We fix $\rho=0.4$ and we consider a decreasing penalisation coefficient $\kappa=10^5,\,10^{-1},\,10^{-2}$. We see that the action of the control confines progressively the variability of $g_\infty$ around the optimal speed $v=v_d(\rho)$, namely around $v=0.6$ in the present case. At the  macroscopic level, this produces the uncertainty damping observed in the fundamental diagrams, cf. Figure~\ref{fig:damping_diag}.

\section{Conclusions}
\label{sect:conclusions}
In this paper, we have shown that the scattering of the fundamental diagram of traffic, usually observed from experimental data at high vehicle density, can be explained organically in terms of the propagation of some uncertainty from the microscopic scale of the interactions among the vehicles to the macroscopic scale of the hydrodynamical flow of the vehicles. Such an uncertainty may be associated with different types of vehicles present in a heterogeneous traffic stream. We have shown that the classical methods of the collisional kinetic theory, suitably coupled with arguments typical of the uncertainty quantification, provide a sound mathematical-physical background to investigate such a propagation across the scales and to design possible countermeasures to dampen it. In particular, we have suggested that these countermeasures may take advantage of automatic driver-assistance systems, namely of feedback controls implemented on a certain percentage of vehicles in the traffic stream. By means of the theoretical paradigm recalled above, we have studied the large scale aggregate impact of some simple microscopic control protocols. Specifically, we have proved that it is actually possible to reduce the macroscopic scattering of the fundamental diagram by inducing some vehicles in the stream to align locally their speed to a congestion-dependent recommended speed. We have proposed that this can be done in two alternative ways, namely by either a stochastic or a deterministic control, and we have proved that, in the long run, these two strategies are actually equivalent. The analysis of the two strategies has required to deal with two different types of kinetic equations: a simpler Boltzmann-type equation for Maxwellian-like particles in the case of the stochastic control; a more difficult Boltzmann-type equation with cut-off in the case of the deterministic control, which imposes a non-constant interaction kernel. In this respect, this paper gives also a contribution to the investigation of this second case, which is less standard, hence to date still under-explored, in the literature of the kinetic models of multi-agent systems.

The possibility to take advantage of the emerging technologies in the realm of vehicle automation to control the fundamental diagram, and in particular to reduce its scattering, is, in our opinion, a non-trivial fact with potential beneficial impacts on the governance of traffic at large scale. Furthermore, from the point of view of the mathematical research, we believe that this work paves the way for the study of macroscopic models of traffic flow with uncertainty and uncertainty control through driver-assist vehicles, cf.~\cite{delis2015CMA,delis2018TRR}, grounded on the multiscale framework of the kinetic theory.

\section*{Acknowledgements}
This research was partially supported  by the Italian Ministry for Education, University and Research (MIUR) through the ``Dipartimenti di Eccellenza'' Programme (2018-2022), Department of Mathematical Sciences ``G. L. Lagrange'', Politecnico di Torino (CUP: E11G18000350001) and Department of Mathematics ``F. Casorati'', University of Pavia; and through the PRIN 2017 project (No. 2017KKJP4X) ``Innovative numerical methods for evolutionary partial differential equations and applications''.

This work is also part of the activities of the Starting Grant ``Attracting Excellent Professors'' funded by ``Compagnia di San Paolo'' (Torino) and promoted by Politecnico di Torino.

Both authors are members of GNFM (Gruppo Nazionale per la Fisica Matematica) of INdAM (Istituto Nazionale di Alta Matematica), Italy.


\begin{thebibliography}{10}

\bibitem{Albi2015}
G.~Albi, M.~Herty, and L.~Pareschi.
\newblock Kinetic description of optimal control problems and applications to
  opinion consensus.
\newblock {\em Commun. Math. Sci.}, 6:1407--1429, 2015.

\bibitem{Albi2014}
G.~Albi, L.~Pareschi, and M.~Zanella.
\newblock Boltzmann-type control of opinion consensus through leaders.
\newblock {\em Philos. Trans. R. Soc. Lond. Ser. A Phys. End. Sci.},
  2028:20140138, 2014.

\bibitem{Aria2016}
E.~Aria, J.~Olstam, and C.~Schwietering.
\newblock Investigation of automated vehicle effects on driver's behavior and
  traffic performance.
\newblock {\em Transp. Res. Procedia}, 15:761--770, 2016.

\bibitem{benzoni2003EJAM}
S.~Benzoni-Gavage and R.~M. Colombo.
\newblock An $n$-populations model for traffic flow.
\newblock {\em European J. Appl. Math.}, 14(5):587--612, 2003.

\bibitem{Boscarino2016}
S.~Boscarino, F.~Filbet, and G.~Russo.
\newblock High order semi-implicit schemes for time dependent partial
  differential equations.
\newblock {\em J. Sci. Comput.}, 68:975--1001, 2016.

\bibitem{carrillo2019VJM}
J.~Carrillo and M.~Zanella.
\newblock Monte {C}arlo g{PC} methods for diffusive kinetic flocking models
  with uncertainties.
\newblock {\em Vitenam J. Math.}, 47(4):931--954, 2019.

\bibitem{Carrillo2019}
J.~A. Carrillo, L.~Pareschi, and M.~Zanella.
\newblock Particle based g{PC} methods for mean-field models of swarming with
  uncertainty.
\newblock {\em Commun. Comput. Phys.}, 25(2):508--531, 2019.

\bibitem{cercignani1994BOOK}
C.~Cercignani, R.~Illner, and M.~Pulvirenti.
\newblock {\em The Mathematical Theory of Dilute Gases}, volume 106 of {\em
  Applied Mathematical Sciences}.
\newblock Springer, 1994.

\bibitem{colombo2018CHAPTER}
R.~M. Colombo, C.~Klingenberg, and M.-C. Meltzer.
\newblock A multispecies traffic model based on the {L}ighthill-{W}hitham and
  {R}ichards model.
\newblock In C.~Klingenberg and M.~Westdickenberg, editors, {\em Theory,
  Numerics and Applications of Hyperbolic Problems I. HYP 2016}, volume 236 of
  {\em Springer Proceedings in Mathematics \& Statistics}, 2018.

\bibitem{cordier2005JSP}
S.~Cordier, L.~Pareschi, and G.~Toscani.
\newblock On a kinetic model for a simple market economy.
\newblock {\em J. Stat. Phys.}, 120(1):253--277, 2005.

\bibitem{delis2015CMA}
A.~I. Delis, I.~K. Nikolos, and M.~Papageorgiou.
\newblock Macroscopic traffic flow modeling with adaptive cruise control:
  {D}evelopment and numerical solution.
\newblock {\em Comput. Math. Appl.}, 70(8):1921--1947, 2015.

\bibitem{delis2018TRR}
A.~I. Delis, I.~K. Nikolos, and M.~Papageorgiou.
\newblock A macroscopic multi-lane traffic flow model for {ACC}/{CACC} traffic
  dynamics.
\newblock {\em Transp. Res. Record}, 2018.

\bibitem{Dimarco2017}
G.~Dimarco, L.~Pareschi, and M.~Zanella.
\newblock Uncertainty quantification for kinetic models in socio-economic and
  life sciences.
\newblock In S.~Jin and L.~Pareschi, editors, {\em Uncertainty quantification
  for Hyperbolic and Kinetic Equations}, volume~14 of {\em SEMA-SIMAI Springer
  Series}, pages 151--191. Springer, 2017.

\bibitem{freguglia2017CMS}
P.~Freguglia and A.~Tosin.
\newblock Proposal of a risk model for vehicular traffic: {A} {B}oltzmann-type
  kinetic approach.
\newblock {\em Commun. Math. Sci.}, 15(1):213--236, 2017.

\bibitem{helbing2001RMP}
D.~Helbing.
\newblock Traffic and related self-driven many-particle systems.
\newblock {\em Rev. Modern Phys.}, 73(4):1067--1141, 2001.

\bibitem{Herty2010}
M.~Herty and L.~Pareschi.
\newblock {F}okker-{P}lanck asymptotics for traffic flow models.
\newblock {\em Kinet. Relat. Mod.}, 3(1):165--179, 2010.

\bibitem{herty2018SIAP}
M.~Herty, A.~Tosin, G.~Visconti, and M.~Zanella.
\newblock Hybrid stochastic kinetic description of two-dimensional traffic
  dynamics.
\newblock {\em SIAM J. Appl. Math.}, 78(5):2737--2762, 2018.

\bibitem{herty2019PREPRINT}
M.~Herty, A.~Tosin, G.~Visconti, and M.~Zanella.
\newblock Reconstruction of traffic speed distributions from kinetic models
  with uncertainties.
\newblock Preprint (arXiv:1912.03706), 2019.

\bibitem{Hu2017}
J.~Hu and S.~Jin.
\newblock Uncertainty quantification for kinetic equations.
\newblock In S.~Jin and L.~Pareschi, editors, {\em Uncertainty Quantification
  for Hyperbolic and Kinetic Equations}, volume~14 of {\em SEMA-SIMAI Springer
  Series}, pages 193--229. Springer, 2017.

\bibitem{Jamson2013}
A.~H. Jamson, N.~Merat, O.~M.~J. Carsten, and F.~C.~H. Lai.
\newblock Behavioural changes in drivers experiencing highly-automated vehicle
  control in varying traffic conditions.
\newblock {\em Transport. Res. C}, 30:116--125, 2013.

\bibitem{Jin2017}
S.~Jin and L.~Pareschi, editors.
\newblock {\em Uncertainty quantification for hyperbolic and kinetic
  equations}, volume~14 of {\em SEMA-SIMAI Springer Series}.
\newblock Springer, 2017.

\bibitem{kerner2004BOOK}
B.~S. Kerner.
\newblock {\em The Physics of Traffic}.
\newblock Understanding Complex Systems. Springer, Berlin, 2004.

\bibitem{Klar1997}
A.~Klar and R.~Wegener.
\newblock {E}nskog-like models for vehicular traffic.
\newblock {\em J. Stat. Phys.}, 87(1--2):91--114, 1997.

\bibitem{marzouk2009CCP}
Y.~Marzouk and D.~Xiu.
\newblock A stochastic collocation approach to {B}ayesian inference in inverse
  problems.
\newblock {\em Commun. Comput. Phys.}, 6(4):826--847, 2009.

\bibitem{mason1997PRE}
A.~D. Mason and A.~W. Woods.
\newblock Car-following model of multispecies systems of road traffic.
\newblock {\em Phys. Rev. E}, 55(3):2203--2214, 1997.

\bibitem{maurya2016TRP}
A.~K. Maurya, S.~Das, S.~Dey, and S.~Nama.
\newblock Study on speed and time-headway distributions on two-lane
  bidirectional road in heterogeneous traffic condition.
\newblock {\em Transp. Res. Proc.}, 17:428--437, 2016.

\bibitem{nagatani2000PHYSA}
T.~Nagatani.
\newblock Traffic behavior in a mixture of different vehicles.
\newblock {\em Phys. A}, 284(1--4):405--420, 2000.

\bibitem{ni2018AMM}
D.~Ni, H.~K. Hsieh, and T.~Jiang.
\newblock Modeling phase diagrams as stochastic processes with application in
  vehicular traffic flow.
\newblock {\em Appl. Math. Model.}, 53:106--117, 2018.

\bibitem{ntousakis2015TRP}
I.~A. Ntousakis, I.~K. Nikolos, and M.~Papageorgiou.
\newblock On microscopic modelling of adaptive cruise control systems.
\newblock {\em Transp. Res. Proc.}, 6:111--127, 2015.

\bibitem{Pareschi2017}
L.~Pareschi and T.~Rey.
\newblock Residual equilibrium schemes for time dependent partial differential
  equations.
\newblock {\em Comput. Fluids}, 156(12):329--342, 2017.

\bibitem{pareschi2001ESAIMP}
L.~Pareschi and G.~Russo.
\newblock An introduction to {M}onte {C}arlo methods for the {B}oltzmann
  equation.
\newblock {\em ESAIM: Proc.}, 10:35--75, 2001.

\bibitem{pareschi2013BOOK}
L.~Pareschi and G.~Toscani.
\newblock {\em Interacting {M}ultiagent {S}ystems: {K}inetic equations and
  {M}onte {C}arlo methods}.
\newblock Oxford University Press, 2013.

\bibitem{Pareschi2018}
L.~Pareschi and M.~Zanella.
\newblock Structure preserving schemes for nonlinear {F}okker-{P}lanck
  equations and applications.
\newblock {\em J. Sci. Comput.}, 74:1575--1600, 2018.

\bibitem{paveri1975TR}
S.~L. Paveri-Fontana.
\newblock On {B}oltzmann-like treatments for traffic flow: a critical review of
  the basic model and an alternative proposal for dilute traffic analysis.
\newblock {\em Transportation Res.}, 9(4):225--235, 1975.

\bibitem{piccoli2019PREPRINT}
B.~Piccoli, A.~Tosin, and M.~Zanella.
\newblock Model-based assessment of the impact of driver-assist vehicles using
  kinetic theory.
\newblock Preprint (arXiv:1911.04911), 2019.

\bibitem{prigogine1971BOOK}
I.~Prigogine and R.~Herman.
\newblock {\em Kinetic theory of vehicular traffic}.
\newblock American Elsevier Publishing Co., New York, 1971.

\bibitem{puppo2016CMS}
G.~Puppo, M.~Semplice, A.~Tosin, and G.~Visconti.
\newblock Fundamental diagrams in traffic flow: the case of heterogeneous
  kinetic models.
\newblock {\em Commun. Math. Sci.}, 14(3):643--669, 2016.

\bibitem{puppo2017CMS}
G.~Puppo, M.~Semplice, A.~Tosin, and G.~Visconti.
\newblock Analysis of a multi-population kinetic model for traffic flow.
\newblock {\em Commun. Math. Sci.}, 15(2):379--412, 2017.

\bibitem{seibold2013NHM}
B.~Seibold, M.~R. Flynn, A.~R. Kasimov, and R.~R. Rosales.
\newblock Constructing set-valued fundamental diagrams from jamiton solutions
  in second order traffic models.
\newblock {\em Netw. Heterog. Media}, 8(3):745--772, 2013.

\bibitem{stern2018TRC}
R.~E. Stern, S.~Cui, M.~L. Delle~Monache, R.~Bhadani, M.~Bunting, M.~Churchill,
  N.~Hamilton, R.~Haulcy, H.~Pohlmann, F.~Wu, B.~Piccoli, B.~Seibold,
  J.~Sprinkle, and D.~B. Work.
\newblock Dissipation of stop-and-go waves via control of autonomous vehicles:
  {F}ield experiments.
\newblock {\em Transportation Res. Part C}, 89:205--221, 2018.

\bibitem{toscani2006CMS}
G.~Toscani.
\newblock Kinetic models of opinion formation.
\newblock {\em Commun. Math. Sci.}, 4(3):481--496, 2006.

\bibitem{tosin2018CMS}
A.~Tosin and M.~Zanella.
\newblock Boltzmann-type models with uncertain binary interactions.
\newblock {\em Commun. Math. Sci.}, 16(4):963--985, 2018.

\bibitem{tosin2018IFAC}
A.~Tosin and M.~Zanella.
\newblock Control strategies for road risk mitigation in kinetic traffic
  modelling.
\newblock {\em IFAC-PapersOnLine}, 51(9):67--72, 2018.

\bibitem{tosin2019MMS}
A.~Tosin and M.~Zanella.
\newblock Kinetic-controlled hydrodynamics for traffic models with
  driver-assist vehicles.
\newblock {\em Multiscale Model. Simul.}, 17(2):716--749, 2019.

\bibitem{villani1998PhD}
C.~Villani.
\newblock {\em Contribution \`{a} l'\'{e}tude math\'{e}matique des
  \'{e}quations de Boltzmann et de Landau en th\'{e}orie cin\'{e}tique des gaz
  et des plasmas}.
\newblock Ph{D} thesis, Paris 9, 1998.

\bibitem{villani1998ARMA}
C.~Villani.
\newblock On a new class of weak solutions to the spatially homogeneous
  {B}oltzmann and {L}andau equations.
\newblock {\em Arch. Ration. Mech. Anal.}, 143(3):273--307, 1998.

\bibitem{Xiu2010}
D.~Xiu.
\newblock {\em Numerical Methods for Stochastic Computations}.
\newblock Princeton University Press, 2010.

\bibitem{Xiu2002}
D.~Xiu and G.~E. Karniadakis.
\newblock The {W}iener-{A}skey polynomial chaos for stochastic differential
  equations.
\newblock {\em SIAM J. Sci. Comput.}, 24(2):614--644, 2002.

\bibitem{zanella2020MCS}
M.~Zanella.
\newblock Structure preserving stochastic {G}alerkin methods for
  {F}okker-{P}lanck equations with background interactions.
\newblock {\em Math. Comput. Simulation}, 168:28--47, 2020.

\bibitem{Zhu2017}
Y.~Zhu and S.~Jin.
\newblock The {V}lasov-{P}oisson-{F}okker-{P}lanck system with uncertainty and
  a one-dimensional asymptotic-preserving method.
\newblock {\em Multiscale Model. Simul.}, 15(4):1502--1529, 2017.

\end{thebibliography}

\end{document}